\documentclass[11pt]{article}%

\usepackage{color}

\usepackage{graphicx, stmaryrd}

\usepackage{amsmath}
\usepackage{amsfonts}
\usepackage{amssymb}
\usepackage{graphicx}
\usepackage{fullpage}%
\providecommand{\U}[1]{\protect\rule{.1in}{.1in}}
\newtheorem{theorem}{Theorem}

\newtheorem{claim}[theorem]{Claim}

\newtheorem{corollary}[theorem]{Corollary}

\newtheorem{definition}[theorem]{Definition}

\newtheorem{lemma}[theorem]{Lemma}

\newenvironment{proof}[1][Proof]{\noindent\textbf{#1.} }{\ \rule{0.5em}{0.5em} \vspace{1 em}}

\newcommand{\cl}{\mathrm{clk}}

\newcommand{\size}{\mathrm{size}}

\newcommand{\intt}{\mathrm{int}}
\newcommand{\ext}{\mathrm{ext}}

\newcommand{\inn}{\mathrm{in}}
\newcommand{\prop}{\mathrm{prop}}

\newcommand{\hin}{H_{\inn}}
\newcommand{\hout}{H_{\out}}
\newcommand{\hprop}{H_{\prop}}

\newcommand{\ket}[1]{|#1\rangle}
\newcommand{\bra}[1]{\langle#1|}
\newcommand{\ketbra}[2]{|#1\rangle\langle#2|}
\newcommand{\braket}[2]{\langle#1|#2\rangle}

\newcommand{\ol}[1]{\overline{#1}}
\newcommand{\mc}[1]{\mathcal{#1}}

\newcommand{\suc}{\mathrm{suc}}
\newcommand{\targ}{\mathrm{targ}}

\newcommand{\comp}{\mathrm{comp}}

\newcommand{\anc}{\mathrm{anc}}
\newcommand{\els}{\mathrm{else}}
\newcommand{\eff}{\mathrm{eff}}

\newcommand{\adv}{\mathrm{adv}}

\newcommand{\early}{\mathrm{early}}
\newcommand{\late}{\mathrm{late}}

\newcommand{\out}{\mathrm{out}}
\newcommand{\bad}{\mathrm{bad}}

\newcommand{\poly}{\mathrm{poly}}
\newcommand{\avg}{\mathrm{avg}}

\newcommand{\eps}{\varepsilon}

\newcommand{\bE}{\mathbb{E}}

\newcommand{\bR}{\mathbb{R}}
\newcommand{\bC}{\mathbb{C}}

\newcommand{\ot}{\otimes}

\DeclareMathOperator{\val}{val}
\DeclareMathOperator{\inval}{inval}

\begin{document}

\title{A Full Characterization of Quantum Advice\footnote{A preliminary extended abstract of this work appeared in ACM STOC 2010.}}
\author{Scott Aaronson\thanks{Email: aaronson@csail.mit.edu. \ \ This material is
based upon work supported by the National Science Foundation under Grant No.
0844626. \ Also supported by a DARPA YFA grant, a Sloan Fellowship, and the
Keck Foundation.}\\MIT
\and Andrew Drucker\thanks{Email: andy.drucker@gmail.com. \ \ Supported by the NSF under agreements Princeton University Prime Award No. CCF-0832797 and Sub-contract No. 00001583. This material is based upon work supported by an Akamai Presidential Graduate Fellowship while the author was a graduate student at MIT.}\\IAS}
\date{}
\maketitle

\begin{abstract}
We prove the following surprising result: given any quantum state $\rho$\ on
$n$ qubits, there exists a local Hamiltonian $H$ on $\operatorname*{poly}%
\left(  n\right)  $\ qubits (e.g., a sum of two-qubit interactions), such that
any ground state of $H$ can be used to simulate $\rho$\ on all quantum
circuits of fixed polynomial size. \ In terms of complexity classes, this
implies that $\mathsf{BQP/qpoly}\subseteq\mathsf{QMA/poly}$, which supersedes
the previous result of Aaronson that $\mathsf{BQP/qpoly}\subseteq
\mathsf{PP/poly}$. \ Indeed, we can exactly characterize quantum advice, as
equivalent in power to \textit{untrusted} quantum advice combined with trusted
\textit{classical} advice.

Proving our main result requires combining a large number of previous
tools---including a result of Alon et al.\ on learning of real-valued concept
classes, a result of Aaronson on the learnability of quantum states, and a
result of Aharonov and Regev on ``$\mathsf{QMA}_{\mathsf{+}}$
super-verifiers''---and also creating some new ones. \ The main new tool is a
so-called \textit{majority-certificates lemma}, which is closely related to
boosting in machine learning, and which seems likely to find independent
applications. \ In its simplest version, this lemma says the following.
\ Given any set $S$ of Boolean functions on $n$ variables, any function $f\in
S$ can be expressed as the pointwise majority of $m=O\left(  n\right)  $
functions $f_{1},\ldots,f_{m}\in S$, such that each $f_{i}$ is the unique
function in $S$ compatible with $O\left(  \log\left\vert S\right\vert \right)
$\ input/output constraints.

\end{abstract}
\tableofcontents

\section{Introduction\label{INTRO}}

\begin{quotation}
\noindent\textit{How much classical information is needed to specify a quantum
state\ of} $n$\textit{ qubits?}
\end{quotation}

This question has inspired a rich and varied set of responses, in part because
it can be interpreted in many ways. \ If we want to specify a quantum state
$\rho$ \textit{exactly}, then of course the answer is ``an infinite
amount,''\ since amplitudes in quantum mechanics are continuous. \ A natural
compromise is to try to specify $\rho$ \textit{approximately}, i.e., to give a
description which yields a state $\widetilde{\rho}$ whose statistical behavior
is close to that of $\rho$ under every measurement. (This statement is
captured by the requirement that $\rho$ and $\widetilde{\rho}$ are close under
the so-called \textit{trace distance} metric.) \ But it is not hard to see
that even for this task, we still need to use an exponential (in $n$) number
of classical bits.

This fact can be viewed as a disappointment, but also as an opportunity, since
it raises the prospect that we might be able to encode massive amounts of
information in physically compact quantum states: for example, we might hope
to store $2^{n}$ classical bits in $n$ qubits. \ But an obvious practical
requirement is that we be able to retrieve the information reliably, and this
rules out the hope of significant ``quantum compression''\ of classical strings,
as shown by a landmark result of Holevo \cite{holevo} from 1973. \ Consider a
sender Alice and a recipient Bob, with a one-way quantum channel between them.
\ Then Holevo's Theorem says that, if Alice wants to encode an $n$-bit
classical string $x$\ into an $m$-qubit quantum state $\rho_{x}$, in such a
way that Bob can retrieve $x$ (with probability $2/3$, say) by measuring
$\rho_{x}$, then Alice must take $m\geq n-O\left(  1\right)  $ (or $m\geq
n/2-O\left(  1\right)  $,\ if Alice and Bob share entanglement). \ In other
words, for this communication task, quantum states offer essentially no
advantage over classical strings. In 1999, Nayak~\cite{nayak}, improving on
Ambainis et al.~\cite{antv_conf} (see~\cite{antv}), generalized Holevo's
result as follows: even if Bob wants to learn only a \textit{single bit}
$x_{i}$\ of $x=x_{1}\ldots x_{n}$\ (for some $i\in\left[  n\right]  $\ unknown
to Alice), and is willing to destroy the state $\rho_{x}$ in the process of
learning that bit, Alice still needs to send $m=\Omega\left(  n\right)  $
qubits for Bob to succeed with high probability.

These results say that the exponential descriptive complexity of quantum
states cannot be effectively harnessed for classical data storage, but they do
not bound the number of practically meaningful ``degrees of freedom''\ in a
quantum state used for purposes other than storing data. \ For example, a
quantum state could be useful for computation, or it could be a physical
system worthy of study in its own right. \ The question then becomes, what
useful information \textit{can} we give about an $n$-qubit state using a
``reasonable'' number (say, $\operatorname*{poly}\left(  n\right)  $) of
classical bits?

One approach to this question is to identify special subclasses of quantum
states for which a faithful approximation can be specified using only
$\operatorname*{poly}\left(  n\right)  $ bits. \ This has been done, for
example, with matrix product states \cite{vidal} and ``tree
states''\ \cite{aar:mlin}. \ A second approach is to try to describe an
\textit{arbitrary} $n$-qubit state $\rho$ concisely, in such a way that the
state $\widetilde{\rho}$ recovered from the description is close to $\rho$
with respect to some natural subclass of \textit{measurements}. \ This has
been done for specific classes like the ``pretty good measurements''\ of
Hausladen and Wootters \cite{hausladen}. \ A more ambitious goal in this vein,
explored by Aaronson in two previous works \cite{aar:adv, aar:learn} and
continued in the present paper, is to give a description of an $n$-qubit state
$\rho$ which yields a state $\widetilde{\rho}$ that behaves approximately like
$\rho$ with respect to all (binary) measurements performable by quantum
circuits of ``reasonable'' size---say, of size at most $n^{c}$, for some fixed
$c > 0$. Then if $c$ is taken large enough, $\widetilde{\rho}$ is arguably
``just as good'' as $\rho$ for practical purposes.

Certainly we can achieve this goal using $2^{n^{c + O(1)}}$\ bits: simply give
approximations to the measurement statistics for every size-$n^{c}$ circuit.
\ However, the results of Holevo \cite{holevo}\ and Ambainis et
al.\ \cite{antv}\ suggest that a much more succinct description might be
possible. \ This hope was realized by Aaronson \cite{aar:adv}, who gave a
description scheme in which an $n$-qubit state can be specified using
$\operatorname*{poly}\left( n\right) $ classical bits. There is a significant
catch in Aaronson's result, though: the encoder Alice and decoder Bob both
need to invest exponential amounts of computation.

In a subsequent paper \cite{aar:learn}, Aaronson gave a closely-related result
which significantly reduces the computational requirements:\ now Alice can
generate her message in polynomial time (for fixed $c$). \ Also, while Bob
cannot necessarily construct the state $\widetilde{\rho}$ efficiently on his
own, if he is presented with such a state (by an untrusted prover, say), Bob
can \textit{verify} the state in polynomial time.\ \ The catch in this result
is a weakened approximation guarantee: Bob cannot use $\widetilde{\rho}$ to
predict the outcomes of \textit{all} the measurements defined by size-$n^{c}$
circuits, but only \textit{most} of them (with respect to a samplable
distribution used by Alice in the encoding process). \ Aaronson conjectured \cite{aar:learn} that the tradeoff between the results of \cite{aar:learn} and of \cite{aar:adv} revealed an inherent limit to quantum compression.

\subsection{Our Quantum Information Result}

The main result of this paper is that Aaronson's conjecture was false: one
really can get the best of both worlds, and simulate an arbitrary quantum
state $\rho$\ on all small circuits, using a different state $\widetilde{\rho
}$ that is easy to recognize. \ Indeed, we can even take $\widetilde{\rho}%
$\ to be the \textit{ground state of a local Hamiltonian}: that is, a pure
state $\widetilde{\rho}=\left\vert \psi\right\rangle \left\langle
\psi\right\vert $\ on $\operatorname*{poly}\left(  n\right)  $\ qubits minimizing the disagreement with $\operatorname*{poly}\left(  n\right)  $\ local constraints,
each involving a constant number of qubits. \ In a sense, then, this paper
completes a ``trilogy''\ of which \cite{aar:adv,aar:learn}\ were the first two installments.

Here is a formal statement of our result.

\begin{theorem}
\label{mainthm0} Let $c, \delta>0$, and let $\rho^*$\ be any $n$-qubit
quantum state. \ Then there exists a $2$-local Hamiltonian $H$ on
$\operatorname*{poly}\left( n, \frac{1}{\delta} \right)  $\ qubits, and a
transformation $C\longrightarrow C^{\prime}$\ of quantum circuits, computable
in time $\operatorname*{poly}\left(  n, 1 / \delta\right) $ given $H$,
such that the following holds: for any ground state $\ket{\psi}$ of $H$, and for any measurement
$C$ definable by a quantum circuit of size $n^{c}$, we have $\left\vert  \bE\left[
C^{\prime}\left( \ketbra{\psi}{\psi} \right)\right]    - \bE\left[C\left(  \rho^*\right) \right]  \right\vert \leq\delta$.
\end{theorem}

In other words, the ground states of local Hamiltonians are ``universal quantum
states''\ in a very non-obvious sense. \ For example, suppose you own a quantum
software store,\ which sells quantum states $\rho$\ that can be fed as input
to quantum computers. \ Then our result says that \textit{ground states of
local Hamiltonians are the only kind of state you ever need to stock}. \ What
makes this surprising is that being a good piece of quantum software\ might
entail satisfying an exponential number of constraints: for example, if $\rho
$\ is supposed to help a customer's quantum computer $Q$\ evaluate some
Boolean function $f:\left\{  0,1\right\}  ^{n}\rightarrow\left\{  0,1\right\}
$, then $Q\left(  \rho, x \right)  $\ should output $f\left(  x\right)  $\ for
\textit{every} input $x\in\left\{  0,1\right\}  ^{n}$. \ By contrast, any
$k$-local Hamiltonian $H$ can be described as a set of at most ${\binom{n }%
{k}} = O(n^{k})$ constraints.

One can also interpret Theorem \ref{mainthm0}\ as a statement about
communication over quantum channels. \ Suppose Alice (who is computationally
unbounded) has a classical description of an $n$-qubit state $\rho^*$. \ She
would like to describe this state to Bob (who is computationally bounded), at
least well enough for Bob to be able to \textit{simulate} $\rho^*$\ on all
quantum circuits of some fixed polynomial size. \ However, Alice cannot just
send $\rho^*$ to Bob, since her quantum communication channel is noisy and
there is a chance that the state might get corrupted along the way. \ Nor can
she send a faithful classical description of $\rho^* $,
since that would require an exponential number of bits.\ Our result provides
an alternative: Alice can send a different quantum state $\sigma$, of
$\operatorname*{poly}(n)$\ qubits, together with a $\operatorname*{poly}%
(n)$-bit classical string $x$. Then, Bob can use $x$ to \textit{verify} that
$\sigma$\ can be used to accurately simulate $\rho^*$\ on all small measurements.

We believe Theorem \ref{mainthm0} makes a significant contribution to the
study of the effective information content of quantum states. \ It does,
however, leave open whether a quantum state of $n$ qubits can be efficiently
encoded \textit{and} decoded in polynomial time, in a way that is ``good
enough'' to preserve the measurement statistics of measurements defined by
circuits of fixed polynomial size. \ This remains an important problem for
future work.

\subsection{Impact on Quantum Complexity Theory\label{MOTIVATION}}

The questions addressed in this paper, and our results, are naturally phrased
and proved in terms of complexity classes. \ In recent years, researchers have
defined quantum complexity classes as a way to study the ``useful information''
embodied in quantum states. One approach is to study the power of nonuniform
\textit{quantum advice}. \ The class $\mathsf{BQP/qpoly}$, defined by
Nishimura and Yamakami \cite{ny}, consists of all languages decidable in
polynomial time by a quantum computer, with the help of a
$\operatorname*{poly}\left(  n\right)  $-qubit advice state that depends only
on the input length $n$. This class is analogous to the classical class
$\mathsf{P/poly}$. To understand the role of quantum information in
determining the power of $\mathsf{BQP/qpoly}$, a useful benchmark of
comparison is the class $\mathsf{BQP/poly}$ of decision problems efficiently
solvable by a quantum algorithm with $\operatorname*{poly}\left(  n\right)  $
bits of \emph{classical} advice (or equivalently, by a non-uniform family of $\operatorname*{poly}\left(  n\right)  $-sized quantum circuits). It is open whether $\mathsf{BQP/qpoly} =
\mathsf{BQP/poly}$.

A second approach studies the power of quantum \textit{proof systems}, by
analogy with the classical class $\mathsf{NP}$. \ Kitaev (unpublished, 1999)
defined the complexity class now called $\mathsf{QMA}$, for ``Quantum
Merlin-Arthur.'' This is the class of decision problems for which a ``yes''
answer can be proved by exhibiting a \textit{quantum witness state} (or
\textit{quantum proof}) $\left\vert \psi\right\rangle $, on
$\operatorname*{poly}\left(  n\right)  $\ qubits, which is then checked by a
skeptical polynomial-time quantum verifier. A useful benchmark class is
$\mathsf{QCMA}$\ (for ``Quantum Classical Merlin-Arthur''), defined by Aharonov
and Naveh \cite{an}. This is the class of decision problems for which a ``yes''
answer can be checked by a \textit{quantum} verifier who receives a
\textit{classical} witness. \ Here the natural open question is whether
$\mathsf{QMA} = \mathsf{QCMA}$.

In this paper we prove a new upper bound on $\mathsf{BQP/qpoly}$:

\begin{theorem}
\label{mainthm1}$\mathsf{BQP/qpoly}\subseteq\mathsf{QMA/poly}$.
\end{theorem}

Previously Aaronson showed in \cite{aar:adv} that $\mathsf{BQP/qpoly}
\subseteq\mathsf{PP/poly}$, and showed in \cite{aar:learn} that
$\mathsf{BQP/qpoly}$ is contained in the ``heuristic'' class
$\mathsf{HeurQMA/poly}$; Theorem \ref{mainthm1} supersedes both of these
earlier results.

Theorem \ref{mainthm1} says that one can always replace polynomial-size
quantum advice by polynomial-size \textit{classical} advice, together with a
polynomial-size untrusted quantum witness. \ Indeed, we can \textit{characterize} the class
$\mathsf{BQP/qpoly}$, as equal to the subclass of $\mathsf{QMA/poly}$ in which
the quantum witness state $\left\vert \psi_{n}\right\rangle $\ can only depend
on the input length $n$.\footnote{We call this restricted class
$\mathsf{YQP/poly}$.  Its definition is closely related to the earlier notion of \emph{input-oblivious nondeterminism}; this concept was used to define several other complexity classes in works of Chakravarthy and Roy~\cite{chakaravarthy} and Fortnow, Santhanam, and Williams~\cite{fsw}.  We have made a significant alteration to the definition of $\mathsf{YQP/poly}$ from prior versions of this work, as discussed in Section~\ref{ss:changes}.}

Using Theorem \ref{mainthm1}, we also obtain several other results for quantum
complexity theory:

\begin{enumerate}
\item[(1)] Without loss of generality, every quantum advice state can be taken
to be the ground state of some local Hamiltonian $H$. \ In essence, this result follows by combining our $\mathsf{BQP/qpoly}\subseteq\mathsf{QMA/poly}%
$\ result with the result of Kitaev~\cite{ksv} that \textsc{Local Hamiltonians} is
$\mathsf{QMA}$-complete.  The proof, however, requires a close analysis of the structure of low-energy states of the Hamiltonian $H$ in Kitaev's 5-local reduction (not proved or needed in~\cite{ksv}).  To show that the locality of $H$ can be reduced to 2, we use gadgets and a perturbation-theoretic result of Oliveira and Terhal~\cite{OT}, which built on Kempe, Kitaev and Regev's original proof of the $\mathsf{QMA}$-completeness of \textsc{2-Local Hamiltonians}~\cite{kkr}.\footnote{Related results appear in~\cite{JF}, although these seem not to give what we need.}  

\item[(2)] It is open whether for every local Hamiltonian $H$ on $n$ qubits,
there exists a quantum circuit of size $\operatorname*{poly}\left(  n\right)
$ that prepares a ground state of $H$. It is easy to show that an affirmative
answer would imply $\mathsf{QMA}=\mathsf{QCMA}$. As a consequence of Theorem
\ref{mainthm1}, we can show that an affirmative answer would also imply
$\mathsf{BQP/qpoly}=\mathsf{BQP/poly}$---thereby establishing a
previously-unknown connection between quantum proofs and quantum advice.

\item[(3)] We generalize Theorem \ref{mainthm1} to show that
$\mathsf{QCMA/qpoly}\subseteq\mathsf{QMA/poly}$.

\item[(4)] We use our new characterization of $\mathsf{BQP/qpoly}$\ to prove a
quantum analogue of the Karp-Lipton Theorem \cite{kl}. \ Recall that the
Karp-Lipton Theorem says that if $\mathsf{NP}\subset\mathsf{P/poly}$, then the
polynomial hierarchy collapses to the second level. \ Our ``Quantum Karp-Lipton
Theorem''\ says that if $\mathsf{NP}\subset\mathsf{BQP/qpoly}$ (that is,
$\mathsf{NP}$-complete problems are efficiently solvable with the help of
quantum advice), then $\mathsf{\Pi}_{\mathsf{2}}^{\mathsf{P}}\subseteq
\mathsf{QMA}^{\mathsf{P{}romiseQMA}}$. \ As far as we know, this\ is the first
nontrivial result to derive unlikely consequences from a hypothesis about
quantum machines being able to solve $\mathsf{NP}$-complete problems in
polynomial time.
\end{enumerate}

Finally, using our result, we are able to provide an illuminating perspective on a 2000 paper of Watrous \cite{watrous}. \ Watrous gave a
simple example of a ``purely-classical'' problem in $\mathsf{QMA}$\ that is not
\textit{obviously} in $\mathsf{QCMA}$---that is, for which quantum proofs
actually seem to help.\footnote{Aaronson and Kuperberg \cite{ak}, however,
give evidence that this problem might be in $\mathsf{QCMA}$, under conjectures
related to the Classification of Finite Simple Groups.} \ This problem is
called \textsc{Group Non-Membership}, and is defined as follows: Arthur is
given a finite \textit{black-box group} $G$ and a subgroup $H\leq
G$\ (specified by their generators), as well as an element $x\in G$. \ His
task is to verify that $x\notin H$. \ It is known that, as a black-box
problem, this problem is not in $\mathsf{MA}$. \ But Watrous showed that
\textsc{Group Non-Membership} is in $\mathsf{QMA}$, by a protocol in which
Merlin is ``expected'' to send the following quantum proof:
\[
\left\vert H\right\rangle =\frac{1}{\sqrt{\left\vert H\right\vert }}\sum_{h\in
H}\left\vert h\right\rangle .
\]
Arthur's verification procedure consists of two tests. \ In the first test,
Arthur \textit{assumes} that Merlin sent $\left\vert H\right\rangle $, and
then uses $\left\vert H\right\rangle $\ to decide whether $x\in H$. \ The test
is a simple, beautiful illustration of the power of quantum algorithms. \ The
second test in Watrous's protocol confirms that Merlin really sent $\left\vert
H\right\rangle $\ , or at least, a state which is ``equivalent'' for purposes of
the first test. \ This second test and its analysis are considerably more
involved, and seem less ``natural.''

Using our results, we see that a slightly weaker version of Watrous's result
can be derived in an almost automatic way from his first test, as follows. If
we assume that the black-box group $H = H_{n}$ is fixed for each input length,
then \textsc{Group Non-Membership} is in $\mathsf{BQP/qpoly}$, by letting
$\left\vert H_{n}\right\rangle $ as above be the trusted advice for length $n$
and using Watrous's first test as the $\mathsf{BQP/qpoly}$ algorithm. Then
Theorem \ref{mainthm1} (which can be readily adapted to the black-box setting)
tells us that Group Non-Membership is in $\mathsf{QMA/poly}$ as well.

\subsection{Changes to the Paper}\label{ss:changes}

We have corrected some significant issues with previous drafts.  First, the definition of so-called $\mathsf{YQP}$ machines needed to be amended to correct a deficiency in the previous definition, that prevented completeness- and soundness-amplification techniques from working as claimed.  This change appears necessary to preserve the claim $\mathsf{BQP/qpoly} = \mathsf{YQP/poly}$.  The revised definition of $\mathsf{YQP/poly}$ is actually more natural, and has the same intuitive interpretation: now as before, a $\mathsf{YQP/poly}$ machine receives trusted classical advice plus untrusted quantum advice, each determined solely by the input length, and applies two computations---a first which \emph{tests} the quantum advice $\rho$ by some measurement process, and a second which \emph{uses} $\rho$ to compute to decide membership of an input $x$ in some language $L$.

The necessary change is that, rather than testing one copy of $\rho$ and separately using another copy for the computation (an unnatural scenario, due to the No-Cloning Theorem of quantum mechanics), a $ \mathsf{YQP/poly}$ algorithm first tests $\rho$, then uses the \emph{modified}, post-measurement state $\rho'$ for computing $L(x)$.  
 The revised correctness requirement is that, for any quantum advice $\rho$ which has a noticeable chance of passing the test, the post-test state $\rho'$ is useful for computation, \emph{conditioned} on passing the test.  Section~\ref{QCC} contains relevant definitions and further discussion.

The second significant issue we have addressed (pointed out to us by a journal referee) is that the analysis of Local-Hamiltonian reductions for $\mathsf{QMA}$ in~\cite{ksv, kkr} does not immediately supply enough information about the structure of ground states to prove Theorem~\ref{mainthm0}.  In particular, ground states of the Hamiltonians produced need not be ``history states'' encoding $\mathsf{QMA}$ verifier computations in the intended format, as we had erroneously claimed.  

In the present version, we instead prove some properties of existing Local-Hamiltonian reductions that suffice for our original application.  First, we show that when Kitaev's reduction~\cite{kkr} is applied to a $\mathsf{QMA}$ verifier $V$ which accepts some proof state with probability close to 1, the resulting 5-local Hamiltonian $H_V$ is such that any nearly-minimal-energy state\footnote{Here, the \emph{energy} of a pure state $\ket{\psi}$ with respect to Hamiltonian $H$ is defined as $\bra{\psi}H\ket{\psi}$, and the minimal-energy states are precisely the ground states.} $\ket{\psi}$ is close (in trace distance) to a history state, and can be used to efficiently obtain a proof state accepted with high probability by $V$.  Next, we show that the reductions of Oliveira and Terhal~\cite{OT}, which can be used to transform a 5-local Hamiltonian $H^{(5)}$ into a 2-local $H^{(2)}$, are such that from any nearly-minimal-energy state for $H^{(2)}$ we can obtain a nearly-minimal-energy state for $H^{(5)}$.  While this property is not immediate from past work, it can be obtained by applying a powerful theorem in~\cite{OT} (building on~\cite{kkr}) which describes the behavior of $H^{(2)}$ on its low-energy subspaces.

\subsection{Proof Overview\label{TECHNIQUES}}

We now give an overview of the proof of Theorem \ref{mainthm1}, that
$\mathsf{BQP/qpoly}\subseteq\mathsf{QMA/poly}$. \ As we will explain, our
proof rests on a new idea we call the ``majority-certificates'' technique, which
is not specifically quantum and which seems likely to find other applications.

We begin with a language $L\in\mathsf{BQP/qpoly}$ and, for $n>0$, a
$\operatorname*{poly}(n)$-size quantum circuit $Q\left(  x, \xi\right)  $ that
computes $L(x)$ with high probability when given the ``correct'' advice state
$\xi=\rho_{n}$ on $\operatorname*{poly}\left(  n\right)  $ qubits. \ The
challenge, then, is to force Merlin to supply a witness state $\rho^{\prime}%
$\ that behaves like $\rho_{n}$\ on every input $x\in\left\{  0,1\right\}
^{n}$.

Every potential advice state $\xi$ defines a function $f_{\xi}:\left\{
0,1\right\}  ^{n}\rightarrow\left[  0,1\right]  $, by $f_{\xi}(x):=\Pr\left[
Q\left(  x, \xi\right)  =1\right]  $. \ For each such $\xi$, let
$\widehat{f}_{\xi}(x):=\left[  f_{\xi}(x)\geq1/2\right]  $ be the Boolean
function obtained by rounding $f_{\xi}$. \ As a simplification, suppose that
Merlin is restricted to sending an advice state $\xi$ for which $f_{\xi
}(x)\notin\left(  1/3,2/3\right)  $: that is, an advice state which renders a
``clear opinion'' about every input $x$. \ (This simplification helps to explain
the main ideas, but does not follow the actual proof.) \ Let $S$ be the set of
all Boolean functions $f:\left\{  0,1\right\}  ^{n}\rightarrow\left\{
0,1\right\}  $\ that are expressible as $\widehat{f}_{\xi}$ for some such
advice state $\xi$. \ Then $S$ includes the ``target function'' $f^{\ast}%
:=L_{n}$ (the restriction of $L$ to inputs of length $n$), as well as a
potentially-large number of other functions. \ However, we claim $S$ is not
\textit{too} large: $\left\vert S\right\vert \leq2^{\operatorname*{poly}%
\left(  n\right)  }$. \ This bound on the ``effective information content'' of
quantum states was derived previously by Aaronson \cite{aar:adv,aar:learn},
building on the work of Ambainis et al.\ \cite{antv}.

One might initially hope that, just by virtue of the size bound on $S$, we
could find some set of $\operatorname*{poly}(n)$ values%
\[
\left(  x_{1},f^{\ast}\left(  x_{1}\right)  \right)  ,\ldots,\left(
x_{k},f^{\ast}\left(  x_{t}\right)  \right)
\]
which \textit{isolate} $f^{\ast}$ in $S$---that is, which differentiate
$f^{\ast}$ from all other members of $S$. \ In that case, the trusted
classical advice could simply specify those values, as ``tests'' for Arthur to
perform on the quantum state sent by Merlin. \ Alas, this hope is unfounded in
general. \ For consider the case where $f^{\ast}$ is the identically-zero
function, and $S$ consists of $f^{\ast}$ along with the ``point function''
$f_{y}$ (which equals $1$ on $y$ and $0$ elsewhere), for all $y\in\left\{
0,1\right\}  ^{n}$. \ Then $f^{\ast}$ can only be isolated in $S$ by
specifying its value at \textit{every} point!

Luckily, this counterexample leads us to a key observation. \ Although $f^{*}$
is not isolatable in $S$ by a small number of values, each point function
$f_{y}$ \textit{can} be isolated (by its value at $y$), and moreover, $f_{y}$
is quite ``close'' to $f^{*}$. \ In fact, if we choose any three distinct
strings $x,y,z$, then $f^{\ast}\equiv\operatorname*{MAJ}\left(  f_{x}%
,f_{y},f_{z}\right)  $. \ (Of course if $f^{\ast}$\ were the identically-zero
function, it could be easily specified with classical advice! \ But $f^{\ast}%
$\ could have been any function in this example.)

This suggests a new, more indirect approach to our general problem: we try to
express $f$ as the pointwise majority vote
\[
f^{\ast}\left(  x\right)  \equiv\operatorname*{MAJ}\left(  f_{1}\left(
x\right)  ,\ldots,f_{m}\left(  x\right)  \right)  ,
\]
of a small number ($m=O\left(  n\right)  $) of \textit{other} functions
$f_{1},\ldots,f_{m}$ in $S$, where each $f_{i}$ is isolatable in $S$ by
specifying at most $k=O\left(  \log\left\vert S\right\vert \right)  $ of its
values. \ Indeed, we will show this can \textit{always} be done. \ We call
this key result the \textit{majority-certificates lemma}; we will say more
about its proof and its relation to earlier work in Section
\ref{MAJCERDISCUSSION}.

With this lemma in hand, we can solve our (artificially simplified) problem:
in the $\mathsf{QMA/poly}$ protocol for $L$, we use certificates which isolate
$f_{1},\ldots,f_{m}\in S$ as above as the classical advice for Arthur.
\ Arthur requests from Merlin each of the $m$ states $\xi_{1}, \ldots, \xi
_{m}$ such that $f_{i}=f_{\xi_{i}}$, and verifies that he receives appropriate
states by checking them against the certificates. \ This involves multiple
measurements of each $\xi_{i} $---and an immediate difficulty is that, since
measurements are irreversible in quantum mechanics, the process of verifying
the witness state\ might also destroy it. \ \ We get around this
difficulty by a somewhat more complicated protocol asking for multiple copies of each state $\xi_i$.  Our analysis builds on ideas of Aharonov and Regev \cite{ar} used to prove the complexity-class equality $\mathsf{QMA} = \mathsf{QMA}^+$; informally, this result says that protocols in which Arthur is granted the (physically unrealistic) ability to perform ``non-destructive measurements''\ on his witness state, can be efficiently simulated by
ordinary $\mathsf{QMA}$\ protocols.

To build intuition, we will begin (in Section \ref{MAJCER}) by proving the
majority-certificates lemma for Boolean functions, as described above.
\ However, to remove the artificial simplification we made and prove Theorem
\ref{mainthm1}, we will need to generalize the lemma substantially, to a
statement about possibly-infinite sets of real-valued functions $f:\left\{
0,1\right\}^{n}\rightarrow\left[  0,1\right]  $.\ \ In the general version,
the hypothesis that $S$ is finite and not too large will be replaced by a more
subtle assumption: namely, an upper bound on the so-called
\textit{fat-shattering dimension} of $S$. \ To prove our generalization, we
use powerful results of Alon et al.\ \cite{abch}\ and Bartlett and Long
\cite{bartlettlong} on the learnability of real-valued functions. \ We then
use a bound on the fat-shattering dimension of real-valued functions defined
by quantum states (from Aaronson \cite{aar:learn}, building on Ambainis et
al.\ \cite{antv}).  Figure \ref{yqpflow}\ shows the flow of ideas and results going into the proof.%

\begin{figure}[ptb]%
\centering
\includegraphics[trim=0in 0in 0in 0in, scale=.3]%
{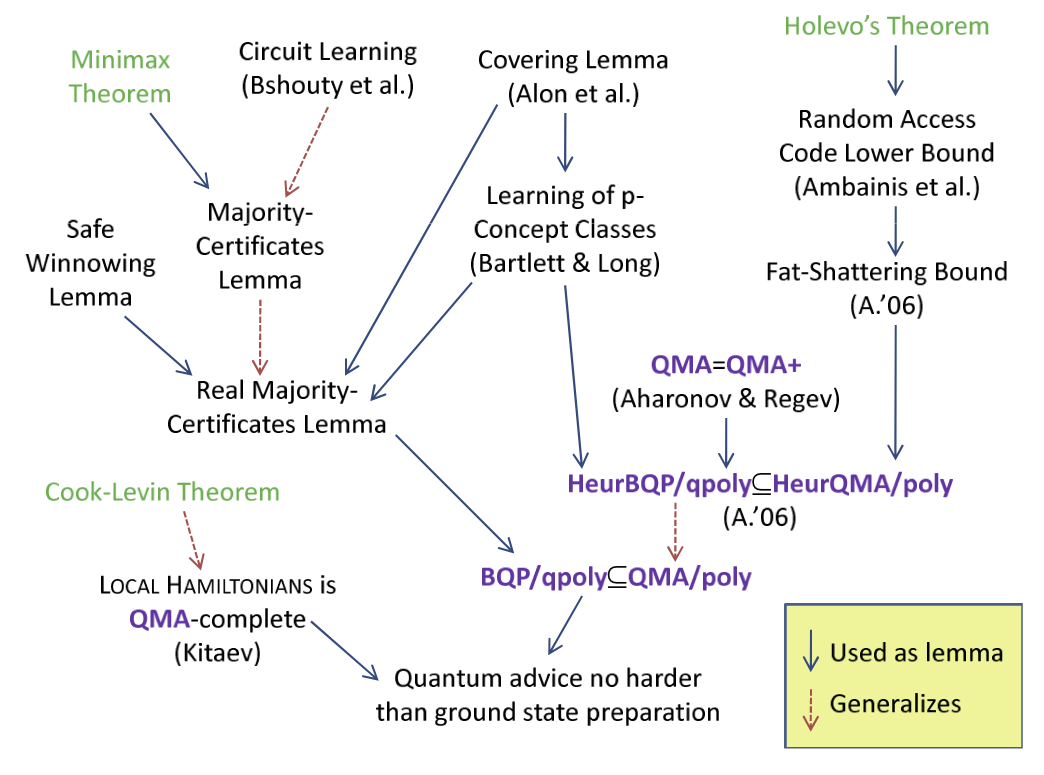}
\caption{Dependency structure of our proof that quantum advice states can be
expressed as ground states of local Hamiltonians.}%
\label{yqpflow}%
\end{figure}

\subsection{Majority-Certificates Lemma in Context\label{MAJCERDISCUSSION}}

The majority-certificates lemma is closely related to the seminal notion of
\textit{boosting} \cite{schapire} from computational learning theory.
\ Boosting is a broad topic with a vast literature, but a common ``generic''
form of the boosting problem is as follows: we want to learn some target
function $f^{\ast}$, given sample data of the form $\left(  x,f^{\ast}\left(
x\right)  \right)  $. \ We assume we have a \textit{weak learning algorithm}
$A^{f^{\ast},\mathcal{D}}$, with the property that, for any probability
distribution $\mathcal{D}$ over inputs $x$, with high probability $A$ finds a
hypothesis $f\in\mathcal{F}$ which predicts $f^{\ast}\left(  x\right)
$\ ``reasonably well'' when $x\sim\mathcal{D}$. \ The task is to ``boost'' this
weak learner into a \textit{strong} learner $B^{f^{\ast}}$. \ The strong
learner should output a collection of functions $f_{1},\ldots,f_{m}%
\in\mathcal{F}$, such that a (possibly-weighted) majority vote over
$f_{1}\left(  x\right)  ,\ldots,f_{m}\left(  x\right)  $ predicts $f^{\ast
}\left(  x\right)  $ ``extremely well.'' \ It turns out \cite{schapire,
fs_adaboost} that this goal can be achieved in a very general setting.

Our majority-certificates lemma has strengths and weaknesses compared to
boosting. \ Our assumptions are much milder than those of boosting: rather
than needing a weak learner, we assume only that the hypothesis class $S$ is
``not too large.'' \ Also, we represent our target function $f^{\ast}$
\textit{exactly} by $\operatorname*{MAJ}\left(  f_{1},\ldots,f_{m}\right)  $,
not just approximately. \ On the other hand, we do not give an efficient
algorithm to \textit{find} our majority-representation. Also, the $f_{i}$'s
are not ``explicitly given:'' we only give a way to \textit{recognize} each
$f_{i}$, under the assumption that the function purporting to be $f_{i}$ is in
fact drawn from the original hypothesis class.

The proof of our lemma also has similarities to boosting. \ As an analogue of
a ``weak learner,'' we show that for every distribution $\mathcal{D}$, there
exists a function $f\in S$ which agrees with the target function $f^{\ast}$ on
most $x\sim\mathcal{D}$, \textit{and} which is isolatable in $S$ by specifying
$O(\log|S|)$ queries. \ Using the Minimax Theorem, we then nonconstructively
``boost'' this fact into the desired majority-representation of $f^{\ast}$. \ We
note that Nisan used the Minimax Theorem for boosting in a similar way, in his
alternative proof of Impagliazzo's ``hard-core set theorem'' (see
\cite{imp_hardcore}).

The majority-certificates lemma is also reminiscent of Bshouty et al.'s
algorithm\ \cite{bshouty}, for learning small circuits in the complexity class
$\mathsf{ZPP}^{\mathsf{NP}}$. \ Our lemma lacks the algorithmic component of
this earlier work, but unlike Bshouty et al., we do not require the functions
being learned to come with any succinct labels (such as circuit descriptions).

\subsection{Organization of the Paper}


In Section \ref{MAJCER}, we prove the Boolean majority-certificates-lemma. In
Section \ref{REALMAJ}, we give our real-valued generalization of this lemma,
and in Section \ref{QADV} we use it to prove Theorem \ref{mainthm1}, and state
some consequences for quantum complexity classes.  Theorem \ref{mainthm0} is
proved in Sections \ref{LOCALHAM} through \ref{sec:2local}. Section \ref{NERD} contains some further
applications to quantum complexity theory, and the Appendices provide some additional
applications of and perspectives on the majority-certificates lemma.

\section{The Majority-Certificates Lemma\label{MAJCER}}

A \textit{Boolean concept class} is a family of sets $\left\{  S_{n}\right\}
_{n\geq1}$, where each $S_{n}$\ consists of Boolean functions $f:\left\{
0,1\right\}  ^{n}\rightarrow\left\{  0,1\right\}  $ on $n$ variables.
\ Abusing notation, we will often use $S$ to refer directly to a set of
Boolean functions on $n$ variables, with the quantification over $n$ being understood.

By a \textit{certificate}, we mean a partial Boolean function $C:\left\{
0,1\right\}  ^{n}\rightarrow\left\{  0,1,\ast\right\}  $. \ The \textit{size}
of $C$, denoted $\left\vert C\right\vert $, is the number of inputs $x$ such
that $C\left(  x\right)  \in\left\{  0,1\right\}  $. \ A Boolean function
$f:\left\{  0,1\right\}  ^{n}\rightarrow\left\{  0,1\right\}  $\ is
\textit{consistent} with $C$ if $f\left(  x\right)  =C\left(  x\right)
$\ whenever $C\left(  x\right)  \in\left\{  0,1\right\}  $. \ Given a set $S$
of Boolean functions and a certificate $C$, let $S\left[  C\right]  $\ be the
set of all functions $f\in S$\ that are consistent with $C$. \ Say that a
function $f\in S$ is \textit{isolated in} $S$\ by the certificate $C$ if
$S\left[  C\right]  =\left\{  f\right\}  $.

We now prove a lemma that represents one of the main tools of this paper
(although it will be generalized, rather than used directly).

\begin{lemma}
[Majority-Certificates Lemma]\label{majcer}Let $S$ be a set of Boolean
functions $f:\left\{  0,1\right\}  ^{n}\rightarrow\left\{  0,1\right\}  $, and
let $f^{\ast}\in S$. \ Then there exist $m=O\left(  n\right)  $\ certificates
$C_{1},\ldots,C_{m}$, each of size $k=O\left(  \log\left\vert S\right\vert
\right)  $, and functions $f_{1}, \ldots, f_{m} \in S$, such that

\begin{enumerate}
\item[(i)] $S\left[  C_{i}\right]  = \{f_{i}\}$\ all $i\in\left[  m\right]  $;

\item[(ii)] $\operatorname*{MAJ}\left(  f_{1}\left(  x\right)  ,\ldots
,f_{m}\left(  x\right)  \right)  =f^{\ast}\left(  x\right)  $ for all
$x\in\left\{  0,1\right\}  ^{n}$.
\end{enumerate}
\end{lemma}

\begin{proof}
Our proof of Lemma \ref{majcer}\ relies on the following claim.

\begin{claim}
\label{heurpart} Let $\mathcal{D}$ be any distribution over inputs
$x\in\{0,1\}^{n}$. Then there exists a function $f\in S$ such that

\begin{itemize}
\item[(i)] $f$ is isolatable in $S$ by a certificate $C$ of size $k=O\left(
\log\left\vert S\right\vert \right)  $;

\item[(ii)] $\Pr_{x\sim\mathcal{D}}[f(x)\neq f^{\ast}(x)]\leq\frac{1}{10}$.
\end{itemize}
\end{claim}

Lemma~\ref{majcer} follows from Claim~\ref{heurpart} by a boosting-type
argument, as follows. Consider a two-player game where:

\begin{itemize}
\item Alice chooses a certificate $C$ of size $k$ that isolates some $f\in S$, and

\item Bob simultaneously chooses an input $x\in\left\{  0,1\right\}  ^{n}$.
\end{itemize}

Alice wins the game if $f\left(  x\right)  =f^{\ast}\left(  x\right)  $.
Claim~\ref{heurpart} tells us that for every mixed strategy of Bob (i.e.,
distribution $\mathcal{D}$ over inputs), there exists a pure strategy of Alice
that succeeds with probability at least $0.9$ against $\mathcal{D}$. \ Then by
the Minimax Theorem, there exists a mixed strategy for Alice---that is, a
probability distribution $\mathcal{C}$ over certificates---that allows her to
win with probability at least $0.9$\ against \textit{every} pure strategy of Bob.

Now suppose we draw $C_{1},\ldots,C_{m}$ independently from $\mathcal{C}$,
isolating functions $f_{1},\ldots,f_{m}$ in $S$. Fix an input $x\in\left\{
0,1\right\}^{n}$; then by the success of Alice's strategy against $x$, and
applying a Chernoff bound,%
\[
\Pr_{C_{1},\ldots,C_{m}\sim\mathcal{C}}\left[  \operatorname*{MAJ}\left(  f_{1}\left(
x\right)  ,\ldots,f_{m}\left(  x\right)  \right)  \neq f^{\ast}(x)\right]
<\frac{1}{2^{n}},
\]
provided we choose $m=O\left(  n\right)  $\ suitably. \ But by the union
bound, this means there must be a \textit{fixed} choice of $C_{1},\ldots
,C_{m}$ such that $\operatorname*{MAJ}\left(  f_{1},\ldots,f_{m}\right)
\equiv f^{\ast}$, where each $f_{i}$ is isolated in $S$ by $C_{i}$. This
proves Lemma \ref{majcer}, modulo the Claim.
\end{proof}

\begin{proof}
[Proof of Claim \ref{heurpart}]By symmetry, we can assume without loss of
generality that $f^{\ast}$ is the identically-zero function. \ Given the mixed
strategy $\mathcal{D}$ of Bob, we construct the certificate $C$ as follows.
\ Initially $C$ is empty: that is, $C\left(  x\right)  =\ast$\ for all
$x\in\left\{  0,1\right\}  ^{n}$. \ In the first stage, we draw $t=O\left(
\log\left\vert S\right\vert \right)  $\ inputs $x_{1},\ldots,x_{t}%
$\ independently from $\mathcal{D}$. \ For any $f:\{0,1\}^{n}\rightarrow
\{0,1\}$, let%
\[
w_{f}:=\Pr_{x\sim\mathcal{D}}\left[  f\left(  x\right)  =1\right]  .
\]
Now suppose $f$ is such that $w_{f}>0.1$. \ Then%
\[
\Pr_{x_{1},\ldots,x_{t}\sim\mathcal{D}}\left[  f\left(  x_{1}\right)
=0\wedge\cdots\wedge f\left(  x_{t}\right)  =0\right]  <0.9^{t}\leq\frac
{1}{\left\vert S\right\vert },
\]
provided $t\geq\log_{10/9}\left\vert S\right\vert $. \ So by the union bound,
there must be a \textit{fixed} choice of $x_{1},\ldots,x_{t}$\ that kills off
every $f\in S$\ such that $w_{f}>0.1$---that is, such that $f\left(
x_{1}\right)  =\cdots=f\left(  x_{t}\right)  =0$\ implies $w_{f}\leq0.1$.
\ Fix that $x_{1},\ldots,x_{t}$, and set $C\left(  x_{i}\right)  :=0$ for all
$i\in\left[  t\right]  $.

In the second stage, our goal is just to isolate some \textit{particular}
function $f\in S\left[  C\right]  $. \ We do this recursively as follows. \ If
$\left\vert S\left[  C\right]  \right\vert =1$\ then we are done. \ Otherwise,
there exists an input $x$ such that $f\left(  x\right)  \neq f^{\prime}\left(
x\right) $\ for some pair $f, f^{\prime}\in S\left[  C\right]  $. \ If setting
$C\left(  x\right)  :=0$\ decreases $\left\vert S\left[  C\right]  \right\vert
$ by at least a factor of $2$, then set $C\left(  x\right)  :=0$; otherwise
set $C\left(  x\right)  :=1$. \ Since $S\left[  C\right]  $ can halve in size
at most $\log_{2}\left\vert S\right\vert $\ times, this procedure terminates
after at most $\log_{2}\left\vert S\right\vert $\ steps with $\left\vert
S\left[  C\right]  \right\vert =1$.

The end result is a certificate $C$ of size $O\left(  \log\left\vert
S\right\vert \right)  $, which isolates a function $f$ in $S$ for which
$w_{f}\leq1/10$. \ We have therefore found a pure strategy for Alice that
fails with probability at most $1/10$\ against $\mathcal{D}$, as desired.
\end{proof}

\section{Extension to Real Functions\label{REALMAJ}}

In this section, we extend the majority-certificates lemma from Boolean
functions to real-valued functions $f:\left\{  0,1\right\}  ^{n}%
\rightarrow\left[  0,1\right]  $. \ We will need this extension for the
application to quantum advice in Section \ref{QADV}. \ In proving our
extension we will have to confront several new difficulties. \ Firstly, the
concept classes $S$ that we want to consider can now contain a
\textit{continuum} of functions---so Lemma \ref{majcer}, which assumed that
$S$ was finite and constructed certificates of size\ $O\left(  \log\left\vert
S\right\vert \right)  $, is not going to work. \ In Section \ref{LEARNBACK},
we review notions from computational learning theory, including fat-shattering
dimension and $\varepsilon$-covers, which\ (combined with results of Alon et
al.\ \cite{abch} and Bartlett and Long \cite{bartlettlong}) can be used to get
around this difficulty. \ Secondly, it is no longer enough to isolate a
function $f_{i}\in S$\ that we are interested in; instead we will need to
``safely''\ isolate $f_{i}$, which roughly speaking means that (i) $f_{i}$ is
consistent with some certificate $C$, and (ii) any $f\in S$\ that is even
\textit{approximately} consistent with $C$ is close to\ $f_{i}$. \ In Section
\ref{SAFEWINNOW_REALMAJ}, we prove a ``safe winnowing lemma''\ that can be used for this
purpose, and put our ingredients together to
prove a real-valued majority-certificates lemma.

\subsection{Background from Learning Theory\label{LEARNBACK}}

A \textit{p-concept class} $S$ over $\{0, 1\}^n$ is a family of functions $f:\left\{
0,1\right\}  ^{n}\rightarrow\left[  0,1\right]  $\ (as usual, quantification
over all $n$ is understood). \ Given functions $f,g:\left\{  0,1\right\}
^{n}\rightarrow\left[  0,1\right]  $ and a subset of inputs $X\subseteq
\left\{  0,1\right\}  ^{n}$, we will be interested in three measures of the
distance between $f$ and $g$ restricted to $X$:%
\begin{align*}
\Delta_{\infty}\left(  f,g\right)  \left[  X\right]   &  :=\max_{x\in
X}\left\vert f\left(  x\right)  -g\left(  x\right)  \right\vert ,\\
\Delta_{2}\left(  f,g\right)  \left[  X\right]   &  :=\sqrt{\sum_{x\in
X}\left(  f\left(  x\right)  -g\left(  x\right)  \right)  ^{2}},\\
\Delta_{1}\left(  f,g\right)  \left[  X\right]   &  :=\sum_{x\in X}\left\vert
f\left(  x\right)  -g\left(  x\right)  \right\vert .
\end{align*}
For convenience, we define $\Delta_{\infty}\left(  f,g\right)  :=\Delta
_{\infty}\left(  f,g\right)  \left[  \left\{  0,1\right\}  ^{n}\right]  $, and
similarly for $\Delta_{2}\left(  f,g\right)  $\ and $\Delta_{1}\left(
f,g\right)  $. \ Also, given a distribution $\mathcal{D}$ over $\left\{
0,1\right\}  ^{n}$, define%
\[
\Delta_{1}\left(  f,g\right)  \left\langle \mathcal{D}\right\rangle
:=\operatorname*{E}_{x\sim\mathcal{D}}\left[  \left\vert f\left(  x\right)
-g\left(  x\right)  \right\vert \right]  .
\]
Finally, we will need the notions of coverability and fat-shattering dimension.

\begin{definition}
[Coverability]Let $S$ be a p-concept class over $\{0, 1\}^n$. \ The subset $C\subseteq S$\ is an
$\varepsilon$\emph{-cover}\ for $S$ if for all $f\in S$,\ there exists a $g\in
C$\ such that $\Delta_{\infty}\left(  f,g\right)  \leq\varepsilon$.  \ We say
$S$ is \emph{coverable} if for all $\varepsilon>0$, there exists an $\varepsilon
$-cover\ for $S$\ of size $2^{\operatorname*{poly}\left(  n,1/\varepsilon
\right)  }$.

\end{definition}

\begin{definition}
[Fat-Shattering Dimension]Let $S$ be a p-concept class over $\{0, 1\}^n$ and $\varepsilon>0$\ be
a real number. \ We say the set $A\subseteq\left\{  0,1\right\}  ^{n}$\ is
\emph{$\varepsilon$-shattered} by $S$ if there exists a function $r:A\rightarrow
\left[  0,1\right]  $\ such that for all $2^{\left\vert A\right\vert }$
Boolean functions $g:A\rightarrow\left\{  0,1\right\}  $, there exists a
p-concept $f\in S$\ such that for all $x\in A$, we have $f\left(  x\right)
\leq r\left(  x\right)  -\varepsilon$\ whenever $g\left(  x\right)  =0$\ and
$f\left(  x\right)  \geq r\left(  x\right)  +\varepsilon$\ whenever $g\left(
x\right)  =1$. \ Then the $\varepsilon$-fat-\emph{shattering dimension}
of\ $S$, denoted $\operatorname*{fat}_{\varepsilon}\left(  S\right)  $, is the size
of the largest set $\varepsilon$-shattered by $S$.

We say $S$ is bounded-dimensional if $\operatorname*{fat}_{\varepsilon}\left(  S\right) \leq\operatorname*{poly}\left(  n,1/\varepsilon\right)  $\ for all $\varepsilon>0$. \end{definition}

The p-concept classes we consider in this paper will be convex, when considered as subsets of $[0, 1]^{2^n}$.  We remark that for such classes, $\operatorname*{fat}_{\varepsilon}\left(  S\right)$ measures the largest dimension of any axis-parallel subcube contained in $S$ of side length $2\eps$.

The following central result was shown by Alon et al.\ \cite{abch} (see also
\cite{tewari}).

\begin{theorem}
[\cite{abch}]\label{coverthm}Every p-concept class $S$ has an $\varepsilon
$-cover of size $\exp\left[  O\left(  \left(  n+\log1/\varepsilon\right)
\operatorname*{fat}_{\varepsilon/4}\left(  S\right)  \right)  \right]$.
\end{theorem}

Building on the work of Alon et al.\ \cite{abch}, Bartlett and Long
\cite{bartlettlong} then proved the following:

\begin{theorem}
[\cite{bartlettlong}]\label{blthm}Let $S$ be a p-concept class and
$\mathcal{D}$ be a distribution over $\left\{  0,1\right\}  ^{n}$. \ Fix an
$f:\left\{  0,1\right\}  ^{n}\rightarrow\left[  0,1\right]  $\ (not
necessarily in $S$)\ and an error parameter $\alpha>0$. \ Suppose we form a
set $X\subseteq\left\{  0,1\right\}  ^{n}$ by choosing $m$ inputs
independently with replacement from $\mathcal{D}$. \ Then there exists a
positive constant $K$ such that, with probability at least $1-\delta$ over
$X$, any hypothesis $h\in S$\ that minimizes $\Delta_{1}\left(  h,f\right)
\left[  X\right]  $ also satisfies%
\[
\Delta_{1}\left(  h,f\right)  \left\langle \mathcal{D}\right\rangle \leq
\alpha+\inf_{g\in S}\Delta_{1}\left(  g,f\right)  \left\langle \mathcal{D}%
\right\rangle ,
\]
provided that%
\[
m\geq\frac{K}{\alpha^{2}}\left(  \operatorname*{fat}\nolimits_{\alpha
/5}\left(  S\right)  \log^{2}\frac{1}{\alpha}+\log\frac{1}{\delta}\right)  .
\]

\end{theorem}

Theorem \ref{blthm}\ has the following corollary, which is similar to
Corollary 2.4 of Aaronson \cite{aar:learn}, but more directly suited to our
purposes here.\footnote{It would also be possible to apply the bound from
\cite{aar:learn} ``off-the-shelf,''\ but at the cost of a worse dependence on
$1/\varepsilon$.}

\begin{corollary}
\label{occamcor}Let $S$ be a p-concept class over $\{0, 1\}^n$ and $\mathcal{D}$ be a
distribution over $\left\{  0,1\right\}  ^{n}$. \ Fix an $f\in S$\ and an
error parameter $\varepsilon>0$. \ Suppose we form a set $X\subseteq\left\{
0,1\right\}  ^{n}$ by choosing $m$ inputs independently with replacement from
$\mathcal{D}$. \ Then there exists a positive constant $K$ such that, with
probability at least $1-\delta$ over $X$, any hypothesis $h\in S$\ that
satisfies $\Delta_{\infty}\left(  h,f\right)  \left[  X\right]  \leq
\varepsilon$ also satisfies $\Delta_{1}\left(  h,f\right)  \left\langle
\mathcal{D}\right\rangle \leq11\varepsilon$,\ provided%
\[
m\geq\frac{K}{\varepsilon^{2}}\left(  \operatorname*{fat}%
\nolimits_{\varepsilon}\left(  S\right)  \log^{2}\frac{1}{\varepsilon}%
+\log\frac{1}{\delta}\right)  .
\]

\end{corollary}

\begin{proof}
Let $S^{\ast}$\ be the p-concept class consisting of all functions $g:\left\{
0,1\right\}  ^{n}\rightarrow\left[  0,1\right]  $\ for which there exists an
$f\in S$\ such that $\Delta_{\infty}\left(  g,f\right)  \leq\varepsilon$.
\ Fix an $f\in S$ and a distribution $\mathcal{D}$, and let $X$ be chosen as
in the statement of the corollary. \ Suppose we choose a hypothesis $h\in
S$\ such that $\Delta_{\infty}\left(  h,f\right)  \left[  X\right]  
\leq\varepsilon$. \  Define a function $g$ by setting $g\left(  x\right)  :=h\left(  x\right)  $\ if $x\in X$\ and $g\left(  x\right)  :=f\left(  x\right)  $\ otherwise.  Note that $\Delta_{\infty}\left( g, f \right) \leq \varepsilon$ and that $g \in S^*$.
Also note that\ $\Delta_{1}\left(  h,g\right)  \left[  X\right]  =0$, which means that $h$ minimizes the functional $\Delta_{1}\left(  h,g\right) \left[  X\right]  $\ over all hypotheses in $S$ (and indeed in $S^{\ast}$).
\ By Theorem \ref{blthm}, this implies that with probability at least
$1-\delta$ over $X$,%
\[
\Delta_{1}\left(  h,g\right)  \left\langle \mathcal{D}\right\rangle \leq
\alpha+\inf_{u\in S^{\ast}}\Delta_{1}\left(  u,g\right)  \left\langle
\mathcal{D}\right\rangle =\alpha
\]
for all $\alpha>0$, provided we take%
\[
m\geq\frac{K}{\alpha^{2}}\left(  \operatorname*{fat}\nolimits_{\alpha
/5}\left(  S^{\ast}\right)  \log^{2}\frac{1}{\alpha}+\log\frac{1}{\delta
}\right)  .
\]
Here we have used the fact that $g\in S^{\ast}$, and hence%
\[
\inf_{u\in S^{\ast}}\Delta_{1}\left(  u,g\right)  \left\langle \mathcal{D}%
\right\rangle =0.
\]
So by the triangle inequality,%
\begin{align*}
\Delta_{1}\left(  h,f\right)  \left\langle \mathcal{D}\right\rangle  &
\leq\Delta_{1}\left(  h,g\right)  \left\langle \mathcal{D}\right\rangle
+\Delta_{1}\left(  g,f\right)  \left\langle \mathcal{D}\right\rangle \\
&  \leq\alpha+\Delta_{\infty}\left(  g,f\right) \\
&  \leq\alpha+\varepsilon.
\end{align*}
Next,\ we claim that $\operatorname*{fat}\nolimits_{\alpha/5}\left(  S^{\ast
}\right)  \leq\operatorname*{fat}\nolimits_{\alpha/5-\varepsilon}\left(
S\right)  $. \ The reason is simply that, if a given set is $\beta$
-fat-shattered by \ $S^{\ast}$, then it must also be $\left(  \beta-\varepsilon
\right)  $-fat-shattered by $S$,\ by the triangle inequality. \ Setting
$\alpha:=10\varepsilon$\ now yields the desired statement.
\end{proof}

\subsection{The Safe Winnowing Lemma and the Real-Valued Majority-Certificates Lemma\label{SAFEWINNOW_REALMAJ}}

An important technical step toward proving the real-valued majority-certificates lemma is our so-called ``Safe Winnowing Lemma.''\ \ This lemma says intuitively that,
given any set $S$ of real-valued functions with a small $\varepsilon$-cover
(or equivalently, with polynomially-bounded fat-shattering dimension), and given any $f^* \in S$ and subset $Y \subseteq \{0, 1\}^n$ of inputs to $f^*$, it is
possible to find a set of $k=\operatorname*{poly}\left(  n\right)  $
constraints $\left\vert f\left(  x_{1}\right)  -a_{1}\right\vert \leq\epsilon
$, \ldots, $\left\vert f\left(  x_{k}\right)  -a_{k}\right\vert \leq\epsilon$, and another function $f \in S$, such that $f$ is close to $f^*$ in $L_{\infty}$ norm on $Y$, and $f$ is \textit{essentially} the only function in $S$ compatible with the constraints. \ Here ``essentially''\ means that (i) any function that satisfies the
constraints is close to $f^*$\ in $L_{\infty}$-norm, and (ii) $f^*$ itself not
only satisfies the constraints, but does so with a ``margin to spare.''

\begin{lemma}
[Safe Winnowing Lemma]\label{linfwinnow}Let $S$ be a p-concept class over $\{0, 1\}^n$. \ Fix a
function $f^{\ast}\in S$ and subset $Y\subseteq\left\{  0,1\right\}  ^{n}$.
\ For some parameter $\varepsilon>0$, let $C$ be a finite $\varepsilon$-cover
for $S$. \ Then there exists an $f\in S$, as well as a subset $Z\subseteq
\left\{  0,1\right\}  ^{n}$\ of size at most $k=\log_{2}\left\vert
C\right\vert $,\ such that:

\begin{enumerate}
\item[(i)] Every $g\in S$\ that satisfies $\Delta_{\infty}\left(  f,g\right)
\left[  Y\cup Z\right]  \leq\frac{\varepsilon}{5k}$\ also satisfies
$\Delta_{\infty}\left(  f,g\right)  \leq3\varepsilon$.

\item[(ii)] $\Delta_{\infty}\left(  f,f^{\ast}\right)  \left[  Y\right]
\leq\varepsilon/5$.
\end{enumerate}
\end{lemma}

Note that Lemma \ref{linfwinnow} is still interesting in the special case
$Y=\varnothing$, in which case $f^{\ast}$\ is irrelevant, and the problem
reduces to finding a $Z$\ such that every $g\in S$\ that satisfies
$\Delta_{\infty}\left(  f,g\right)  \left[  Z\right]  \leq\frac{\varepsilon
}{5k}$\ also satisfies $\Delta_{\infty}\left(  f,g\right)  \leq3\varepsilon$.

\ In Appendix \ref{WINNOW}, we will develop the theory of ``winnowability''\ of
p-concept classes for its own sake. \ We show there that the condition
$\Delta_{\infty}\left(  f,g\right)  \left[  Z\right]  =O\left(  \varepsilon
/k\right)  $ can be improved to $\Delta_{1}\left(  f,g\right)  \left[
Z\right]  =O\left(  \varepsilon\right)  $. \ On the other hand, the proof
becomes more involved, and we no longer know how to incorporate $f^{\ast}%
$\ and $Y$. \ We also show that the condition $\Delta_{\infty}\left(
f,g\right)  \left[  Z\right]  =O\left(  \varepsilon/k\right)  $%
\ \textit{cannot} be improved to $\Delta_{\infty}\left(  f,g\right)  \left[
Z\right]  =O\left(  \varepsilon\right)  $\ or even $\Delta_{2}\left(
f,g\right)  \left[  Z\right]  =O\left(  \varepsilon\right)  $.


We defer the proof of Lemma~\ref{linfwinnow}, showing first how it helps us to prove our generalization of Lemma \ref{majcer}\ to the case of real-valued functions:

\begin{lemma}
[Real Majority-Certificates]\label{realmajcer}Let $S$ be a p-concept class over $\{0, 1\}^n$,
let $f^{\ast}\in S$, and let $\varepsilon>0$. \ Then for some $m=O\left(
n/\varepsilon^{2}\right)  $, there exist functions $f_{1},\ldots,f_{m}\in S$,
sets $X_{1},\ldots,X_{m}\subseteq\left\{  0,1\right\}  ^{n}$\ each of size
$k=O\left(  \left(  n+\frac{\log^{2}1/\varepsilon}{\varepsilon^{2}}\right)
\operatorname*{fat}\nolimits_{\varepsilon/48}\left(  S\right)  \right)  $, and
an\ $\alpha=\Omega\left(  \frac{\varepsilon}{\left(  n+\log1/\varepsilon
\right)  \operatorname*{fat}\nolimits_{\varepsilon/48}\left(  S\right)
}\right)  $ for which the following holds. \ All $g_{1},\ldots,g_{m}\in
S$\ that satisfy $\Delta_{\infty}\left(  f_{i},g_{i}\right)  \left[
X_{i}\right]  \leq\alpha$\ for $i\in\left[  m\right]  $\ also satisfy
$\Delta_{\infty}\left(  f^{\ast},g\right)  \leq\varepsilon$, where%
\[
g\left(  x\right)  :=\frac{g_{1}\left(  x\right)  +\cdots+g_{m}\left(
x\right)  }{m}.
\]

\end{lemma}

\begin{proof}
Let%
\begin{align*}
\beta &  :=\frac{\varepsilon}{48},\\
t  &  :=C\left(  n+\log\frac{1}{\beta}\right)  \operatorname*{fat}%
\nolimits_{\beta}\left(  S\right)  ,\\
\alpha &  :=\frac{0.4\beta}{t},
\end{align*}
where $C$ is a suitably large constant. \ Also, let $S_{\operatorname*{fin}}%
$\ be a finite $\alpha$-cover for $S$: that is, a finite subset
$S_{\operatorname*{fin}}\subseteq S$\ such that for all $f\in S$, there exists
a $g\in S_{\operatorname*{fin}}$\ such that $\Delta_{\infty}\left(
f,g\right)  \leq\alpha$.\footnote{We will need $S_{\operatorname*{fin}}$ for
the technical reason that the basic Minimax Theorem only works with finite
strategy spaces.} \ Given $f$ and $X$, let $S\left[  f,X\right]  $\ be the set
of all $g\in S$\ such that $\Delta_{\infty}\left(  f,g\right)  \left[
X\right]  \leq\alpha$.

Now consider a two-player game where Alice chooses a function $f\in
S_{\operatorname*{fin}}$ and a set $X\subseteq\left\{  0,1\right\}  ^{n}$ of
size $k$, and Bob simultaneously chooses an input $x\in\left\{  0,1\right\}
^{n}$. \ Alice's \textit{penalty} in this game (the number she is trying to
minimize) equals%
\[
\sup_{g\in S\left[  f,X\right]  }\left\vert f^{\ast}\left(  x\right)
-g\left(  x\right)  \right\vert .
\]
We claim that there exists a mixed strategy for Alice---that is, a probability
distribution $\mathcal{P}$ over $\left(  f,X\right)  $\ pairs---that gives her
an expected penalty of at most $\varepsilon/2$\ against every pure strategy of Bob.

Let us see why Lemma~\ref{realmajcer} follows from this claim. \ Fix an input $x\in\left\{
0,1\right\}  ^{n}$,\ and suppose Alice draws $\left(  f_{1},X_{1}\right)
,\ldots,\left(  f_{m},X_{m}\right)  $\ independently from $\mathcal{P}$.
\ Then for all $i\in\left[  m\right]  $,%
\[
\bE_{\left(  f_{i},X_{i}\right)  \sim\mathcal{P}}\left[
\sup_{g\in S\left[  f,X\right]  }\left\vert f^{\ast}\left(  x\right)
-g\left(  x\right)  \right\vert \right]  \leq\frac{\varepsilon}{2}.
\]
Thus, letting $z_{1},\ldots,z_{m}$\ be independent random variables in
$\left[  0,1\right]  $,\ each with expectation at most $\varepsilon/2$, the
expression%
\[
\Pr_{\left(  f_{1},X_{1}\right)  ,\ldots,\left(  f_{m},X_{m}\right)
\sim\mathcal{P}}\left[  \exists g_{1}\in S\left[  f_{1},X_{1}\right]
,\ldots,g_{m}\in S\left[  f_{m},X_{m}\right]  ~:~\left\vert f^{\ast}\left(
x\right)  -\frac{g_{1}\left(  x\right)  +\cdots+g_{m}\left(  x\right)  }%
{m}\right\vert >\varepsilon\right]
\]
is at most $\Pr\left[  z_{1}+\cdots z_{m}>\varepsilon m\right]  $\ using the
triangle inequality. \ This, in turn, is less than%
\[
2\exp\left(  -\frac{2\left(  \varepsilon m / 2\right)  ^{2}}{m}\right)
<2^{-n}%
\]
by Hoeffding's inequality, provided we choose $m=O\left(  n/\varepsilon
^{2}\right)  $\ suitably. \ By the union bound, this means that there must be
a fixed choice of $f_{1},\ldots,f_{m}$\ and $X_{1},\ldots,X_{m}$\ such that%
\[
\left\vert f^{\ast}\left(  x\right)  -\frac{g_{1}\left(  x\right)
+\cdots+g_{m}\left(  x\right)  }{m}\right\vert \leq\varepsilon
\]
for all $g_{1}\in S\left[  f_{1},X_{1}\right]  ,\ldots,g_{m}\in S\left[
f_{m},X_{m}\right]  $\ and all inputs $x\in\left\{  0,1\right\}  ^{n}%
$\ simultaneously, as desired.

We now prove the claim. \ By the Minimax Theorem, our task is equivalent to
the following: given any mixed strategy\ $\mathcal{D}$\ of Bob, find a
\textit{pure} strategy of Alice that achieves a penalty of at most
$\varepsilon/2$ against $\mathcal{D}$. \ In other words, given any
distribution $\mathcal{D}$\ over inputs $x\in\left\{  0,1\right\}  ^{n}$, we
want a fixed function $f\in S_{\operatorname*{fin}}$, and a set $X\subseteq
\left\{  0,1\right\}  ^{n}$\ of size $k$, such that%
\[
\bE_{x\sim\mathcal{D}}\left[  \sup_{g\in S\left[  f,X\right]
}\left\vert f^{\ast}\left(  x\right)  -g\left(  x\right)  \right\vert \right]
\leq\frac{\varepsilon}{2}.
\]
We construct this $\left(  f,X\right)  $\ pair as follows.
In the first stage, we let $Y$\ be a set, of size at most%
\[
M:=\frac{K}{\beta^{2}}\left(  \operatorname*{fat}\nolimits_{\beta}\left(
S\right)  \log^{2}\frac{1}{\beta}+\log\frac{1}{\delta}\right)  ,
\]
formed by choosing $M$\ inputs independently with replacement from
$\mathcal{D}$. \ Here $\beta=\varepsilon/48$ as defined earlier, $\delta=1/2$,
and $K$ is the constant from Corollary \ref{occamcor}. \ Then by Corollary
\ref{occamcor}, with probability at least $1-\delta=1/2$\ over the choice of
$Y$, any $g\in S$\ that satisfies $\Delta_{\infty}\left(  f^{\ast},g\right)
\left[  Y\right]  \leq\beta$\ also satisfies $\Delta_{1}\left(  f^{\ast
},g\right)  \left\langle \mathcal{D}\right\rangle \leq11\beta$. \ So there
must be a \textit{fixed} choice of $Y$\ with that property. \ Fix that $Y$,
and let $S^{\prime}$ be the set of all $g\in S$\ such that $\Delta_{\infty
}\left(  f^{\ast},g\right)  \left[  Y\right]  \leq\beta$.

In the second stage, our goal is just to use Lemma~\ref{linfwinnow} to winnow $S^{\prime}$ down to a
particular function $f$. \ More precisely, we want to find an $f\in S^{\prime
}\cap S_{\operatorname*{fin}}$, and a set $X\subseteq\left\{  0,1\right\}
^{n}$ containing $Y$, such that any $g\in S$ that satisfies $\Delta_{\infty
}\left(  f,g\right)  \left[  X\right]  \leq\alpha$\ also satisfies
$\Delta_{\infty}\left(  f,g\right)  \leq11\beta$.  We assert that such a pair $\left(f, X \right)$ can be found.  It will then follow that
\begin{align*}
\bE_{x\sim\mathcal{D}}\left[  \sup_{g\in S\left[  f,X\right]
}\left\vert f^{\ast}\left(  x\right)  -g\left(  x\right)  \right\vert \right]
&  \leq\Delta_{1}\left(  f^{\ast},f\right)  \left\langle \mathcal{D}%
\right\rangle +\sup_{g\in S\left[  f,X\right]  }\Delta_{\infty}\left(
f,g\right) \\
&  \leq11\beta+13\beta\\
&  =\frac{\varepsilon}{2},%
\end{align*}
which proves that $(f, X)$ give a strategy for Alice having the needed quality against the mixed strategy $\mathcal{D}$ for Bob.

We find the desired $\left(  f,X\right)  $\ pair as follows. \ By Theorem
\ref{coverthm}, the class $S^{\prime}$ has a $4\beta$-cover of size%
\[
N=\exp\left[  O\left(  \left(  n+\log\frac{1}{4\beta}\right)
\operatorname*{fat}\nolimits_{\beta}\left(  S^{\prime}\right)  \right)
\right]  \leq\exp\left[  O\left(  \left(  n+\log\frac{1}{\beta}\right)
\operatorname*{fat}\nolimits_{\beta}\left(  S\right)  \right)  \right]  .
\]
Let $t:=\log_{2}N$. \ Then by Lemma \ref{linfwinnow}, there exists a function
$u\in S^{\prime}$, as well as a subset $Z\subseteq\left\{  0,1\right\}  ^{n}%
$\ of size at most $t$, such that:

\begin{enumerate}
\item[(i)] $\Delta_{\infty}\left(  u,f^{\ast}\right)  \left[  Y\right]
\leq0.8\beta$.

\item[(ii)] Every $g\in S^{\prime}$\ that satisfies $\Delta_{\infty}\left(
u,g\right)  \left[  Y\cup Z\right]  \leq\frac{0.8\beta}{t}$\ also satisfies
$\Delta_{\infty}\left(  u,g\right)  \leq12\beta$.
\end{enumerate}

Let $X:=Y\cup Z$, and observe that%
\begin{align*}
\left\vert X\right\vert  &  =O\left(  \frac{1}{\beta^{2}}\operatorname*{fat}%
\nolimits_{\beta}\left(  S\right)  \log^{2}\frac{1}{\beta}+\left(  n+\log
\frac{1}{\beta}\right)  \operatorname*{fat}\nolimits_{\beta}\left(  S\right)
\right) \\
&  =O\left(  \left(  n+\frac{\log^{2}1/\varepsilon}{\varepsilon^{2}}\right)
\operatorname*{fat}\nolimits_{\varepsilon/48}\left(  S\right)  \right)
\end{align*}
as desired. \ Now let $f$ be a function in $S_{\operatorname*{fin}}$\ such
that $\Delta_{\infty}\left(  f,u\right)  \leq\alpha$. \ Let us check that $(f, X)$
have the properties we want. \ First,%
\begin{align*}
\Delta_{\infty}\left(  f^{\ast},f\right)  \left[  Y\right]   &  \leq
\Delta_{\infty}\left(  f^{\ast},u\right)  \left[  Y\right]  +\Delta_{\infty
}\left(  u,f\right)  \left[  Y\right] \\
&  \leq0.8\beta+\alpha\\
&  <0.9\beta,
\end{align*}
hence $f\in S^{\prime}$\ as desired. \ Next, consider any $g\in S$ that satisfies
$\Delta_{\infty}\left(  f,g\right)  \left[  X\right]  \leq\alpha$.  Then we also have
\begin{align*}
\Delta_{\infty}\left(  f^{\ast},g\right)  \left[  Y\right]   &  \leq
\Delta_{\infty}\left(  f^{\ast},f\right)  \left[  Y\right]  +\Delta_{\infty
}\left(  f,g\right)  \left[  Y\right] \\
&  \leq0.9\beta+\alpha\\
&  <\beta,
\end{align*}
hence $g\in S^{\prime}$, so that (by our construction of $Y$) we have $\Delta_{1}\left(  f^{\ast},g\right)
\left\langle \mathcal{D}\right\rangle \leq11\beta$. \ Next, observe that
\begin{align*}
\Delta_{\infty}\left(  u,g\right)  \left[  X\right]   &  \leq\Delta_{\infty
}\left(  u,f\right)  \left[  X\right]  +\Delta_{\infty}\left(  f,g\right)
\left[  X\right] \\
&  \leq2\alpha\\
&  =\frac{0.8\beta}{t},
\end{align*}
so that, using our guarantee (ii) above, we have $\Delta_{\infty}\left(  u,g\right)  \leq12\beta$.  Then we find that%
\begin{align*}
\Delta_{\infty}\left(  f,g\right)   &  \leq\Delta_{\infty}\left(  f,u\right)
+\Delta_{\infty}\left(  u,g\right) \\
&  \leq\alpha+12\beta\\
&  \leq13\beta,
\end{align*}
as required.  This shows that $(f, X)$ have the required properties, and completes the proof of Lemma~\ref{realmajcer}.
\end{proof}

\begin{proof}[Proof of Lemma~\ref{linfwinnow}]
Let $\delta:=\frac{\varepsilon}{5k}$. \ We construct $\left(  f,Z\right)
$\ by an iterative procedure. \ Initially let $S_{0}:=S$, let $f_{0}:=f^{\ast
}$, and let $Z_{0}:=Y$. \ We will form new sets $S_{1},S_{2},\ldots$\ by
repeatedly adding constraints of the form $f\left(  x\right)  \leq\alpha$\ or
$f\left(  x\right)  \geq\alpha$\ for various $x,\alpha$, maintaining the
invariant that $f_{t}\in S_{t}$. \ At iteration $t$, suppose there exists a
function $g\in S_{t-1}$\ such that $\Delta_{\infty}\left(  f_{t-1},g\right)
\left[  Y\cup Z_{t-1}\right]  \leq\delta$,\ but nevertheless $\left\vert
f_{t-1}\left(  z_{t}\right)  -g\left(  z_{t}\right)  \right\vert
>3\varepsilon$\ for some input $z_{t}$. \ Then first set $Z_{t}:=Z_{t-1}%
\cup\left\{  z_{t}\right\}  $ (i.e., add $z_{t}$\ into our set of inputs, if
it is not already there). \ Let $v:=\frac{1}{2}\left[  f_{t-1}\left(
z_{t}\right)  +g\left(  z_{t}\right)  \right]  $,\ let $A$\ be the set of all
functions $h\in S_{t-1}$\ such that $h\left(  z_{t}\right)  <v$, and let
$B$\ be the set of all $h\in S_{t-1}$\ such that $h\left(  z_{t}\right)  \geq
v$. \ Also, for any given set $M$, let $M^{\Diamond}:=M\cap C$. \ Then clearly
$\min\left\{  \left\vert A^{\Diamond}\right\vert ,\left\vert B^{\Diamond
}\right\vert \right\}  \leq\left\vert S_{t-1}^{\Diamond}\right\vert /2$.\ \ If
$\left\vert A^{\Diamond}\right\vert <\left\vert B^{\Diamond}\right\vert $,
then set $S_{t}:=A$; otherwise set $S_{t}:=B$. \ Then set $f_{t}:=f_{t-1}$\ if
$f_{t-1}\in S_{t}$\ and $f_{t}:=g$\ otherwise. \ Since $\left\vert
S_{t}^{\Diamond}\right\vert $\ can halve at most $k=\log_{2}\left\vert
C\right\vert $\ times, it is clear that after $T\leq k$\ iterations we have
$\left\vert S_{T}^{\Diamond}\right\vert \leq1$. \ Set $f:=f_{T}$\ and
$Z:=Z_{T}$. \ Then by the triangle inequality,%
\[
\Delta_{\infty}\left(  f,f^{\ast}\right)  \left[  Y\right]  \leq T\delta
\leq\frac{\varepsilon}{5},
\]
and also%
\[
\left\vert f\left(  z_{t}\right)  -f_{t}\left(  z_{t}\right)  \right\vert
\leq\left(  T-t\right)  \delta<\frac{\varepsilon}{5}%
\]
for all $t\in\left[  T\right]  $. \ So suppose by contradiction that there
still exists a function\ $g\in S_{T}$\ such that $\Delta_{\infty}\left(
f,g\right)  \left[  Y\cup Z\right]  \leq\delta$\ but $\left\vert f\left(
x\right)  -g\left(  x\right)  \right\vert >3\varepsilon$\ for some $x$, and
consider functions $p,q\in C$\ in the cover such that $\Delta_{\infty}\left(
f,p\right)  \leq\varepsilon$\ and $\Delta_{\infty}\left(  g,q\right)
\leq\varepsilon$. \ Then $p,q\in S_{T}^{\Diamond}$\ but $p\neq q$, which
contradicts the fact that $\left\vert S_{T}^{\Diamond}\right\vert \leq1$.
\ Also notice that for all $g\in S$, if $\Delta_{\infty}\left(  f,g\right)
\left[  Y\cup Z\right]  \leq\delta$\ then $g\in S_{T}$. \ Thus $\Delta
_{\infty}\left(  f,g\right)  \left[  Y\cup Z\right]  \leq\delta$\ implies
$\Delta_{\infty}\left(  f,g\right)  \leq3\varepsilon$ as desired.
\end{proof}

\section{Application to Quantum Advice Classes\label{QADV}}

In this section, we prove Theorem \ref{mainthm1}, as well as several other results.   We will be defining quantum circuits over some fixed universal basis of 2-local unitary and measurement gates. We use $\size(C)$ to denote the number of gates of a classical or quantum circuit (including the input and output gates).

\subsection{Classical Descriptions for Quantum States}

Fix a quantum circuit $Q$ taking an $n$-bit string $x$ and a $p$-qubit state $\rho$ and producing a 1-bit output. \ For a given state $\rho$,
let $f_{\rho}\left(  x\right)  :=\bE\left[  Q(x,\rho)\right]  $.
\ Let $S$\ be the p-concept class consisting of $f_{\rho}$\ for all $p$-qubit mixed states $\rho$. \ Then Aaronson \cite{aar:learn}%
\ proved the following result, which allows us to apply the real-valued majority-certificates lemma to the study of quantum advice.

\begin{theorem}
[\cite{aar:learn}]\label{qlearn}$\operatorname*{fat}\nolimits_{\gamma}\left(
S\right)  =O\left(  p  /\gamma^{2}\right)  .$
\end{theorem}

The next claim gives a useful consequence of Theorem~\ref{qlearn} and the majority-certificates lemma. 


\begin{lemma}\label{yqpplus_general}  Let $Q_n(x, \rho)$ be a quantum circuit taking as input a string $x \in \{0, 1\}^n$ and a quantum state $\rho$ on $p$ qubits, and outputting a single bit.  Fix any $p$-qubit state $\rho^*_n$.

Let $c \geq 1$ be a constant.  For suitably chosen integers $m, k \leq \poly(n, p)$ and a real parameter $\alpha \geq 1/\poly(n, p)$, there exists:
\begin{itemize}
\item a second circuit $Q'_n(x, \sigma)$ of size at most $\poly(\size(Q_n))$ taking as input $x \in\{0, 1\}^n$ and an $m\cdot p$-qubit state $\sigma$;
\item a collection $\mathcal{C}_n = \{C_{(i, j)}(\sigma) \}_{(i , j) \in [m] \times [k]}$ of circuits, each of size $\poly(\size(Q_n))$, and each taking as input a quantum state $\sigma$ on $m \cdot p$ qubits; and,
\item a collection $\{r_{(i, j)}\}_{(i , j) \in [m] \times [k]}$ of rational numbers in $[0, 1]$, each specified by a decimal expansion of length $O(\log (n + p))$.
\end{itemize}
(Here, $Q'_n$ can be uniformly constructed in time $\poly(s, n)$ given a description of $Q_n$, while $\mathcal{C}_n, \{r_{(i, j)}\}$ are non-uniformly chosen.)
We have the following properties:
\begin{enumerate}
\item[(i)] There exists a state $\sigma$ on $m \cdot p$ qubits, of the form $\sigma = \sigma_1 \otimes \ldots \otimes \sigma_m$, that satisfies $|\bE[C_{(i, j)}(\sigma)] - r_{(i, j)}| \leq \alpha$ for each $(i , j) \in [m] \times [k]$;
\item[(ii)] If we are given \emph{any} state $\sigma$ on $m \cdot p$ qubits, satisfying
\[ |\bE[C_{(i, j)}(\sigma) ] -  r_{(i, j)} |  \ \leq \ 4\alpha  \quad{}\quad{} \forall (i , j) \in [m] \times [k] \ , \]
then it also holds that
\[  | \bE[Q'_n(x, \sigma) ] -  \bE[Q_n(x, \rho^*_n)]  |  \ \leq \ n^{-c} \quad\quad{}  \forall x \in \{0, 1\}^n \ .   \]
\end{enumerate}
\end{lemma}

%
%

\begin{proof}
For each $x \in \{0, 1\}^n$ and state $\xi$ on $p$ qubits, let $f_{\xi}\left(  x\right)  :=\bE\left[
Q_n\left(  x,\xi\right)  \right]  $.  Let $S$ be collection $\{f_{\xi}\}$, ranging over all $p$-qubit mixed states
$\xi$. \ Then Theorem \ref{qlearn}\ implies that\ $\operatorname*{fat}%
\nolimits_{\gamma}\left(  S\right)  =O\left(  p  /\gamma
^{2}\right)  $ for all $\gamma>0$. \ Set $\eps := n^{-c}$, $\gamma:= \eps/48$.\ \ By Lemma
\ref{realmajcer}, for some $m, k \leq \poly(n)$, there exist $p$-qubit\ mixed states $\rho_1  ,\ldots,\rho_m$, sets $X_{1},\ldots,X_{m}\subseteq\left\{  0,1\right\}  ^{n}%
$\ each of size $k$, and
an\ $\alpha=\Omega\left(  \frac{1}{\poly(n, p) }\right)  $ for
which the following holds:

\begin{itemize}
\item[(*)] \textit{All collections}
$\sigma_1  ,\ldots,\sigma_m$\textit{\ of $p\left(  n\right)$-qubit states that
satisfy }$\Delta_{\infty}\left(  f_{\rho_i  },f_{\sigma_i  }\right)  \left[  X_{i}\right]  \leq \alpha$\textit{\ for }%
$i\in\left[  m\right]  $\textit{\ also satisfy }$\Delta_{\infty}\left(
f_{\rho^*_{n}},f_{\sigma_{\avg}}\right)  \leq n^{-c}$\textit{, where} $\sigma_{\avg} := \frac{1}{m}(\sigma_1 + \ldots + \sigma_m)$.
\end{itemize}
For an $m \cdot p$-qubit state $\sigma$ and $i \in [m]$, let $\sigma[i]$ denote the reduced state of $\sigma$ on the $i^{th}$ register of $p$ qubits.  Let $x^{(i, j)} \in \{0, 1\}^n$ denote the $j^{th}$ element in $X_i$ (under some fixed ordering).
The circuits $\{ C_{(i, j)}   \}_{(i, j) \in [m] \times [k] }$
are then defined as follows: each $C_{(i, j)}$, on input state $\sigma$, simulates $Q_n\left(  x^{i, j},\sigma[i]\right)  $ (by applying $Q_n(x^{(i, j)}, \cdot)$ to the $i^{th}$ register of $\sigma$) and
outputs the resulting bit. \  The value $r_{i, j}$ is chosen as a rational approximation to the value $\bE[Q_n\left(  x^{(i, j)} ,\rho_i\right)]$, accurate to within $\pm .1\alpha$; this can be achieved with $O(\log (n + p))$ bits of precision, since $\alpha \geq 1/\poly(n, p)$.
Finally, for the circuit $Q_n'(x, \sigma)$, we let $Q_n'$ choose a uniformly random register $i \in [m]$ and simulate $Q_n(x, \sigma[i])$, outputting the result.  All of our efficiency claims for $Q'_n$ and $\{C_{(i, j)}\}_{(i, j) \in [m] \times [k]}$, and our uniform constructibility claim for $Q'_n$, follow from the definitions.

To establish item (i) in the Theorem's conclusion, it is enough to verify that $\sigma := \rho_1 \otimes \ldots \otimes \rho_m$ is a suitable choice of $\sigma$, by our settings to $\{C_{(i, j)}, r_{(i, j)}\}$.  
For item (ii), let the $m \cdot p(n)$-qubit state $\sigma$ satisfy the hypothesis in that item.  By our definitions and the quality of our rational approximations $\{r_{(i, j)}\}$, this implies that $\Delta_{\infty}\left(  f_{\rho_i  },f_{\sigma[i]  }\right)  \left[  X_{i}\right]  \leq \alpha$ for $i \in [m]$. \ Then by (*), we
have\ $\Delta_{\infty}\left(  f_{\rho_{n}},f_{\sigma[\avg]}\right)  \leq n^{-c}$, where we here define $\sigma[\avg] := \frac{1}{m}(\sigma[1] + \ldots + \sigma[m])$.  Also, for our choice of $Q_n'$ we have 
\[ \bE[Q_n'(x, \sigma)] \ = \  \frac{1}{m}\sum_{i \in [m]} \bE[Q_n(x,\sigma[i])] \ = \ \bE[Q_n(x,\sigma[\avg]) ] \ = \ f_{\sigma[\avg]}(x) \ . \]
This gives item (ii), completing the proof of Lemma~\ref{yqpplus_general}.
\end{proof}

\subsection{Advice-Testing Quantum Circuits and Input-Oblivious Testers}

Next we define a class of quantum circuits that will play an important role in our work.

\begin{definition} \label{def:adv_test}  An \emph{advice-testing circuit} (for the input length $n > 0$) is a quantum circuit $Y = Y_n$ with a classical $n$-bit  input register, along with advice and ancilla registers and two designated 1-qubit ``advice-testing'' and ``output'' registers.
On input a string $x \in \{0, 1\}^n$, and with the advice register initialized to some advice state $\rho$, the remaining registers are each initialized to the all-zero state.  $Y$ acts as follows:
\begin{enumerate}
\item First $Y$ applies a subcircuit $A$ to all registers, after which the advice-testing register is measured, producing a value $b_{\adv} \in \{0, 1\}$;
\item Next, $Y$ applies a second subcircuit $B$ to all registers, then measures the output register, producing a value $b_{\out} \in \{0, 1\}$.  \end{enumerate}
If in step 1 above, the subcircuit $A$ ignores the input register, then $Y$ is said to be an \emph{input-oblivious} advice-testing circuit. 
\end{definition}



Next, suppose we have a quantum circuit $Q_n(x, \rho)$ taking a classical string $x\in \{0, 1\}^n$ and a quantum state $\rho$, that we wish to simulate for a specific desired setting $\rho := \rho^*$.  The next result gives a general method to do so by an input-oblivious advice-testing algorithm with polynomial classical advice.   Our use of Lemma~\ref{yqpplus_general} in proving this result draws ideas from the proof of Aharonov and Regev of the equality of complexity classes $\mathsf{QMA}^{+} = \mathsf{QMA}$~\cite{ar:quantum}.

%
%
\begin{theorem}\label{new_generalized}  Let $Q_n(x, \rho)$ be a quantum circuit taking as input a string $x \in \{0, 1\}^n$ and a quantum state $\rho$ on $p \leq s$ qubits, and outputting a single bit.  Fix any $p$-qubit state $\rho^*$, and let $d \geq 1$ be a fixed constant.

Then there exists an input-oblivious advice-testing circuit $Y_n$ of size $\poly(\size(Q_n))$, taking an input $x \in \{0, 1\}^n$ and a $P$-qubit advice state (for some $P \leq \poly(n, p)$), with the following properties:
\begin{enumerate}
\item[(i)] There exists an advice state $\ol{\sigma}^*$ on $P$ qubits such that for all $x \in \{0, 1\}^n$, in the execution of $Y_n(x, \ol{\sigma}^*)$ we have $\Pr[b_{\adv} = 1] \geq 1 - e^{-n}$;

\item[(ii)] For each $n$ and advice state $ \ol{\sigma}$ on $P$ qubits, it holds that in the execution of $Y_n(x, \ol{\sigma})$ (for each $x \in \{0, 1\}^n$) we have
\[  \Pr[b_{\adv} = 1]  \ \geq \ n^{-d}  \quad{} \Longrightarrow \quad{} \left|   \bE[b_{\out} |  b_{\adv} = 1]  \ - \  \bE[Q_n(x, \rho^*) ]    \right|  \ \leq \  n^{-d} \ . \]
\end{enumerate}
\end{theorem}

\begin{proof}[Proof of Theorem~\ref{new_generalized}]  
For $n > 1$, let 
\[ m, \  k, \ \alpha, \ Q'_n ,  \ \mathcal{C}_n, \ \{r_{(i, j)}\}_{i \in [m], j \in [k]} \]
 be as given by Lemma~\ref{yqpplus_general} applied to $Q, \rho^*$, and with $c := 2d$.  
We set $M :=  \lceil 10 n^{8d} mk / \alpha \rceil  , N :=  \lceil 10 \ln M/\alpha^2 \rceil$, and $P := MN m p$.  We regard a $P$-qubit state as having $MN$ registers (indexed by $[M] \times [N]$) of $m \cdot p$ qubits each.  We refer to the register indexed by $(s, t) \in [M] \times [N]$ as the ``$(s, t)^{th}$ proof register.''




The subroutine $A$ for $Y_n$ is defined as follows:  

\vspace{1 em}

\textbf{Algorithm} A$(\ol{\sigma}, y)$:
\begin{enumerate}
\item Set $b_{\adv} := 1$, and choose $S \in [M]$ uniformly;
\item For $s = 1, 2, \ldots, (S - 1)$:
\begin{itemize}
\item[2.a.] Choose $(i(s), j(s)) \in [m] \times [k]$ uniformly;
\item[2.b.] Apply $C_{(i(s), j(s))}$ successively to the proof registers $(s, 1), (s, 2), \ldots, (s, N)$, and let $\hat{r}_{s} \in [0, 1]$ be the fraction of these computations that accept;
\item[2.c] If $| \hat{r}_{s}   -    r_{(i(s), j(s))}  | > .5 \alpha$, set $b_{\adv} := 0$.
\end{itemize}
\end{enumerate}

Note that in step (2.b), the joint state on the proof registers may change after each application of $C_{(i(s), j(s))}$.  If $S = 1$, the proof registers go untouched and $b_{\adv} = 1$.

Next, the subroutine $B$ acts as follows.  $B$ measures the value $S$ chosen by $A$ (and stored in the ancilla register).  It then chooses $t \in [N]$ uniformly and simulates $Q'_n$ applied to input $x$ and with the $(S, t)^{th}$ proof register as the quantum advice state for $Q'_n$, taking the resulting bit as $b_{\out}$.


$Y_n$ can clearly be implemented in size $\poly(\size(Q_n))$.  Now let us analyze $Y_n$ to establish items (i)-(ii) in the Theorem's conclusion.  For item (i), consider the execution $Y_n(x, \ol{\sigma})$ on the advice state $\ol{\sigma}$ which is the tensor product of $MN$ independent copies of the state $\sigma$ guaranteed to exist by item (i) in our application of Lemma~\ref{yqpplus_general}.  Then in the operation of the subroutine $A$, for each execution of step (2.b) (indexed by an $s \in [M]$), the expected fraction $\bE[\hat{r}_{s}]$ is within $\pm .1\alpha$ of $r_{i(s), j(s)}$ after conditioning on $i(s), j(s)$.  Also, the outcome of the executions of $C_{i(s), j(s)}$ are mutually independent, since $\ol{\sigma}$ is a product state over the $MN$ registers.
Chernoff bounds and our setting of $N$ then imply that $\hat{r}_{s}$ is within $\pm .5\alpha$ of  $r_{i(s), j(s)}$ with probability $> 1 - e^{-n}/M$.  A union bound over all $s \in [M]$ completes the proof of item (i) in the Theorem.


We now turn to item (ii).  Let $\ol{\sigma}$ be any $P$-qubit state for which, in the execution of $Y_n(x, \ol{\sigma})$, we have $\bE[b_{\adv}] \geq n^{-d}$.  (If this holds for some $x \in \{0, 1\}^n$ then it holds for all such $x$; we fix some such $x$ in what follows.) For $s \in [M - 1]$, let $q_s$ denote the probability that  $| \hat{r}_{s}   -    r_{(i(s), j(s))}  | \leq .5 \alpha$ holds in the execution of subroutine $A$ in the operation of $Y_n(x, \ol{\sigma})$, \emph{conditioned} on the following two events:
 \begin{enumerate}
 \item $S = s + 1$, so that the For loop in Step 2 of $A$ executes for the value $s$;
 \item $| \hat{r}_{s'}   -    r_{(i(s'), j(s'))}  | \leq .5 \alpha$ for all $s' < s$.
 \end{enumerate}
 Note that the value $q_s$ would be unchanged if in the first item above we instead conditioned on $[S = s'']$, for any $s'' > s$.  Also, for future use we define $\overline{\sigma}^{(s)}$ as the $N m p$-qubit reduced state on the proof registers $(s, 1), (s, 2), \ldots, (s, t)$, conditioned on items 1 and 2 above.

 
Let $I_{\bad} \subseteq [M - 1]$ be the set of indices $s$ for which $q_s < 1 -   \alpha/(n^{3d} m k)$.  We will upper-bound $\Pr[S \in I_{\bad} \wedge b_{\adv} = 1]$.  
Let $I^{\early}_{\bad}$ be the first $W := \lceil n^{4d} m k / \alpha \rceil$ elements of $I_{\bad}$ in increasing order (or if $I_{\bad} \leq W$, then $I^{\early}_{\bad} := I_{\bad}$).  Let $I^{\late}_{\bad} := I_{\bad} \setminus I_{\bad}^{\early}$.  We have
\begin{equation*}\label{eq:split} \Pr[S \in I_{\bad} \wedge b_{\adv} = 1] \ \leq \ W/M +   \Pr[S \in I^{\late}_{\bad} \wedge b_{\adv} = 1]  \ ,  \end{equation*}
since $\Pr[S \in I^{\early}_{\bad}]   \leq W/M$.
If $I_{\bad}^{\late} \neq \emptyset$, then conditioned on any value of $S$ with $S > \max(I^{\early}_{\bad})$, the probability that $b_{\adv}$ is not set to 0 in the $S - 1$ executions of step 2 of $A$ equals 
\[ \prod_{s < S}q_s  \ \leq  \ \prod_{s \in I^{\early}_{\bad}}q_s \ \leq  \ (1 - \alpha/(n^{3d} m k))^{W}  \ \leq  \ n^{-4d} \ . \]  Thus, $\Pr[S \in I^{\late}_{\bad} \wedge b_{\adv} = 1] \leq n^{-4d}$, and $\Pr[S \in I_{\bad} \wedge b_{\adv} = 1] \leq n^{-4d} + W/M$; this is at most $2n^{-4d}$, by our setting to $M$.  It follows that 
\[ \Pr[S \in I_{\bad} |  b_{\adv} = 1]  \ \leq \  \frac{2 n^{-4d}}{\Pr[b_{\adv} = 1]}  \ \leq  \ 2 n^{-3d}\ ,\] using our assumption in item (ii) that $\Pr[b_{\adv} = 1] \geq n^{-d}$.

Next, we claim that for each $s \in [M] \setminus I_{\bad}$, the conditional expectation $ \bE[ b_{\out} = 1 | S = s \wedge b_{\adv} = 1] $ satisfies
\[   \left|  \bE[ b_{\out} = 1 | S = s \wedge b_{\adv} = 1]    - \  \bE[Q_n(x, \rho^*_n) ]    \right|  \ \leq \  n^{-3d} \ .  \]
To see this, fix any such $s$.  First note that, if we condition on $[S = s \wedge b_{\adv} = 1]$, the joint post-conditioned state of the proof registers $(s, 1), (s, 2), \ldots, (s, t)$ is precisely $\overline{\sigma}^{(s)}$ as defined previously.  Now consider the experiment in which we choose a pair $(i, j)$ uniformly from $[m] \times [k]$ and apply $C_{(i, j)}$ to each of these proof registers, prepared in the joint state $\overline{\sigma}^{(s)}$, and let $\hat{r}_{(i, j)} \in [0, 1]$ be the fraction of 1s measured.  The probability in this experiment that $|\hat{r}_{(i, j)} - r_{(i, j)}| \leq .5 \alpha$ is, by the linearity of quantum mechanics, equal to $q_s$; this is greater than $1 -  \alpha/(n^{3d} m k)$ since $s \notin I_\bad$.  Then by an application of Markov's inequality, for \emph{every} $(i^*, j^*) \in [m] \times [k]$, if we perform this experiment on $\overline{\sigma}^{(s)}$ with the fixed choice $(i, j) = (i^*, j^*)$, then we see $|\hat{r}_{(i^*, j^*)} - r_{(i^*, j^*)}| \leq .5 \alpha$ with probability greater than $1 -  \alpha n^{-3d} > 1 - .2\alpha$.  Thus $|\bE[\hat{r}_{(i^*, j^*)} ] - r_{(i^*, j^*)}| \leq .7 \alpha$.

For $t \in [N]$, let $\sigma^{(s, t)}$ denote the reduced state of $\overline{\sigma}^{(s)}$ on the $(s, t)$ proof register.  Let $\sigma^{(s, \avg)} := \frac{1}{N}\sum_{t \in [N]} \sigma^{(s, t)}$, and note that in the experiment above with fixed pair $(i^*, j^*)$, we have $\bE[\hat{r}_{(i^*, j^*)}] = \bE[ C_{(i^*, j^*)}(\sigma^{(s, \avg)}) ]$.  By our work above, $|\bE[ C_{(i^*, j^*)}(\sigma^{(s, \avg)}) ] -  r_{(i^*, j^*)} | \leq .7 \alpha$.  As $(i^*, j^*)$ was arbitrary, it follows from item (ii) in our application of Lemma~\ref{yqpplus_general} that
\[    \left|  \bE[  Q_n'(x, \sigma^{(s, \avg)}) ]    -  \bE[  Q_n(x, \rho^*_n) ] \right|   \ \leq \ n^{-2d} \ .   \]
Now let us return to the definition of the algorithm $Y_n$ and note that, in the execution $Y_n(x, \ol{\sigma})$, if we condition on $[b_{\adv} = 1 \wedge S = s]$, then $Y_n$ simulates $Q'_n$ applied to $x$ and to an advice state whose density operator is (under our conditioning) precisely that of $\sigma^{(s, \avg)}$,  and $Y_n$ outputs the resulting bit.  
Thus, $\left| \bE[b_{\out}| b_{\adv} = 1 \wedge S = s ] -  \bE[  Q_n(x, \rho^*_n) ] \right|    \leq  n^{-2d}$, and since $s$ was an arbitrary element of $[M] \setminus I_{\bad}$, we also have $\left| \bE[b_{\out}| b_{\adv} = 1 \wedge S \notin I_{\bad} ] -  \bE[  Q_n(x, \rho^*_n) ] \right|    \leq  n^{-2d}$.  Combining our findings, we see that
\begin{align*}
\left|\bE[b_{\out}| b_{\adv} = 1]  -    \bE[  Q_n(x, \rho^*_n) ]   \right| \ &\leq \  \Pr[S \in I_{\bad}| b_{\adv} = 1]  +   n^{-2d} \\
&\leq\   2n^{-3d}  + n^{-2d} \\
&\leq \ n^{-d} \ ,
\end{align*}
for $n > 1$.  The statement of item (ii) is trivial for $n = 1$, so this proves item (ii), completing the proof of the Theorem.
\end{proof}

\subsection{Bestiary of Quantum Complexity Classes\label{QCC}}

In this section we define some old and new complexity classes which our techniques shed light on.
Given a language $L\subseteq\left\{  0,1\right\}  ^{\ast}$, let $L:\left\{
0,1\right\}  ^{\ast}\rightarrow\left\{  0,1\right\}  $ be the characteristic
function of $L$. \ We now give a formal definition of the class
$\mathsf{BQP/qpoly}$.

\begin{definition}\label{def:bqppoly}
A language $L$\ is in\ $\mathsf{BQP/qpoly}$\ if there exists a polynomial-time
quantum algorithm $A$ and polynomial-time computable function $p(n) \leq \poly(n)$ such that for all $n$, there exists
an advice state $\rho_{n}$\ on $p\left(  n\right)  $\ qubits such that
$A\left(  x,\rho_{n}\right)  $\ outputs $L\left(  x\right)  $\ with
probability $\geq2/3$\ for all $x\in\left\{  0,1\right\}  ^{n}$.
\end{definition}

Closely related to quantum advice are \textit{quantum proofs}. \ We now recall
the definition of $\mathsf{QMA}$\ (Quantum Merlin-Arthur), a quantum version
of $\mathsf{NP}$.

\begin{definition}
A language $L$\ is in\ $\mathsf{QMA}$\ if there exists a polynomial-time
quantum algorithm $A$ and polynomial-time computable function $p(n) \leq \poly(n)$ such that for all $x\in\left\{
0,1\right\}  ^{n}$:

\begin{enumerate}
\item[(i)] If $x\in L$\ then there exists a witness $\rho_{x}$\ on $p\left(
n\right)  $\ qubits such that $A\left(  x,\rho_{x}\right)  $\ accepts with
probability $\geq2/3$.

\item[(ii)] If $x\notin L$\ then $A\left(  x,\rho\right)  $\ accepts with
probability $\leq1/3$ for all $\rho$.
\end{enumerate}
\end{definition}

We will define some complexity classes involving untrusted (classical or quantum) advice that depends only on the input length. \ This
notion has been studied before: Chakaravarthy and Roy \cite{chakaravarthy}%
\ and Fortnow, Santhanam, and Williams \cite{fsw}\ defined the complexity
class $\mathsf{ONP}$\ (``Oblivious $\mathsf{NP}$''),\ which is like
$\mathsf{NP}$\ except that the witness can depend only on the input length.
\ Independently, Aaronson \cite{aar:learn}\ defined the complexity class
$\mathsf{YP}$,\footnote{$\mathsf{YP}$\ stands for ``Yoda Polynomial-Time,''\ a
nomenclature that seems to make neither more nor less sense than
``Arthur-Merlin.''} which is easily seen to equal $\mathsf{ONP}\cap
\mathsf{coONP}$. \ We will adopt the ``$\mathsf{Y}$'' notation in this paper.

We now give a formal definition of $\mathsf{YP}$, as well as a slight variant
called $\mathsf{YP}^{\mathsf{\ast}}$.

\begin{definition}
A language $L$\ is in\ $\mathsf{YP}$\ if there exist polynomial-time
algorithms $A,B$ and a polynomial-time computable function $p(n) \leq \poly(n)$ such that:

\begin{enumerate}
\item[(i)] For all $n$, there exists an advice string $y_{n}\in\left\{
0,1\right\}  ^{p\left(  n\right)  }$\ such that $A\left(  x,y_{n}\right)
=1$\ for all $x\in\left\{  0,1\right\}  ^{n}$.

\item[(ii)] If $A\left(  x,y\right)  =1$, then $B\left(  x,y\right)  =L\left(
x\right)  $.
\end{enumerate}

$L$ is in $\mathsf{YP}^{\mathsf{\ast}}$\ if moreover $A$ ignores $x$,
depending only on $y$.
\end{definition}

Clearly $\mathsf{P}\subseteq\mathsf{YP}^{\mathsf{\ast}}\subseteq
\mathsf{YP}\subseteq\mathsf{P/poly}\cap\mathsf{NP}\cap\mathsf{coNP}$. \ Also,
Aaronson \cite{aar:learn}\ showed that $\mathsf{ZPP}\subseteq\mathsf{YP}%
$.\ \ We will be primarily interested in a quantum analogue of $\mathsf{YP}^{\mathsf{\ast}}$.  This analogue builds on Definition~\ref{def:adv_test}.  However, it also models a distinctively quantum ingredient: we consider two-phase protocols in which an untrusted quantum advice state is first tested in an input-oblivious fashion and, if accepted, is passed along \emph{in altered form} to be used in computation with the given input.  This model is natural, since quantum measurements unavoidably alter the measured states; the alterations performed by the initial testing are also crucial to the power of these protocols. (Roughly speaking, this works as follows: if the given quantum advice state is a mixture $\rho = t \rho_{1} + (1 - t) \rho_{2}$ of a ``good state'' $\rho_{1}$ which passes our test with high probability and is useful for computation, and a ``bad state'' $\rho_{2}$ which is useless for computation and passes the test with low probability, then conditioning on passing the test ``filters out'' the contribution of $\rho_2$, making the resulting state more useful.\footnote{Conversely, if our testing procedure did \emph{not} alter the advice state, as per our definitions in previous drafts (which essentially assumed the availability of two identical, independent copies of the state---one for testing and one for computation), and if $t = .5$, say, then $\rho$ as above will pass the test with probability close to .5, but $\rho$ cannot be useful for computation with correctness guarantee close to 1, due to the continuing presence of the useless $\rho_2$ in equal mixture with $\rho_1$.  This weakness necessitated the change in definitions.}  We emphasize, however, that the test involves various measurements that significantly alter even a state that passes with high probability.  The technical core of this procedure has already been given in Theorem~\ref{new_generalized}.)

\begin{definition}[$\mathsf{YQP}$ and $\mathsf{YQP}^{\ast}$]\label{def:yqp}
A language $L$\ is in\ $\mathsf{YQP}$\ if there exists a uniform (i.e., polynomial-time constructible) family of advice-testing quantum circuits $\{Y_n(x, \rho)\}_{n > 0}$ (as per Definition~\ref{def:adv_test}). Each $Y_n$ is of size $\poly(n)$ and takes as input an $x \in \{0, 1\}^n$ and a $p(n)$-qubit state $\rho$ (for some $p(n) \leq \poly(n)$).  We have the following properties:

\begin{enumerate}
\item[(i)] For all $n$, there exists a setting $\rho_{n}$ to the quantum advice register such that for any $x \in \{0,1\}^n$, in the execution of $Y$ on $(x, \rho_n)$ we have $\bE[b_{\adv } ] \geq 9/10$.

\item[(ii)] If for any settings $(x, \rho)$ to the input and advice registers we have $\bE[b_{\adv}] \geq 1/10$,
then $\Pr[b_{\out} = L(x) | b_{\adv} = 1] \geq 9/10$.
\end{enumerate}

$L$ is in $\mathsf{YQP}^{\ast}$ if the circuit family $\{Y_n\}_{n > 0}$ can be additionally be chosen to obey the input-oblivious property.

We define the corresponding non-uniform classes $\mathsf{YQP}\mathsf{/poly}, \mathsf{YQP}^{\mathsf{\ast}}\mathsf{/poly}$ by removing the requirement that the family $\{Y_n\}_{n > 0}$ be uniform.
\end{definition}

Clearly $\mathsf{BQP}\subseteq\mathsf{YQP}^{\mathsf{\ast}}\subseteq\mathsf{YQP}\subseteq \mathsf{BQP/qpoly}\cap\mathsf{QMA}\cap\mathsf{coQMA}$.

\subsection{Characterizing Quantum Advice\label{YQPSEC}}

We now prove the following characterization of $\mathsf{BQP/qpoly}$, which
immediately implies (and strengthens) Theorem \ref{mainthm1}:

\begin{theorem}
\label{yqpthm}$\mathsf{BQP/qpoly}=\mathsf{YQP^{\ast}/poly}$.
\end{theorem}

\begin{proof}
One direction ($\mathsf{YQP^{\ast}/poly}\subseteq\mathsf{BQP/qpoly}$) is obvious,
since untrusted quantum advice and trusted classical advice can both be
simulated by trusted quantum advice. \ We prove that $\mathsf{BQP/qpoly}%
\subseteq\mathsf{YQP^{\ast}/poly}$.  Let $L \in \mathsf{BQP/qpoly}$, and let $Q(x, \rho), \{\rho_n^*\}_{n > 0}$ be a polynomial-time quantum algorithm (given by a uniform circuit family $\{Q_n\}_{n > 0}$ for input length $n$) and polynomial-size quantum advice family defining $L$.  We insist that $Q$ enjoy completeness and soundness parameters $(99/100, 1/100)$ in place of $2/3, 1/3$ in Definition~\ref{def:bqppoly}; this can be achieved by standard soundness amplification by providing multiple copies of the trusted advice state.  We apply Theorem~\ref{new_generalized} to $Q_n(x, \rho)$ and $\{\rho_n^*\}_{n > 0}$ with $d := 1$, for each $n $.  We obtain a (non-uniform) family of input-oblivious advice-testing quantum circuits $\{Y_n\}_{n > 0}$, such that:
\begin{enumerate}
\item[(i)] For each $n$, there is a state $\sigma$ such that in the execution of $Y_n(x, \sigma)$ we have $\Pr[b_{\adv} = 1] \geq 1 - e^{-n}$;
\item[(ii)]  For any $n > 1$ and advice state $\sigma$, it holds that for each $x \in \{0, 1\}^n$, in the execution of $Y_n(x, \sigma)$,
\[ \Pr[b_{\adv} = 1] \geq n^{-1} \ \quad{} \Longrightarrow  \quad{}  \left|   \bE[b_{\out} |  b_{\adv} = 1]  \ - \  \bE[Q_n(x, \rho^*_n) ]    \right|  \ \leq \  n^{-1} \ . \]
\end{enumerate}
Now by the definitions of $Q_n$ and $\rho^*_n$, we have $|\bE[Q_n(x, \rho^*_n)] - L(x)| \leq 1/100$ for all $x \in \{0, 1\}^n$.  Thus, if $n$ is sufficiently large, we have
\begin{enumerate}
\item[(iii)]  For any advice state $\sigma$ for length $n$, it holds that for each $x \in \{0, 1\}^n$, in the execution of $Y_n(x, \sigma)$, if $\Pr[b_{\adv} = 1] \geq 1/10$, then we have 
\[ \left|   \bE[b_{\out} |  b_{\adv} = 1]  \ - \  L(x)     \right|  \ \leq \  n^{-1}  +  1/100   \ \leq \ 1/10 \ . \]
\end{enumerate}
Thus the family $\{Y_n\}_{n > 0}$ witnesses that $L \in \mathsf{YQP^{\ast}/poly}$.  This proves Theorem~\ref{yqpthm}.
\end{proof}

One interesting consequence of Theorem~\ref{yqpthm} is that $\mathsf{YQP/poly} = \mathsf{YQP^{\ast}/poly}$.  We do not know of an easier proof of this equality, and we leave as an open question whether, in the uniform setting, the corresponding equality $\mathsf{YQP} = \mathsf{YQP}^{\ast}$ holds.

Since we never critically used the assumption that
the $\mathsf{BQP/qpoly}$\ machine\ computes a \textit{language} (i.e., a total
Boolean function), a strengthening of Theorem~\ref{yqpthm} we can easily observe is the promise-class equality
$\mathsf{P{}romiseBQP/qpoly}=\mathsf{P{}romiseYQP^{\ast}/poly} = \mathsf{P{}romiseYQP/poly}$.

\subsection{Application to Quantum Communication}

We can also use our Theorem~\ref{new_generalized} to obtain a new positive result about the possibility of robust communication over fault-prone \emph{quantum communication channels} (augmented with a trustworthy classical channel).  Our result does not assume any particular error model for quantum channels.  Rather, it asserts that a successful outcome is achieved by the protocol under a perfect transmission, and that the protocol guards against a certain type of bad outcome under \emph{any} corruption of the transmitted quantum state.

\begin{theorem}
\label{channelthm}Suppose that Alice, who is computationally unbounded, has a
classical description of an $N$-qubit quantum state $\rho^*$. \ She wants to
send $\rho^*$\ to Bob, who is computationally bounded. \ Assume that Alice
has at her disposal a noiseless one-way classical channel to Bob, as well as a
noisy one-way quantum channel.  Bob holds a binary measurement $E$ for which he wishes to learn $\bE[E(\rho^*)]$ to within an accuracy $\eps > 0$.  We assume $E$ is implemented by a circuit with at most $m$ gates (under some fixed finite basis); here $m$ is known to Alice, but $E$ is known only to Bob.

Then for all $\varepsilon>0$, there
exists a protocol whereby

\begin{itemize}
\item Alice sends Bob a classical string $z$\ of $\operatorname*{poly}%
\left(  N,m,1/\varepsilon\right)  $\ bits, as well as a state $\sigma$\ of
$\operatorname*{poly}\left(  N,m,1/\varepsilon\right)  $\ qubits;

\item Bob receives $z$\ together with a possibly-corrupted version
$\widetilde{\sigma}$\ of $\sigma$, and performs a (non-binary) measurement $f_{z}\left(  E\right)  $\ on $\widetilde{\sigma}$, outputting a real value $\beta \in [0, 1]$ along with a ``success bit'' $b_{\suc} \in \{0, 1\}$.  This $f_{z}\left(  E\right)$ can be computed and performed in $\operatorname*{poly}\left(  N ,m,1/\varepsilon\right)  $ steps, given $z$\ together with a description of $E$.
\end{itemize}

The following properties hold:
\begin{enumerate}
\item[(i)] If $\widetilde{\sigma}=\sigma$, then with probability greater than $1 - 2^{- N}$ we have $\left\vert
 \beta - \bE[E(\rho^*)]  \right\vert \leq\varepsilon$
and $b_{\suc} = 1$; \ 
\item[(ii)] For every $\widetilde{\sigma}$ and every measurement $E$ as described above,
with probability at least $1-2^{-N}$, Bob either sets $b_{\suc} = 0$, or 
outputs a $\beta\in\left[  0,1\right]  $ such that $\left\vert \beta - \bE[E(\rho)]   \right\vert \leq\varepsilon$.
\end{enumerate}
\end{theorem}

\begin{proof}
We will apply Theorem\ \ref{new_generalized}\ to the
communication setting. The string $z$\ plays the role of the trusted
classical advice; the state $\widetilde{\sigma}$\ plays the role of the
untrusted quantum advice; the measurement $E$ plays the role of the input
$x$; Bob plays the role of the advice-testing algorithm $Y$.  We will perform multiple trials to increase our confidence.

We prove the result under the assumption that $\eps$ is at least inverse-polynomial in $N$, which allows us to apply our prior work more directly.  We will assume that $\eps \geq N^{-1}$; the general result will follow, since in our construction we may begin by padding the quantum register with $1/\eps$ dummy qubits.  The protocol will succeed for sufficiently large $N$---smaller values of $N$ can be handled by brute force.

Let $n > 0$ be a fixed description length adequate to describe any $m$-gate measurement $E$ that may be held by Bob in our communication scenario, for our specific values of interest $m, N$; here we can take $N \leq n \leq \poly(N, m)$.  Let $Q_n(E, \xi)$ be a quantum circuit which receives a description of a binary measurement $E$ of description length $n$, described by a circuit in our fixed finite basis.  $Q_n$ also receives a quantum state $\xi$ on $N$ qubits, and outputs the result of $E(\xi)$.  This $Q_n$ can be implemented in size $\poly(n, N) \leq \poly (N ,m)$.
Let $Y_n = Y_n(E, \ol{\sigma })$ be the input-oblivious advice-testing circuit of size $\poly(N, m)$ given by Theorem~\ref{new_generalized} for $(Q_n, \rho^*, d := 2)$.  

In our protocol, Alice sends a description of $Y_n$ as the reliable classical message $z$ to Bob, and for the fault-prone quantum state $\sigma$, Alice sends $T := n^4$ independent copies of the $P$-qubit advice state $\ol{\sigma}^*$ guaranteed to exist by item (i) of Theorem~\ref{new_generalized}; we have $|z| \leq \poly(N, m)$ and $\sigma$ is on $\poly(N, m)$ qubits, as needed.

Bob receives the (correct) string $z$, and a quantum state $\tilde{\sigma}$ on $T \cdot P$ qubits, where we consider this state to be defined over $T $ registers called the ``transmission registers.''  Bob acts as follows (these steps define the measurement $f_z(E)$): For $i = 1, 2, \ldots, T$, Bob executes $Y_n$ applied to input bitstring $E$, classical advice $z$, and with the $i^{th}$ transmission register used as the quantum advice state.  For each such application of $Y_n$ in turn, Bob measures the bits $b_{\adv, i}, b_{\out, i}$ (here, we use $b_{\adv, i}$ to denote the value of $b_{\adv}$ on the $i^{th}$ trial, and similarly for $b_{\out, i}$).  If $b_{\adv, i} = 0$ for any $i$, Bob sets $b_{\suc} := 0$ (and sets $\beta := 0$, say).  Otherwise, Bob sets $b_{\suc} := 1$ and outputs the value $\beta := \frac{1}{T}\sum_{i \in [T]} b_{\out, i}$.

Let us analyze this procedure.  First note that when Bob receives the same state $\ol{\sigma}^*$ sent by Alice, item (i) of Theorem~\ref{new_generalized} tells us that each $b_{\adv, i}$ equals 1 with probability at least $1 - e^{-n}$.  Then by a union bound over all $i$, for sufficiently large $N$, each of these bits equals 1 with probability at least $1 - 2^{-(n +1)}$.  So $\Pr[b_{\suc} = 1] \geq 1 - 2^{-(n+1)}$.  Also, item (ii) of Theorem~\ref{new_generalized} tells us that each $b_{\out, i}$ satisfies $|\bE[b_{\out, i}] -  \bE[E(\rho^*)] | = |\bE[b_{\out, i}] -  Q_n(E, \rho^*) |   \leq  n^{-2}$, and these bits are independent.  By Chernoff's bound, $\Pr[|\beta -   \bE[E(\rho^*)] | \leq n^{-1}] \geq 1 - 2^{-(n + 1)}$ for large $n$.  A union bound completes the proof of item (i) in the Theorem's statement.

For item (ii), consider any quantum state $\tilde{\sigma}$ on $T \cdot P$ qubits received by Bob.  Each execution of Bob's algorithm determines, for each $i \in [T]$, a mixed state $\xi_i$ on $P$ qubits that describes the reduced state on the $i^{th}$ transmission register, immediately after Bob has applied $Y_n$ to the first $(i - 1)$ transmission registers and measured $b_{\adv, 1}, b_{\out, 1}, \ldots, b_{\adv, i - 1}, b_{\out, i - 1}$.  We consider $\xi_i$ as a random variable determined by Bob's execution (acting on the pair $z, \tilde{\sigma}$).

Say that state $\xi$ on $P$ qubits is \emph{good}, if in the execution of $Y_n(x, \xi)$, we have $\Pr[b_{\adv} = 1]  \ \geq \ n^{-2}$.  Let $G \subseteq [T]$ be the (random) set $\{i: \xi_i$ is good$\}$.  Conditioned on any outcomes $b_{\adv, 1}, b_{\out, 1}, \ldots, b_{\adv, i - 1}, b_{\out, i - 1}$ which determine a state $\xi_i$ which is good, item (ii) of Theorem~\ref{new_generalized} tells us that the expected value of $b_{\out, i}$, \emph{conditioned} on $[b_{\adv, i} = 1]$, is within $\pm n^{-2}$ of $\bE[Q(E, \rho^*)] =  \bE[E(\rho^*)]$.

For $i \in [T]$, let the random variable $Z_i \in \{0, 1\}$ be defined by 
\[ Z_i := 
\begin{cases} b_{\out, i} & \text{if $i \in G$ and $b_{\adv, i} = 1$,}
\\
\text{an independent coin flip with bias $\bE[E(\rho^*)]$}  &\text{otherwise.}
\end{cases}
\]
Note that we have the relation $\left|\bE[Z_i | Z_1, \ldots, Z_{i - 1}] - \bE[E(\rho^*)] \right| \leq n^{-2}$.  By an application of Azuma's inequality,
\[  \Pr\left[  \left|  \frac{1}{T}\sum_{i \in T} Z_i  -    \bE[E(\rho^*)]  \right|   \geq  n^{-2} + .5n^{-1}    \right]  \ \leq \  \exp\left( - \Omega(  (.5 n^{-1})^2 \cdot T )\right) \ \leq \     e^{-n} \ , \]
for sufficiently large $N$.



Now, it is clear that $\Pr[|[T]\setminus G| > n \ \wedge \ b_{\suc} = 1] \leq (n^{-2})^n < e^{-n}$.  If $|[T]\setminus G| \leq n$ and $b_{\suc}= 1$, then we have $b_{\adv, i} = 1$ for all $i$ so that $\left|\frac{1}{T}\sum_{i \in T} Z_i  -\frac{1}{T}\sum_{i \in T}b_{\out, i}\right| \leq n/T$.
Combining this with our previous work, it follows that 
\[    \Pr\left[   b_{\suc} = 1 \  \wedge \   \left|  \frac{1}{T}\sum_{i \in T} b_{\out, i} -  \bE[E(\rho^*)]  \right|   \geq  \left(  n^{-2} + .5 n^{-1}  \right) +  n/T    \right]  \ \leq \  2 \cdot e^{-n} \ \leq 2^{-n} \ ,  \]
for large $N$; for such $N$ we have $n^{2} + .5n^{-1} + n/T  \leq n^{-1}$.  As $n^{-1} \leq N^{-1} \leq  \eps$, this gives item (ii).
 \end{proof}


\section{Local Hamiltonians and the Complexity of Preparing Quantum Advice States\label{LOCALHAM}}

In this section we begin the proof of Theorem \ref{mainthm0} from the Introduction, which we will obtain from a slightly more general result.

Let $\mathcal{B}^{\otimes N}$ denote the $2^N$-dimensional complex Hilbert space whose unit ball consists of the $N$-qubit pure quantum states.  Recall that a \emph{Hamiltonian} on $N$-qubit states is a Hermitian operator $H: \mathcal{B}^{\otimes N} \rightarrow \mathcal{B}^{\otimes N}$.  (We will only discuss the action of Hamiltonians on pure states.)  $H$ is called a \emph{$k$-local Hamiltonian} if it can be written as $H = \sum_{i = 1}^{s} H_i$, where each $H_i$ is a Hermitian operator acting on at most $k$ qubits.

If we combine Theorem \ref{new_generalized}\ with known $\mathsf{QMA}$-completeness
reductions (and some further analysis of these reductions), we can obtain a striking consequence for quantum complexity theory.
\ Namely, \textit{the preparation of quantum advice states can always be
reduced to the preparation of ground states of 2-local Hamiltonians}---despite
the fact that quantum advice states involve an exponential number of
constraints, while ground states of local Hamiltonians involve only a
polynomial number.\textit{ \ }(In particular, if ground states of local
Hamiltonians can be prepared by polynomial-size circuits, then we have not
only\textit{ }$\mathsf{QMA}=\mathsf{QCMA}$, but also $\mathsf{BQP/qpoly}%
=\mathsf{BQP/poly}$.) \ 
Our objective in Sections~\ref{sec:5local} and~\ref{sec:2local} is to prove the following result:

\begin{theorem}
\label{localthm}Let $C^{*}(z, \rho)$ be a quantum circuit of $T$ gates (each 2-local) taking an input string $z \in \{0, 1\}^N$ and a quantum state $\rho$ on $\ell$ qubits (we may assume $\ell \leq 2T$).  Let $\rho^*$ be a distinguished state on $\ell$ qubits.  For all $\delta>0$, there exists a second quantum
circuit $C^{\prime}$ and a $2$-local Hamiltonian $H$ acting on
$\ell' \leq \poly\left( T,  N,1/\delta\right)  $ qubits, such that for
any ground state $\ket{\psi}$\ of $H$ and any input
$z\in\left\{  0,1\right\}^{N}$,%
\[
\left\vert \bE\left[  C^{\prime}(z, \ketbra{\psi}{\psi})\right]  -\bE\left[  C^*(z,\rho^*)\right]  \right\vert
\leq \ \delta \ .
\]
While a description of $H$ may not be efficiently computable, $C^{\prime}$\ can be constructed in (classical, deterministic) time $\poly(T, N, 1/\delta)$, given $\delta$ and descriptions of $C^*$ and $H$.
\end{theorem}

Our proof of Theorem~\ref{localthm} combines Theorem~\ref{new_generalized} with the following result on the expressive power of ground states of 2-local Hamiltonians.

\begin{theorem}\label{thm:new_key}
Let $V(\xi)$ be a quantum ``verifier'' circuit of $T$ gates (each 2-local), which acts on an $m$-qubit quantum state $\xi$ and an ancilla register of $N - m$ qubits (we may assume $N \leq 2T$), with the ancilla register initially in the all-zero state.  Suppose that $V$ defines a binary measurement on $\xi$.  Fix any $\eps > 0$, and assume that $\max_{\rho}\bE[V(\rho)] \geq 1 - \eps$. Then there exists
\begin{itemize}
\item A 2-local Hamiltonian $H_{V, \eps}$ acting on $N'$-qubit states, for some $N' \leq \poly(T, 1/\eps)$, expressed as a sum of 2-local terms $H_i$ with operator norm $\frac{1}{\poly(T, 1/\eps)}  \leq ||H_i|| \leq \poly(T, 1/\eps)$; and
\item A quantum operation $R_{V, \eps}$ mapping $N'$-qubit states to $m$-qubit states,\footnote{The state output by $R_{V, \eps}$ may be mixed, even if its input state is pure.} implemented by a quantum circuit with $\poly(T, 1/\eps)$ gates,
\end{itemize}
for which the following property holds: if $\ketbra{\psi}{\psi}$ is any ground state of $H_{V, \eps}$, then for $\xi := R_{V, \eps}(\ketbra{\psi}{\psi})$ we have
\[   \bE[ V(\xi) ] \ \geq \ 1 -   \kappa \cdot T^{\kappa }\eps^{1/\kappa} \ , \]
where $\kappa > 1$ is an absolute constant.
Furthermore, $H_{V, \eps}$ and $R_{V, \eps}$ can be constructed in (classical, deterministic) time $\poly(T, 1/\eps)$, given a description of $V$.
\end{theorem}

We will obtain Theorem~\ref{thm:new_key} by a detailed analysis of known $\mathsf{QMA}$-completeness reductions.  We defer the proof.

Theorem~\ref{mainthm0} is now easily obtained:

\begin{proof}[Proof of Theorem~\ref{mainthm0}]
Define a circuit $C^*(E, \rho)$ which takes as input a circuit $E$ of size $n^c$ defining a binary measurement, and a quantum state $\rho$ on $n$ qubits, and executes $C(\rho)$.  The circuit $C^*$ can be implemented in size $\poly(n)$ using 2-local gates, and we have $\bE[C^*(E, \rho)] = \bE[E(\rho)]$ for all inputs $(E, \rho)$ to $C^*$.  The result follows by an application of Theorem~\ref{localthm} to $C^*$ and $\rho^*$.
\end{proof}

\begin{proof}[Proof of Theorem~\ref{localthm}]  We may (by a padding argument as in the proof of Theorem~\ref{channelthm}) assume that $\delta \geq 2/N$.  We may also assume that $N \geq 2$ and $\delta < .5$.
Let $n$ be a value such that for any $z \in \{0, 1\}^N$, a description of length exactly $n$ can be given for the specialized circuit $C^{*}(z, \cdot)$; here, we can take $N \leq n \leq \poly(T, N)$.

Let $P(C, \xi)$ be a polynomial-time quantum algorithm which receives a description of a circuit $C$, of description length $n$, defining a binary measurement, and applies $C$ to an $\ell$-qubit input state $\xi$ (where $\ell$ is as in the statement of Theorem~\ref{localthm}), outputting the result.

Let $Y_n = Y_n(C, \ol{\sigma })$ be the input-oblivious advice-testing circuit provided by Theorem~\ref{new_generalized} for $(P, \rho^*, d := 2)$. 
The number of gates in $Y_n$ is at most $\poly(n) \leq \poly(T, N)$.  Let $p$ be the number of qubits in the quantum advice register for $Y_n$.
Let $C' = C'(z, \ol{\sigma})$ be the circuit which executes $Y_n(C^*(z, \cdot), \ol{\sigma})$ and outputs the measured bit $b_{\out}$.  

Next we will define $H$ as in the Theorem statement, using Theorem~\ref{thm:new_key}.  The circuit $Y_{n}$ has two subcircuits $A, B$, following Definition~\ref{def:adv_test}.  
Let $V(\ol{\sigma})$ be the circuit which executes $A(\ol{\sigma})$, and outputs the measured bit $b_{\adv}$. 
By item (i) of Theorem~\ref{new_generalized}, there exists a state $\ol{\sigma}^*$ on $p$ qubits for which $\bE[V(\ol{\sigma}^*)] \geq 1 - 2^{-n}$.  For large enough $N$ this is greater than $1 - \eps$, where $ \eps := (\delta/(2\kappa T^{\kappa}))^{\kappa}$ for the constant $ \kappa > 1$ from Theorem~\ref{thm:new_key}.

Theorem~\ref{thm:new_key} now gives us a Hamiltonian $H = H_{V, \eps}$ and quantum operation $R = R_{V, \eps}$.  These have the property that for any ground state $\ketbra{\psi}{\psi}$ of $H$, for $\xi := R(\ketbra{\psi}{\psi})$ we have 
\[ \bE[V(\xi)] \geq 1 - \delta/2 \ > \ n^{-2} \]
(the first inequality holding by of our choice of $\eps$).  By definition of $V$, this means that in the execution of $Y_n(C(z, \cdot), \xi)$, we have $\Pr[b_{\adv} = 1] > n^{-2}$.
(This holds for any $z \in \{0, 1\}^N$; the expectation above is independent of $z$ since $Y_{n}$ has the input-oblivious testing property.)
By our guarantee for $Y_{n}$ given in Theorem~\ref{new_generalized}, item (ii), it follows that in the execution of $Y_{n}(C(z, \cdot), \xi)$ on any circuit $C(z, \cdot)$ of description length $n$,
\[  \left|  \bE[b_{\out}| b_{\adv} = 1] -  \bE[P(C(z, \cdot), \rho_{n})]    \right|   \ \leq \   n^{-2} \ .   \]
Recall from our definition that the output bit of $C'(z, \xi)$ is distributed as $b_{\out}$ in the execution of $Y_{n}(C^*(z, \cdot), \xi)$.
Thus,
\[   \left|\bE[ C'(z, \xi) ]-  \bE[P(C^*(z, \cdot), \rho^*)]    \right|  \ \leq \ n^{-2} + \Pr[b_{\adv} = 0] \ ,     \]
where $b_{\adv}$ is as in the execution of $Y_{n}(C(z, \cdot), \xi)$.  We have seen that in this execution $\Pr[b_{\adv} = 1] > 1 - \delta/2$, so the right-hand side above is at most $n^{-2} + \delta/2 \leq \delta$.   Also, by our definitions, $\bE[P(C(z, \cdot), \rho^*)] =  \bE[C(z, \rho^*)]$.
This proves the Theorem.
\end{proof}

\section{Reduction to 5-local Hamiltonians}\label{sec:5local} 

In Sections~\ref{sec:5local} and~\ref{sec:2local}, we prove Theorem~\ref{thm:new_key}. The proof is achieved by a sequence of reductions.  Each reduction was defined previously, but we need to establish facts about these reductions not found in previous references~\cite{ksv, an, kkr, OT}.  This requires careful work.

For a Hamiltonian $H$, we use $\lambda_1(H) \leq \ldots \leq \lambda_{M}(H)$ to denote the real eigenvalues of $H$, counted according to their geometric multiplicity\footnote{That is, an eigenvalue $\lambda$ appears $p$ times in the list, where $p$ is the dimension of the eigenspace for $\lambda$.  By the spectral theorem we have $M = $ dim$(\mc{B}^{\ot N}) = 2^N$.} and sorted in nondecreasing order.  We will use $||H||$ to denote the operator norm of $H$. 

The \emph{energy} of a pure state $\ket{\psi}$ with respect to $H$ is defined as $\bra{\psi}H\ket{\psi}$.
It is a basic fact that for all vectors $\ket{\psi}$ we have $\bra{\psi}H\ket{\psi} \geq  \lambda_1(H) \cdot ||\ket{\psi}||$, and the ground states of $H$ are precisely those unit vectors for which equality holds.
In proving Theorem~\ref{thm:new_key} a key role will be played by \emph{nearly-minimal-energy} states---those unit vectors $\ket{\psi}$ for which $\bra{\psi}H\ket{\psi} \approx  \lambda_1(H)$.

In this section, we will use the original $\mathsf{QMA}$-completeness reduction, due to Kitaev~\cite{ksv}, to prove Theorem~\ref{thm:variant} below, a variant of Theorem~\ref{thm:new_key}.  This variant is weaker, in that the Hamiltonian $H$ produced is only required to have locality 5, rather than 2; but it is stronger in that the reduction $R$ is required to produce a useful state given any nearly-minimal-energy state for $H$ (not just any ground state).  This ``robust'' guarantee will be important in our subsequent construction of 2-local Hamiltonians.  Theorem~\ref{thm:variant} is also stronger in that $H, R$ are chosen independent of $\eps$, although this property is not essential for our work.

\begin{theorem}\label{thm:variant}
Let $V(\xi)$ be a quantum ``verifier'' circuit of $T$ gates (each 2-local), which acts on an $m$-qubit quantum state $\xi$ and an ancilla register of $N - m$ qubits (we may assume $N \leq 2T$), with the ancilla register initially in the all-zero state.  Suppose that $V$ defines a binary measurement on $\xi$.  Then there exists
\begin{itemize}
\item A 5-local Hamiltonian $H_{V}$ acting on $N'$-qubit states, for some $N' \leq O(T)$, expressed as a sum of 5-local terms $H_i$ of operator norm $\frac{1}{\poly(T, 1/\eps)} \leq || H_i || \leq \poly(T, 1/\eps)$, and
\item A quantum operation $R_{V}$ mapping $N'$-qubit states to $m$-qubit states, implemented by a quantum circuit with $\poly(T)$ gates,
\end{itemize}
for which the following property holds for any $\eps > 0$: if $\max_{\rho}\bE[V(\rho)] \geq 1 - \eps$, and if $\ket{\psi}$ is any $N'$-qubit state such that
\[     \bra{\psi}H_{V}\ket{\psi} \ < \  \lambda_1(H_{V}) +  \eps   \ , \]
then for $\xi := R_{V}(\ketbra{\psi}{\psi})$ we have
\[   \bE[ V(\xi) ] \ \geq \ 1 -   c \cdot T^{c}\eps^{1/c} \ , \]
where $c > 1$ is an absolute constant.
Furthermore, $H_{V}$ and $R_{V}$ can be constructed in time $\poly(T)$, given a description of $V$.
\end{theorem}


Theorems~\ref{localthm} and~\ref{thm:new_key} can be similarly strengthened, so that their guarantees hold for nearly-minimal-energy states of the local Hamiltonian as well as for ground states.  The dependence of the output Hamiltonian upon the choice of error parameters appears necessary in these results, however.

Similarly to Kitaev's work, it turns out to be convenient to first prove a weakened form of Theorem~\ref{thm:variant} in which the Hamiltonian is only required to be $O(\log T)$-local.  This forms the bulk of our work in this section. It will then be a simple step to reduce the locality to 5.

\subsection{The $O(\log T)$-Local Reduction}\label{ss:log_reduction}


\paragraph{The Hamiltonian:} Say that $V$, which expects a proof state $\xi$ on $m$ qubits, acts upon the ``proof register'' containing $\xi$ and an $(N- m)$-qubit ``ancilla register,'' initialized to the all-zero state, by the sequence $U_1, \ldots, U_T$ of unitary transformations, each of which is 2-local.  Here we may assume (by padding, if necessary) that $T + 1$ is a power of 2.  
The transformation performed by $V$, applied to a pure input state $\ketbra{\psi}{\psi}$, produces the state
\[  U_T  \ldots U_1 \cdot \left( \ket{\psi} \ot \ket{0^{N - m}}\right) \ .  \]
Afterward, we assume that the first qubit is measured in the standard basis; $V$ outputs the measured value.  We use $V$ to define a Hamiltonian $H = H_V$ acting on $N' := N + D$ qubits, where $D :=  \log_2 (T+1)$, as follows.  We speak of the first $N$ qubits (consisting of the proof and ancilla registers) jointly as the ``circuit register,'' and the last $N$ qubits as a ``clock register.''  The local unitaries $U_1, \ldots, U_T$ will be regarded as operators on the Hilbert space of the circuit register. We identify the computational basis states of the clock register with the integers $\{0, 1, \ldots, T\}$, and we write these basis states as $\ket{t}$ for $0 \leq t \leq T$.  

To specify projective operators acting on the circuit register, we use the notation
$\ketbra{b}{b}_i$ for $b \in \{0, 1\}$, $i \in [N]$ to denote the projection onto the subspace spanned by all computational basis vectors whose $i^{th}$ coordinate is $b$.  Formally,
\[   \ketbra{b}{b}_i  \ :=  \  I_{i - 1} \ot \ketbra{b}{b} \ot I_{N - i} \ .   \]

We define a Hamiltonian operator $H = H_{V}$ having three terms, $H_{\inn}, H_{\out},$ and $H_{\prop}$.  For our analysis we will depart slightly from~\cite{kkr} in our definitions; however, each of the three terms will be a positive scalar multiple of the corresponding term in~\cite{kkr}.  We define
\begin{equation}\label{eq:Hdef}
H \ := \ H_{\inn}  +   H_{\out} +  H_{\prop} \ ,
\end{equation}
where
\begin{equation}\label{eq:Hindef}
\hin \ := \ \frac{1}{2} \sum_{i = m + 1}^N \ketbra{1}{1}_i  \ot \ketbra{0}{0} \ 
\end{equation}
(here the rightmost projector $\ketbra{0}{0}$ is onto the basis vector $\ket{t = 0}$ for the clock register),
\begin{equation}\label{eq:Houtdef}
\hout \ := \  \frac{1}{2}  \ketbra{0}{0}_1 \ot  \ketbra{T}{T} \ ,
\end{equation}
and
\begin{equation}\label{eq:Hpropdef}
\hprop \ := \  \sum_{t = 1}^T H_{\prop, t} \ ,
\end{equation}
where the operators $H_{\prop, t}$ are defined for $t \in [T]$ by
\begin{equation}\label{eq:Hprop_tdef}
H_{\prop, t} \ := \  \frac{1}{2} \left(   I_N \ot \ketbra{t}{t} + I_N \ot \ketbra{t- 1}{t-1}  -   U_t \ot \ketbra{t}{t - 1}  -  U^{\dag}_{t} \ot \ketbra{t - 1}{t}  \right)   \ .
\end{equation}

Note immediately that the operator norms of the individual $O(\log T)$-local terms of $H$ are each $\Theta(1)$.

One can verify that $H_{\prop, t}$ is Hermitian.  More strongly, $H_{\inn}, H_{\out}$, and the terms $H_{\prop, t}$ are all positive semidefinite (PSD).  For the first two this is obvious: $H_{\inn}, H_{\out}$ are orthogonal projectors.  
To see that $H_{\prop, t}$ is PSD, it is clearly enough to show that $\bra{w} H_{\prop, t} \ket{w} \geq 0$ for any $\ket{w}$ of form $\ket{w} = \ket{w_{t - 1}}\ot \ket{t - 1} + \ket{w_t} \ot \ket{t}$.  We compute
\begin{align*} 2 \cdot \bra{w} H_{\prop, t} \ket{w}  \ &= \   \braket{w_t}{w_t}  + \braket{w_{t-1}}{w_{t-1}} - \bra{w_t}U_t \ket{w_{t - 1}}  - \bra{w_{t - 1}}U^\dag_t \ket{w_t} \\
&= \  || \ket{w}||^2 -  \bra{w_t}U_t \ket{w_{t - 1}} - \overline{\bra{w_t}U_t \ket{w_{t - 1}} } \\
&\text{(a real value, so $H_{\prop, t}$ is Hermitian)} \\
&\geq \  || \ket{w}||^2 -  2|| \ket{w_t} || \cdot || U_t \ket{w_{t - 1}} || \\
&= \ || \ket{w}||^2 -   2|| \ket{w_t} || \cdot ||  \ket{w_{t - 1}} || \\
&\geq \  || \ket{w}||^2 - ||\ket{w_t}||^2  -  ||\ket{w_{t - 1}}||^2 \\
&= \ 0 \ ,
\end{align*}
as needed.  Thus $H$, a sum of PSD operators, is itself PSD (and $\lambda_1(H) \geq 0$).  This will be important for our analysis.

\paragraph{The transformation of quantum states:} For our transformation $R = R_V$ of quantum states as in Theorem~\ref{thm:new_key}, we use the operation which first measures the clock register, observing some value $t \in [0, T]$, and then applies $U_1^{\dagger}\ldots U_T^{\dagger}$ to the circuit register, outputting the resulting $m$-qubit reduced state on the proof register alone (eliminating the ancilla and clock registers).  This transformation is implementable in size $\poly(T)$, since an inverse unitary operation $U^\dag$ is $k$-local whenever $U$ is $k$-local.

\paragraph{Objective of the analysis:} It is shown in~\cite{ksv, an} that, if $\max_{\rho} \bE[V(\rho)] \geq 1 - \eps$, then the minimal eigenvalue $\lambda_1(H)$ is at most $O(\eps)$.  (This fact is unaffected by our scalar-multiple adjustments to the definitions of $H_{\inn}, H_{\out}, H_{\prop}$.)  In our analysis, we will assume that $\lambda_1(H) < .01\delta/T$, where $\delta > 0$ will be defined as a sufficiently small inverse-polynomial in $T$.  This smallness assumption is without loss of generality, since our sought-after bound in Theorem~\ref{thm:variant} allows a $\poly(T)$ slack factor.  We will then show that if $\ket{\psi}$ is any state satisfying $\bra{\psi}H\ket{\psi} < .02\delta/T$, the $m$-qubit (mixed) state $\xi := R(\ketbra{\psi}{\psi})$ satisfies $\bE[V(\xi)] \geq 1 - \delta^{\Omega(1)}$.  This suffices to prove the weakened version of Theorem~\ref{thm:variant} in which $H$ is only required to be $O(\log T)$-local.

\subsection{Describing the Action of $H$ on a State}\label{ss:action}


Here we introduce notation and derive some useful expressions which describe the action of $H$ on an arbitrary pure state.

Consider an $(N + \log_2 (T+1))$-qubit state $\ket{\psi}$, given by
\[  \ket{\psi} \ = \ \sum_{y \in \{0, 1\}^N, t \in \{0, 1, \ldots, T\}} \alpha_{y, t} \ket{y} \ot \ket{t} \ ,    \]
with $\sum_{y, t} |\alpha_{x, t}|^2 = 1$.  We may write
\[   \ket{\psi} \ = \ \sum_{t \in \{0, 1, \ldots, T\}} \ket{\psi_t} \ot \ket{t}    \ ,  \]
where
\[  \ket{\psi_t} \ := \ \sum_{y \in \{0, 1\}^N} \alpha_{y, t} \ket{y} \ . \]
is a state on the circuit register.  Note, $\ket{\psi_t}$ is not in general a unit vector; we have $\sum_t ||\ket{\psi}_t ||^2 = 1$.
We define vectors $\ket{\xi_0}, \ldots, \ket{\xi_T}$ by the relation
\begin{equation}\label{eq:xi_def}   H \ket{\psi} \ = \ \sum_{t = 0}^T  \ket{\xi_t} \ot \ket{t} \ ,   \end{equation}
noting that the $\ket{\xi_t}$ will also not in general be unit vectors (nor will $H \ket{\psi}$ be).

Now for $t \in [0, T]$ define
\[   \ket{\phi_t} \ := \   U^{\dag}_1 U^{\dag}_2 \ldots U^{\dag}_t  \ket{\psi_t} \ ,  \]
so that
\[   \ket{\psi_t} \ = \   U_t U_{t - 1} \ldots U_1  \ket{\phi_t} \ .    \]
(Here, $\ket{\phi_0} = \ket{\psi_0}$ and $\ket{\phi_1} = U_1^{\dag}\ket{\psi_1}$.)  Note that $|| \ket{\phi_t}|| = ||\ket{\psi_t}||$ and
\[
\sum_{t= 0}^T || \ket{\phi_t} ||^2 \ =  \ \sum_t || \ket{\psi_t} ||^2  \ = \ 1 \ ,
\]
as the $U_t$ are unitary.  

With these definitions, we first examine the action of $H_{\prop}$ on $\ket{\psi}$.  
For $t \in [T]$, the operator $H_{\prop, t}$ acts as
\begin{equation}\label{eq:Hpropt_action}
H_{\prop, t} \ket{\psi} \ = \  \frac{1}{2} \left(  \ket{\psi_t}  \ot \ket{t}   +  \ket{\psi_{t-1}}  \ot \ket{t - 1} \  -  \   U_t \ket{\psi_{t-1}}  \ot \ket{t}   -   U^{\dag}_t \ket{\psi_{t}} \ot \ket{t - 1}  \right) \ ,
\end{equation}
which we can express as
\begin{equation}\label{eq:Hpropt_action2}
H_{\prop, t} \ket{\psi} \ = \  \frac{1}{2} \bigg( U_t \ldots U_1 \left( \ket{\phi_t} - \ket{\phi_{t - 1}}  \right)  \ot \ket{t} \ + \  U_{t-1} \ldots U_1 \left(   \ket{\phi_{t - 1}}   -   \ket{\phi_t}    \right)  \ot \ket{t - 1}    \bigg) \ .
\end{equation}
Next, observe that $H_{\inn}$ only outputs vectors in the span of the basis vectors with clock-register equal to 0, i.e., in the span of $\{\ket{y} \ot \ket{0}\}_y$, and that $H_{\out}$ outputs vectors in the span of $\{\ket{y} \ot \ket{T}\}_y$.  Thus for $t \in [ T - 1]$, the only contribution of terms of form $\ket{y}\ot \ket{t}$ to the output of $H \ket{\psi}$ comes from $H_{\prop, t}$ and $H_{\prop, t + 1}$, and we compute that for such $t$, 
\begin{equation}\label{eq:xi_interm}
\ket{\xi_t} \ = \  U_t \ldots U_1 \left(  \ket{\phi_t} -  .5 \ket{\phi_{t -1}} - .5 \ket{\phi_{t + 1}}     \right) \ . 
\end{equation}       
In particular, as $(U_t \ldots U_1)$ is unitary we have
\begin{equation}\label{eq:xi_norm_interm}
|| \ket{\xi_t}\ot \ket{t} || \ = \ || \ket{\xi_t}  || \ = \  ||  \ket{\phi_t} -  .5 \ket{\phi_{t -1}} - .5 \ket{\phi_{t + 1}}   || \ . 
\end{equation} 
Next we examine the terms in $H \ket{\psi}$ on clock-value $t = 0$, which come solely from the actions of $H_{\inn}$ and $H_{\prop, 1}$.  Define the orthogonal projector $\Pi_{\inn}$ acting on the $N$-qubit circuit register by 
\[  \Pi_{\inn}  \ := \ \sum_{i = m + 1}^N \ketbra{1}{1}_i \ ; \] 
the operator 
\[ \Pi'_{\inn} \ :=\  (I_N - \Pi_{\inn}) \]
 is also an orthogonal projection.
 We have
\begin{align*}\label{eq:xi_0}
\ket{\xi_0} \ &= \  .5 \ket{\psi_{0}} - .5 U^{\dag}_1 \ket{\psi_{1}}  +  .5\Pi_{\inn}\ket{\psi_0}  \\ 
&= \ .5 \ket{\phi_{0}} - .5  \ket{\phi_{1}}  +   .5\Pi_{\inn}\ket{\phi_0} \\
&= \ \ket{\phi_0} - .5(  \ket{\phi_0} -  \Pi_{\inn}\ket{\phi_0} )   - .5 \ket{\phi_1}   \\
&= \   \ket{\phi_0} - .5 \Pi'_{\inn}\ket{\phi_0} - .5 \ket{\phi_1} \ . 
\end{align*} 
Thus,
\begin{equation}\label{eq:xi_norm_0}
|| \ket{\xi_0}\ot \ket{0} || \ = \ || \ket{\xi_0}  || \ = \  || \ket{\phi_0} - .5 \Pi'_{\inn}\ket{\phi_0} - .5 \ket{\phi_1}    || \ . 
\end{equation} 
Finally we examine the terms in $H \ket{\psi}$ on clock-value $T$, which come solely from the actions of $H_{\out}$ and $H_{\prop, T}$.  Define the projector $\Pi_{\out}$ acting on the circuit register by $\Pi_{\out} :=  \ketbra{0}{0}_1$; define the operators 
\[ \Phi_{\out} \ := \ U^{\dag}_1 \ldots U^{\dag}_T \Pi_{\out} U_T \ldots U_1 \]
 and 
 \[ \Phi'_{\out} \ := \ I_N - \Phi_{\out}\]
 acting on $N$ qubits.  Then we have
\begin{align}\label{eq:xi_T}
\ket{\xi_T} \ &= \  .5 \ket{\psi_{T}} - .5 U_T \ket{\psi_{T - 1}}  +  .5\Pi_{\out}\ket{\psi_T} \\ 
&= \   \ket{\psi_{T}}  -  .5 U_T \ket{\psi_{T - 1}}  - .5\ket{\psi_T} + \Pi_{\out}\ket{\psi_T}  \\
&= \  U_T \ldots U_1 \bigg( \ket{\phi_{T}}  -  .5  \ket{\phi_{T - 1}} - .5 \ket{\phi_T} \bigg)  + .5\Pi_{\out}\ket{\psi_T} \\
&= \ U_T \ldots U_1 \bigg( \ket{\phi_{T}}  -  .5  \ket{\phi_{T - 1}} - .5 \ket{\phi_T} \bigg)  +  .5 \Pi_{\out} U_T \ldots U_1 \ket{\phi_T} \\
&= \ U_T \ldots U_1 \bigg( \ket{\phi_{T}}  -  .5  \ket{\phi_{T - 1}} - .5 \ket{\phi_T}   +  .5 U^\dag_1 \ldots U^\dag_T \Pi_{\out} U_T \ldots U_1 \ket{\phi_T}  \bigg) \\
&= \ U_T \ldots U_1 \bigg( \ket{\phi_{T}}  -  .5  \ket{\phi_{T - 1}} - .5\Phi'_{\out} \ket{\phi_T}  \bigg) \ .
\end{align} 
Thus,
\begin{equation}\label{eq:xi_norm_T}
|| \ket{\xi_T}\ot \ket{T} || \ = \ || \ket{\xi_T}  || \ = \  ||   \ket{\phi_{T}} -  .5 \Phi'_{\out} \ket{\phi_T} - .5 \ket{\phi_{T - 1}}    || \ . 
\end{equation}

\subsection{Analyzing Low-Energy States of $H$}\label{ss:small_delta}


Here we argue that if $\ket{\psi}$ is any state for which the energy $\bra{\psi}H\ket{\psi}$ is sufficiently small, then our operation $R = R_V$, when applied to $\ket{\psi}\bra{\psi}$, produces a state accepted with high probability by $V$.  No corresponding result is needed or established in Kitaev's original work~\cite{ksv}, which analyzed the minimal eigenvalue of $H$, but not the structure of ground states themselves.  Subsequent works, including~\cite{kkr, OT}, have provided more detailed information about the low-energy subspaces of several local-Hamiltonian reductions (although these works do not immediately yield the conclusions we seek).  We will make crucial use of results from~\cite{kkr, OT} in Section~\ref{sec:2local}.

We first describe the idea of our analysis.  Suppose $\ket{\psi}$ is any unit vector for which $|| H \ket{\psi}||$ is ``very small.''  We have
\[  || H \ket{\psi}||^2 = \sum_{t} || \ket{\xi_t} \ot \ket{t}||^2 \ = \ \sum_{t} || \ket{\xi_t}||^2 \ ,\]
so each $\ket{\xi_t}$ is a very small vector.  If $t \in [T - 1]$, then Eq.~(\ref{eq:xi_norm_interm}) tells us that $\ket{\phi_t} = U^\dag_1 \ldots U^\dag_t \ket{\psi_t}$ is nearly equal to the average of $\ket{\phi_{t - 1}}$ and $\ket{\phi_{t +1}}$.  For $t = 0$, Eq.~(\ref{eq:xi_norm_0}) tells us that $\ket{\phi_0}$ is nearly the average of $\Pi'_{\inn} \ket{\phi_0}$ and $\ket{\phi_1}$; and for $t = T$, Eq.~(\ref{eq:xi_norm_T}) tells us that $\ket{\phi_T}$ is nearly the average of $\Phi'_{\out} \ket{\phi_T}$ and $\ket{\phi_{T - 1}}$. 
Thus, the sequence
\begin{equation}\label{eq:ap}  \Pi'_{\inn} \ket{\phi_0}, \ \ket{\phi_0}, \ \ket{\phi_1}, \ \ldots,  \ \ket{\phi_T}, \ \Phi'_{\out} \ket{\phi_T}   \end{equation}
is very nearly an \emph{arithmetic progression} within the $N$-qubit Hilbert space of the circuit register.  

Now there are essentially two possibilities.  In the first, ``good'' case, the terms in this near-arithmetic progression are all nearly equal to $\ket{\phi_0}$, so that each $\ket{\psi_t}$ is nearly equal to $U_t \ldots U_1 \ket{\psi_0}$.  Inspecting the definitions of $\Pi'_{\inn}$ and $\Phi'_{\out}$, we then find that $\Pi_{\inn}\ket{\psi_0}$ and $\Pi_{\out}\ket{\psi_T}$ are both $\approx 0$.  This implies that $\ket{\psi_0}$, after normalization, is close to a legal input state (i.e., with the ancilla register in the all-zero state) causing the verifier $V$ to accept with high probability.  Moreover, we may obtain a near-perfect copy of $\ket{\psi_0}$ by the operation $R_V$ defined earlier. 

In the second, ``bad'' case, our near-arithmetic progression has some nontrivial step size, and its terms are close to being $(T + 3)$ equally-spaced points along some line in Hilbert space.  Now in such an arrangement, it is an intuitive fact that the furthest of these points from the origin will be either the first or the last point along the line.  Thus, either $\Pi'_{\inn} \ket{\phi_0}$ or $ \Phi'_{\out} \ket{\phi_T} $ will have the largest norm from among the vectors in our sequence.  However, one easily verifies that $\Pi'_{\inn}$ and $\Phi'_{\out}$ each have operator norm at most 1, so that $||\Pi'_{\inn} \ket{\phi_0} || \leq ||\ket{\phi_0}||$ and $|| \Phi'_{\out} \ket{\phi_T} || \leq || \ket{\phi_T} ||$.  So in fact the bad case cannot occur.


With this informal sketch in mind, we begin.  Fix any $\delta > 0$ satisfying 
\begin{equation*}
\label{eq:eps_small} \delta \ < \   \frac{1}{8^8 (T + 3)^{18}} \ .      
\end{equation*}
As discussed in Section~\ref{ss:log_reduction}, we will assume that there is some unit vector $\ket{\psi} = \sum_{t = 0}^T \ket{\psi_t}\ot\ket{t}$ such that
\[      \bra{\psi}H\ket{\psi}  \  \leq \   .02 \delta/T  \ ,  \]
and will show that $\bE[V(\xi)] \geq 1 - \delta^{\Omega(1)}$, where $\xi := R(\ketbra{\psi}{\psi})$.   First, we claim that the vector $H\ket{\psi}$ has small norm.  To see this, first use the spectral theorem to write
\[  H  \ = \ \sum_{\ell \in [2^{N'}]} \lambda_{\ell} \ketbra{\ell}{\ell} \ ,  \]
where $\{ \ket{\ell} \}$ is an orthonormal eigenbasis for $H$ and $\{\lambda_{\ell} = \lambda_{\ell}(H)\}$ are the corresponding eigenvalues.  We have $0 \leq \lambda_1 \leq \ldots \leq \lambda_{2^{N'}}$.  Write $\ket{\psi} = \sum_{\ell \in [2^{N'}]} \beta_{\ell} \ket{\ell}$, with $\beta_{\ell} \in \bC$ and $\sum_{\ell \in [2^{N'}]} |\beta_{\ell}|^2 = 1$.
We have the expressions
\[ H \ket{\psi} \ = \  \sum_{\ell \in [2^{N'}]} \beta_{\ell}  \lambda_\ell \ket{\ell}    \ ,  \quad{} \quad{}   || H \ket{\psi} ||^2  \ = \   \sum_{\ell \in [2^{N'}]} |\beta_{\ell}|^2  \lambda_\ell^2       \ ,   \quad{}\quad{}  \bra{\psi}  H \ket{\psi} \ = \  \sum_{\ell \in [2^{N'}]} |\beta_{\ell}|^2  \lambda_\ell   \ .   \]
Thus $  || H \ket{\psi} ||^2 \leq  \lambda_{2^{N'}} \cdot \bra{\psi}  H \ket{\psi}  =  ||H||\cdot \bra{\psi}  H \ket{\psi}$.  We have the crude operator-norm bound $||H|| \leq 10 T$, which follows by summing bounds on the norms of each term of $H$.
Thus,
\begin{equation}\label{eq:low_penalty}
|| H \ket{\psi}||^2 \ = \   \sum_{t = 0}^T || \ket{\xi_t} ||^2  \ \leq \ \delta \ ,
\end{equation}
where we again define $\{\ket{\xi_t}\}_t$ by the relation $H \ket{\psi} = \sum_t \ket{\xi_t} \ot \ket{t}$.  


For each $t \in \{0, 1, \ldots, T\}$, let $\delta_t := ||\ket{\xi_t}||^2$.  Define 
\begin{equation*}\ket{\Delta_0} \ :=  \ \ket{\phi_0} - \Pi'_{\inn}\ket{\phi_0}
\end{equation*}
as the difference between the first two terms in the sequence from Eq.~(\ref{eq:ap}).  For $t \in [T]$, define 
\[\ket{\Delta_t} \ := \ \ket{\phi_t} - \ket{\phi_{t - 1}}  \ , \]
and define
\[\ket{\Delta_{T + 1}} \ := \ \Phi'_{\out} \ket{\phi_T} - \ket{\phi_{T}}  \ .  \]
By Eq.~(\ref{eq:xi_norm_0}), we have
\begin{align}\label{eq:appd_norm_0}
\delta_0 \ &= \ ||\ket{\xi_0}||^2 \nonumber \\  
&= \  || .5(\ket{\phi_0} - \Pi'_{\inn}\ket{\phi_0}) - .5(\ket{\phi_1} - \ket{\phi_0}) ||^2 \nonumber \\
&= \ .25 ||  \ket{\Delta_0} -\ket{\Delta_1} ||^2 \ .
\end{align}
Similarly, by Eq.~(\ref{eq:xi_norm_interm}), for $t \in [T - 1]$ we have
\begin{align}\label{eq:appd_norm_interm}
\delta_t \ &= \ ||\ket{\xi_t}||^2 \nonumber \\  
&= \  || .5(\ket{\phi_t} - \ket{\phi_{t - 1}}) - .5(\ket{\phi_{t+1}} - \ket{\phi_{t}}) ||^2 \nonumber \\
&= \ .25  ||  \ket{\Delta_t} -\ket{\Delta_{t + 1}} ||^2 \ .
\end{align}
Finally, by Eq.~(\ref{eq:xi_norm_T}) we have
\begin{align}\label{eq:appd_norm_T}
\delta_T \ &= \ ||\ket{\xi_T}||^2   \nonumber  \\  
&= \  || .5(\ket{\phi_T} - \ket{\phi_{T - 1}}) - .5(\Phi'_{\out}\ket{\phi_{T}} - \ket{\phi_{T}}) ||^2   \nonumber  \\
&= \ .25  ||  \ket{\Delta_T} -\ket{\Delta_{T + 1}} ||^2 \ .
\end{align}
Combining our work, we find that for each $t \in \{0, 1, \ldots, T\}$ we have
\begin{equation}\label{eq:appd_norm_all}
  || \ket{\Delta_{T + 1}}  - \ket{\Delta_T} || = 2 \sqrt{\delta_t} \ .  
\end{equation}
At this point, for notational convenience we define
\[\ket{\phi_{-1}} \ := \ \Pi'_{\inn}\ket{\phi_0} \ , \quad{}\quad{}  \ket{\phi_{T + 1}} \ := \ \Phi'_{\out}\ket{\phi_T} \ .  \]
From the definitions of $\Pi'_{\inn}, \Phi'_{\out}$ one can verify that their operator norms are each at most 1, so that
\begin{equation*}\label{eq:norms_low_0} ||\ket{\phi_{-1}}|| \ \leq  \ || \ket{\phi_{0}}|| \ \leq \ || \ket{\psi}|| \ = \ 1
\end{equation*}
and 
\begin{equation*}\label{eq:norms_low_T} ||\ket{\phi_{T + 1}}|| \ \leq  \ ||\ket{\phi_{T}}|| \ \leq \ ||\ket{\psi}|| \ = \ 1 \ .
\end{equation*}

By our definitions, for each $t \in [T + 1]$ we have
\begin{align}
\ket{\phi_t} \ &= \ \ket{\phi_{-1}} + \sum_{t' = 0}^t \ket{\Delta_{t'}}  \nonumber  \\
&= \ \ket{\phi_{-1}} + \sum_{t' = 0}^t \left( \ket{\Delta_0}  + \sum_{t'' = 0}^{t' - 1}  (\ket{\Delta_{t'' + 1}} - \ket{\Delta_{t'' }})   \right)  \nonumber \\
&= \ \ket{\phi_{-1}}  + (t + 1)\ket{\Delta_0} + \sum_{s = 0}^{t-1}   (t - s) \cdot (\ket{\Delta_{s+1}} - \ket{\Delta_{s }}) \ .   \nonumber 
\end{align}
Using the triangle inequality and Eq.~(\ref{eq:appd_norm_all}), we find that for $t \in [T + 1]$,
\begin{align}\label{eq:good_approx}
||\ket{\phi_t} -  (  \ket{\phi_{-1}}  + (t + 1)\ket{\Delta_0}    ) || \ &= \ \left| \left| \sum_{s = 0}^{t-1}   (t - s ) \cdot (\ket{\Delta_s} - \ket{\Delta_{s - 1}})  \right| \right|  \nonumber \\
&\leq  2 \cdot \sum_{s = 0}^{t-1} (t - s )  \sqrt{\delta_s} \ .
\end{align}
In particular, this implies that
\begin{equation*}\label{eq:small}
||\ket{\phi_T}||  \  \leq  \  ||  \ket{\phi_{-1}}  + (T + 1)\ket{\Delta_0}     ||  +    2 \cdot \sum_{s = 0}^{T-1} (T - s )  \sqrt{\delta_s}
\end{equation*}
and
\begin{equation*}\label{eq:big}
||\ket{\phi_{T + 1}}||  \  \geq  \    ||  \ket{\phi_{-1}}  + (T + 2)\ket{\Delta_0}     ||  -    2 \cdot \sum_{s = 0}^{T} (T - s  + 1)  \sqrt{\delta_s} \ .
\end{equation*}
To understand these bounds, consider the linear function $\ell: \bR^{T + 1} \rightarrow \bR$ given by
\[ \ell(x_0, \ldots, x_T) \ := \ \sum_{s = 0}^T  (T - s + 1)x_s  \ . \]
It is a standard fact that the maximum value of $\ell$ in the disk $B_{0, r} \:= \{\overline{x} \in \bR^{T + 1}: \sum_s x_s^2 \leq r^2\}$ is attained at the point
\[  \overline{x}^* \ = \ (x^*_0, \ldots, x^*_t) \ := \ r \cdot \frac{\nabla \ell}{||\nabla \ell ||} \ ,  \] 
where the gradient function $\nabla \ell$ is defined as
\[ \nabla \ell \ := \ \left(\frac{\partial \ell}{\partial x_0} , \ldots, \frac{\partial \ell}{\partial x_T} \right) \ . \]
In our case, the gradient is the constant vector $\nabla \ell = (T + 1, T, T - 1, \ldots, 1)$.

Now recall that $\sum_{s = 0}^T \delta_s \leq \delta$.  It follows that
\begin{align}\label{eq:sum_bd}
\sum_{s = 0}^T (T - s + 1)\sqrt{\delta_s} \ &= \ \ell( \sqrt{\delta_0}, \sqrt{\delta_1}, \ldots , \sqrt{\delta_T}  ) \nonumber \\
&\leq \ \ell \left(  \sqrt{\delta} \cdot \frac{\nabla \ell}{|| \nabla \ell||}   \right) \nonumber \\
&= \ \frac{\sqrt{\delta}}{|| \nabla \ell ||} \cdot  \ell(\nabla \ell)  \nonumber \\
&= \   \frac{\sqrt{\delta}}{\sqrt{\sum_{s = 0}^T (T - s + 1)^2}} \cdot \left( \sum_{s = 0}^T (T - s + 1)^2 \right) \nonumber   \\ 
&= \ \sqrt{\delta \cdot \sum_{s = 0}^T (T - s + 1)^2   } \nonumber  \\
&= \  \sqrt{\frac{\delta (T + 2)(T + 3)(2T + 3)}{6} }        \nonumber  \\
&< \ \sqrt{\frac{\delta (T + 3)^3}{3}} \ .
\end{align}
Combining this with Eq.~(\ref{eq:good_approx}), we find that for $t \in [T + 1]$,
\begin{equation}\label{eq:useful_approx}
||\ket{\phi_t} -  (  \ket{\phi_{-1}}  + (t + 1)\ket{\Delta_0}    ) || \ \leq \  2 \sqrt{\frac{\delta (T + 3)^3}{3}} \ . 
\end{equation}
In particular, we have
\begin{equation}\label{eq:small2}
||\ket{\phi_T}||  \  \leq  \  ||  \ket{\phi_{-1}}  + (T + 1)\ket{\Delta_0}     ||     + 2\sqrt{\frac{\delta (T + 3)^3}{3}}
\end{equation}
and
\begin{equation}\label{eq:big2}
||\ket{\phi_{T + 1}}||  \  \geq  \    ||  \ket{\phi_{-1}}  + (T + 2)\ket{\Delta_0}     ||  -    2\sqrt{\frac{\delta (T + 3)^3}{3}} \ .
\end{equation}
Next, we claim that the quantity 
\[Q_{\ext} \ := \  || \ket{\phi_{-1}} ||^2 + ||\ket{\phi_{-1}} + (T + 2)\ket{\Delta_0}  ||^2\] is slightly larger than 
\[Q_{\intt} \ := \  || \ket{\phi_{0}} ||^2 + ||\ket{\phi_{-1}} + (T + 1)\ket{\Delta_0}||^2 \ , \]
if $\ket{\Delta_0}$ is of noticeable size.  This will be a useful way to quantify our intuition that the largest point in an arithmetic progression should be one of the endpoints.  Recall that $\ket{\phi_0} = \ket{\phi_{-1}} + \ket{\Delta_0}$.  We have
\[  \! \! \!\! \! \!\! \! \!\! \! \!\! \!\! \! \!\! \! \!\! \! \!\! \! \!\! \! \!\! \! \!\! \! \!\! \! \!\! \! \!\! \! \!\! \! \!\! \! \!\! \! \!\! \! \!\! \! \!\! \! \!\! \! \!\! \! \!\! \! \!\! \! \! \! \! \!\! \! \!\! \! \!\! \! \! Q_{\intt}  = \    \left( \braket{\phi_{-1}}{\phi_{-1}} + \braket{\Delta_0}{\Delta_0} +  \braket{\phi_{-1}}{\Delta_0} + \overline{\braket{\phi_{-1}}{\Delta_0}} \right)  \] 
 \[ + \   \left(  \braket{\phi_{-1}}{\phi_{-1}} +  (T+ 1)^2  \braket{\Delta_0}{\Delta_0}  
+  (T + 1)\left( \braket{\phi_{-1}}{\Delta_0} + \overline{\braket{\phi_{-1}}{\Delta_0}}  \right)   \right)    \]
\[   \!\! \! \!\! \! \!\! \! \!\! \! \!\! \! \!\! \! \!\! \! \! = \   2 ||\ket{\phi_{-1}}||^2  +   (T^2 + 2T + 2) ||\ket{\Delta_0} ||^2   + (T + 2) \left( \braket{\phi_{-1}}{\Delta_0} + \overline{\braket{\phi_{-1}}{\Delta_0}}  \right)  \ . \]
By a similar calculation,
\[ Q_{\ext}\ = \    2 ||\ket{\phi_{-1}}||^2  +   (T^2 + 2T + 4) ||\ket{\Delta_0} ||^2       + (T + 2) \left( \braket{\phi_{-1}}{\Delta_0} + \overline{\braket{\phi_{-1}}{\Delta_0}}  \right)   \ ,      \]
so that
\begin{equation}\label{eq:q_diff} Q_{\ext} - Q_{\intt} =  2 \cdot ||\ket{\Delta_0} ||^2 \ . \end{equation}
We next define
\[  Q'_{\ext} \ :=  \  || \ket{\phi_{-1}} ||^2  +   ||\ket{\phi_{T + 1}} ||^2  \]
and
\[  Q'_{\intt} \ :=  \  || \ket{\phi_{0}} ||^2  +   ||\ket{\phi_{T }} ||^2 \ . \]
Using Eq.~(\ref{eq:big2}), we have
\begin{align}  \label{eq:q_ext_diff}    
Q_{\ext} - Q'_{\ext} \ &= \   ||\ket{\phi_{-1}} + (T + 2)\ket{\Delta_0}  ||^2   -   ||\ket{\phi_{T + 1}} ||^2  \nonumber     \\
&\leq \    \left(  ||\ket{\phi_{T + 1}} || +  2\sqrt{\frac{\delta (T + 3)^3}{3}} \right)^2        -   ||\ket{\phi_{T + 1}} ||^2 \nonumber     \\
&\leq \  \frac{4\delta (T + 3)^3}{3}   +  4\sqrt{\frac{\delta (T + 3)^3}{3}}  \nonumber    \\
&\leq \ 8\sqrt{\frac{\delta (T + 3)^3}{3}}
\end{align}
where in the last two steps we used the fact that $|| \ket{\phi_{T + 1}}||\leq 1$ and our smallness assumption on $\delta$.

Similarly, using Eq.~(\ref{eq:small2}),
\begin{align}  \label{eq:q_int_diff}   
Q'_{\intt} - Q_{\intt} \ &= \   ||\ket{\phi_{T }} ||^2  -   ||\ket{\phi_{-1}} + (T + 1)\ket{\Delta_0}  ||^2   \nonumber   \\
&\leq \   ||\ket{\phi_{T }} ||^2  -    \left(  ||\ket{\phi_{T }} || -  2\sqrt{\frac{\delta (T + 3)^3}{3}} \right)^2       \nonumber   \\
&\leq \  4\sqrt{\frac{\delta (T + 3)^3}{3}} \ ,
\end{align}
where in the last step we used that $|| \ket{\phi}_T|| \leq 1$.

Combining Eqs.~(\ref{eq:q_diff}),~(\ref{eq:q_ext_diff}), and~(\ref{eq:q_int_diff}), we compute that
\begin{align}\label{eq:q_pos}
Q'_{\ext} - Q'_{\intt} \ &= \ (Q'_{\ext} - Q_{\ext}) + (Q_{\ext} - Q_{\intt}) + (Q_{\intt} - Q'_{\intt}) \nonumber \\
&\geq \   -8\sqrt{\frac{\delta (T + 3)^3}{3}}  +  2 || \Delta_0||^2 -4\sqrt{\frac{\delta (T + 3)^3}{3}}   \nonumber \\
&= \ 2 \left( || \Delta_0||^2    -  \sqrt{12\delta (T + 3)^3}   \right) \ .
\end{align}
On the other hand, recall that $|| \ket{\phi_{-1}}|| \leq || \ket{\phi_0} ||$ and $||\ket{\phi_{T + 1}} || \leq ||\ket{\phi_T} ||$.  Thus, $Q'_{\ext} - Q'_{\intt} \leq 0$.  With Eq.~(\ref{eq:q_pos}), this implies that
\begin{equation}\label{eq:delta_small}
|| \ket{\Delta_0} || \ \leq   \left(  12 \delta \right)^{1/4} (T +3)^{3/4} \ .
\end{equation}
Informally, this tells us that $\ket{\Delta_0}$ is small, so that we are not in the ``bad case'' described earlier.

Now, Eqs.~(\ref{eq:useful_approx}) and~(\ref{eq:delta_small}) combine to show us that $\ket{\phi_1}, \ket{\phi_2}, \ldots, \ket{\phi_{T+1}}$ are all close to $\ket{\phi_{-1}}$: for $t \in [T + 1]$,
\begin{align}\label{eq:all_vectors_close}
||  \ket{\phi_t} -  \ket{\phi_{-1}} || \ &\leq \  \bigg|\bigg|\ket{\phi_t} -  (  \ket{\phi_{-1}}  + (t + 1)\ket{\Delta_0}    ) \bigg|\bigg|   +\bigg|\bigg|(  \ket{\phi_{-1}}  + (t + 1)\ket{\Delta_0}    )  - \ket{\phi_{-1}}  \bigg|\bigg|   \nonumber
 \\
&\leq \   2 \sqrt{\frac{\delta (T + 3)^3}{3}}  \  + \    (t + 1)|| \ket{\Delta_0} ||   \nonumber  \\
&\leq \  2 \sqrt{\frac{\delta (T + 3)^3}{3}}  \ + \   \left(  12 \delta \right)^{1/4}  (T +3)^{7/4} \nonumber \\
&\leq \  4 \delta^{1/4}   (T +3)^{7/4}  \ ,
 \end{align}
using our smallness assumption on $\delta$ for the last step.  This also implies that
\begin{equation}\label{eq:pi_out_close}   
 ||\ket{\phi_{T + 1}} -  \ket{\phi_T} || \ \leq \  8 \delta^{1/4}   (T +3)^{7/4}  \ .
\end{equation} 
Also, using Eq.~(\ref{eq:delta_small}) again, we have
\begin{equation}\label{eq:pi_in_close}   
 || \ket{\phi_0} - \ket{\phi_{-1}}  || \ = \  || \ket{\Delta_0}|| \ \leq \ 2 \delta^{1/4}  (T +3)^{3/4} \ .
\end{equation}
Next we argue that for each $t \in [0, T]$, the vector $\ket{\phi_t}$ has norm close to $T^{-1/2}$.  Recall that $|| \ket{\phi_t} || = || \ket{\psi_t}||$ for $t \in [0, T]$, and that $\sum_{t = 0}^T|| \ket{\psi_t}||^2 = 1$.  Thus there is at least one value $t^* \in [0, T]$ for which $|| \ket{\phi_{t^*}}|| \geq (T + 1)^{-1/2}$.  Then, using Eq.~(\ref{eq:all_vectors_close}), for any $t \in [0, T + 1]$ we have
\begin{align}\label{eq:vecs_big}  || \ket{\phi_t} || \ &\geq \  ||\ket{\phi_{t^*}}|| -  ||\ket{\phi_{t^*}} - \ket{\phi_{-1}}|| - ||\ket{\phi_{-1}} - \ket{\phi_t}|| \nonumber \\    &\geq \ (T + 1)^{-1/2} -  8 \delta^{1/4}   (T +3)^{7/4} \nonumber \\ 
&\geq \ (1 - \delta^{1/8}) (T + 1)^{-1/2}  \ ,  \end{align}
where in the last step we again used our smallness assumption on $\delta$.  Similarly, using Eqs.~(\ref{eq:all_vectors_close}) and~(\ref{eq:pi_in_close}) we obtain
\begin{equation*}\label{eq:bignorm_in}
||\ket{\phi_{-1}} || \ \geq \ (1 - \delta^{1/8}) (T + 1)^{-1/2} \ .
\end{equation*}
On the other hand, there is also a $t_* \in [0, T]$ for which $|| \ket{\phi_{t_*}}|| \leq (T + 1)^{-1/2}$.  By modifying the above arguments only slightly, we find that for each $t \in [-1, T]$,
\begin{equation}\label{eq:smallnorm_all}
||\ket{\phi_{t}} || \ \leq \ (1 + \delta^{1/8}) (T + 1)^{-1/2} \ .
\end{equation}

For $t \in [-1, T + 1]$, define the normalized vector
\[ \widehat{\ket{\phi_t}} \ :=  \    \frac{\ket{\phi_t}}{||  \ket{\phi_t}||} \ . \]
Also, for $t \in [0, T]$, similarly define
\[ \widehat{\ket{\psi_t}} \ :=  \    \frac{\ket{\psi_t}}{||  \ket{\psi_t}||} \ . \]
Next, we define $\gamma_t$ by the relation
\[   \widehat{\ket{\phi_t}} \ =   (1 + \gamma_t) \sqrt{T + 1}\cdot \ket{\phi_t}     \ ;  \]
by Eqs.~(\ref{eq:vecs_big})-(\ref{eq:smallnorm_all}), we have $\gamma_t \in [- \delta^{1/8}, + \delta^{1/8}]$.
Then for $t \in [0, T + 1]$ we have
\begin{align}\label{eq:normal_dists}   ||  \widehat{\ket{\phi_t}} - \widehat{\ket{\phi_{-1}}}  ||  \ &\leq \   \sqrt{T + 1} \cdot  \left(  ||  \ket{\phi_t} - \ket{\phi_{-1}} ||  +  |\gamma_t | \cdot || \ket{\phi_t} || + |\gamma_{-1}| \cdot ||\ket{\phi_{-1}} ||  \right)     \nonumber \\
&\leq \  \sqrt{T + 1}\cdot \left(  4 \delta^{1/4}   (T +3)^{7/4}   + 2 \delta^{1/8}(1 + \delta^{1/8})(T + 1)^{-1/2}  \right) \nonumber \\
&\leq \   4\delta^{1/4}(T + 3)^{9/4}  + 3\delta^{1/8} \nonumber \\
&\leq \  4 \delta^{1/8}  \ .
 \end{align}
This in particular implies
\begin{equation} \label{eq:nice_relation}
4 \delta^{1/8} \ \geq \ || U_T \ldots U_1 ( \widehat{\ket{\phi_T}} - \widehat{\ket{\phi_{-1}}} ) || \ = \ \left| \left|  \widehat{  \ket{\psi_T} } - U_T \ldots U_1 \left( \frac{ \Pi'_{\inn} \ket{\psi_0}}{|| \Pi'_{\inn} \ket{\psi_0}||}   \right) \right|   \right|   \ .
\end{equation}

Next, expanding the definitions of terms in Eq.~(\ref{eq:pi_out_close}), we find that
\begin{align}\label{eq:output}   
8 \delta^{1/4}   (T +3)^{7/4} \ &\geq \  ||\Phi'_{\out}\ket{\phi_T} -  \ket{\phi_T} || \nonumber \\ 
&= \  ||\Phi_{\out}\ket{\phi_T} || \nonumber \\
&= \  || U^\dag_1 \ldots U^\dag_T \Pi_{\out} U_T \ldots U_1 (U^\dag_1 \ldots U^\dag_T \ket{\psi_T})|| \nonumber \\
&= || U^\dag_1 \ldots U^\dag_T \Pi_{\out} \ket{\psi_T}|| \nonumber \\
&= || \Pi_{\out}\ket{\psi_T} || \ .
\end{align} 
Now Eq.~(\ref{eq:vecs_big}), applied with $t := T$, tells us that $|| \ket{\psi_T} || = || \ket{\phi_T} || \geq (1 - \delta^{1/8})(T + 1)^{-1/2}$.  Combining this with Eq.~(\ref{eq:output}), we have
\begin{equation} \label{eq:penalty_small}
 ||\Pi_{\out}\widehat{\ket{\psi_T}} || \ \leq \  2\sqrt{T + 1} \cdot 8 \delta^{1/4}   (T +3)^{7/4} \ \leq \ 16 \delta^{1/4}(T + 3)^{9/4}   \ . 
\end{equation}
$\Pi_{\out}$ is an orthogonal projection and has operator norm 1.  This fact, combined with Eqs.~(\ref{eq:nice_relation}) and~(\ref{eq:penalty_small}), allows us to infer that
\begin{equation*}\label{eq:oh_yeah} 
\left| \left| \Pi_{\out} \left(  U_T \ldots U_1  \left(  \frac{ \Pi'_{\inn} \ket{\psi_0}}{|| \Pi'_{\inn} \ket{\psi_0}||}  \right)  \right)  \right| \right|  \ \leq \  4\delta^{1/8} +  16 \delta^{1/4}(T + 3)^{9/4}  \ \leq \ 8 \delta^{1/8} \ .
\end{equation*}
 
This shows that the quantum state $\widehat{\ket{\phi_{-1}}} = \frac{ \Pi'_{\inn} \ket{\psi_0}}{|| \Pi'_{\inn} \ket{\psi_0}||} $, when set as the initial state of the circuit register of the verifier $V$, causes $V$ to accept with probability $\geq 1 - \delta^{\Omega(1)}$.  Also, $\widehat{\ket{\phi_{-1}}}$ lies in the kernel of the orthogonal projector $\Pi_{\inn} = I_N - \Pi'_{\inn}$, so its final $N - m$ qubits are in the all-zero state.

 
Finally, we claim that being given the state $\ket{\psi}$ allows us to recover a close approximation to $\widehat{\ket{\phi_{-1}}}$ by applying our quantum operation $R = R_V$.  This procedure first measures the clock register.  If $t \in [T]$ is observed, the post-measurement circuit register state is $\widehat{\ket{\psi_t}}$; the transformation
\[   \widehat{\ket{\psi_t}}  \ \longrightarrow \ U^\dag_1 \ldots U^\dag_t\widehat{\ket{\psi_t}} \ = \ \widehat{\ket{\phi_t}} \ ;     \]
is then performed.  (If the value $t = 0$, the post-measurement state on the circuit register is $\widehat{\ket{\psi_0}} = \widehat{\ket{\phi_0}}$ and $R$ applies none of these unitaries.)  Eq.~(\ref{eq:normal_dists}) tells us that the resulting state $\widehat{\ket{\phi_t}}$ on the circuit register is $4\delta^{1/8}$-close to the desired state $\widehat{\ket{\phi_{-1}}}$.
Thus, the reduced state of $\widehat{\ket{\phi_t}}$ on the $m$-qubit proof register (which $R$ outputs) causes $V$ to accept with probability $\geq 1 - \delta^{\Omega(1)}$. 
We have established the variant of Theorem~\ref{thm:variant} which requires $H$ only to be $O(\log T)$-local.


 \subsection{Reduction to Locality 5}
 
 Following~\cite{ksv, an}, we now describe a small alteration of the above $O(\log T)$-local reduction that produces a 5-local Hamiltonian.
 
\paragraph{The modified reduction:} The Hilbert space used still consists of an $N$-qubit along with a ``clock register.''  This time, however, the clock register consists of $T$ qubits; informally, its ``intended purpose'' is to store a time-index $t \in [0, T]$ by the unary encoding $\ket{1^t 0^{T - t}}$.  A clock-register basis state of this form is called \emph{valid}; basis states not of this form are said to be \emph{invalid}, and will be penalized by our Hamiltonian.
For $t \in [0, T - 2]$ and bits $a, b, c, a', b', c'$, we let 
\[  \ketbra{a' b' c'}{a b c}_{\cl(t)}    \]
denote the 3-local operator $\ketbra{a' b' c'}{a b c}$ applied to the $t^{th}$, $(t+1)^{st}$, and $(t+2)^{nd}$ clock register qubits.  Similarly, $ \ketbra{a' b'}{a b}_{\cl(t)}$ denotes $ \ketbra{a' b'}{a b}$ applied to the $t^{th}$ and $(t+1)^{st}$ clock qubits.

We modify the Hamiltonian $H = H_V : \mc{B}^{\ot(N + D)} \rightarrow \mc{B}^{\ot(N + D)}$ from our previous work to produce a new Hamiltonian $H' = H'_V$ acting on the new Hilbert space $\mc{B}^{\ot(N + T)}$.  
First, in each tensor term appearing in $H_{\inn}, H_{\out}$, we replace the clock-register projectors $\ketbra{ 0}{0}, \ketbra{ T}{T}$ with $\ketbra{00}{00}_{\cl(0)}, \ketbra{11}{11}_{\cl(T - 1)}$ respectively to get modified operators $H'_{\inn}, H'_{\out}$ acting on our new Hilbert space:
\[   H'_{\inn}  \ := \   \frac{1}{2}\sum_{i = m + 1}^N  \ketbra{1}{1}_i \ot \ketbra{00}{00}_{\cl(0)}    \ , \quad\quad  H'_{\out}  \ := \  \frac{1}{2}\ketbra{0}{0}_1 \ot  \ketbra{11}{11}_{\cl(T - 1)} \ .  \]
Similarly, we define $H'_{\prop} := \sum_{t = 1}^T H'_{\prop, t}$ as follows. In each tensor-product term defining $H_{\prop, t}$, if $t \in [2, T - 1]$ then we replace the clock-register projectors
\[  \ketbra{t}{t} \ , \  \ketbra{t-1}{t-1} \  , \  \ketbra{t}{t-1} \ , \  \ketbra{t - 1}{t}   \]
with, respectively,
\[ \ketbra{110}{110}_{\cl(t - 1)} \ , \   \ketbra{100}{100}_{\cl(t - 1)} \ , \  \ketbra{110}{100}_{\cl(t - 1)} \ , \ \ketbra{100}{110}_{\cl(t - 1)} \ ,    \]
to obtain $H'_{\prop, t}$.
Finally, we introduce a new ``clock term'' $H'_{\cl} :=  \sum_{t =1}^T I_N \ot \ketbra{01}{01}_{\cl(t - 1)}$, penalizing invalid clock-register states.  We let $H' := H'_{\inn} + H'_{\out} + H'_{\prop} + H'_{\cl}$.  The operator norms of the individual 5-local terms of $H'$ are $\Theta(1)$, satisfying the norm requirement in Theorem~\ref{thm:variant}'s statement.

The modified quantum operation $R'$ is defined in close analogy to $R$ from our previous reduction. The only difference is that when $R'$ first measures the clock register (now on $T$ qubits), a measurement outcome $1^t 0^{T - t}$ is interpreted as the time-index $t$, and an outcome not of this form is interpreted (arbitrarily) as seeing the time-index $t = 0$.


\paragraph{The analysis:} Following previous works, we make several observations about $H'$.  First, $H'$ is PSD by the same argument as for $H$, and its operator norm still satisfies the crude upper-bound $||H'|| \leq 10 T$ used previously.  Next, define the subspace $S_{\val} \leq \mc{B}^{\ot(N + T)}$ as all vectors which place amplitude 0 on invalid clock-register basis states.  Note that $H'(S_{\val}) \subseteq S_{\val}$, and therefore (as $H'$ is Hermitian) also $H'(S_{\val}^{\perp}) \subseteq S_{\val}^{\perp}$.

 Let $L: S_{\val} \rightarrow \mc{B}^{\ot(N + D)}$ be the linear mapping defined on basis states by
\[ L(\ket{x} \ot \ket{1^t 0^{T - t}} ) \ := \   \ket{x} \ot \ket{t}  \   \quad{}\quad{} \text{for $x \in \{0, 1\}^N$, $t \in [0, T]$.}         \]
Then we observe that for any $\ket{\phi} \in S_{\val}$, we have the relation 
\begin{equation}\label{eq:iso}
H' (\ket{\phi}) \ = \ H(L(\ket{\phi})) \ .
\end{equation}
Moreover, $L$ is surjective; it follows that $\lambda_1(H') \ \leq \ \lambda_1(H)$.  We claim, however, that for any $\ket{\phi}$ in the orthogonal complement $S^{\perp}_{\val}$ (consisting of vectors which place zero amplitude on valid clock-register states), we have $\bra{\phi}H'\ket{\phi} \geq 1$.  To see this, just note that $\bra{\phi}H'_{\cl}\ket{\phi} \geq 1$, and that $\bra{\phi}(H'_{\inn} + H'_{\out} + H'_{\perp})\ket{\phi} \geq 0$ (since each of the three inner summands is PSD).  Thus $S_{\val}$ is spanned by eigenvalues of $H'$ all of which are $\geq1$.

Following the discussion at the end of Section~\ref{ss:log_reduction}, let us once more assume that $\max_{\xi}\bE[V(\xi)] \geq 1 - \gamma$, where $\gamma = \Theta(\delta/T)$ is sufficiently small that $\lambda_1(H) < .001 \delta / T$, where $\delta$ is as in Eq.~(\ref{eq:low_penalty}).
Let $\ket{\phi} \in \mc{B}^{\ot(N + T)}$ be any unit vector satisfying $\bra{\phi} H'\ket{\phi} < .002\delta/T$ (some such $\ket{\phi}$ must exist, since $\lambda_1(H') \leq \lambda_1(H)$).  Decompose $\ket{\phi} = \alpha\ket{\phi}_{\val} + \beta \ket{\phi}_{\inval}$ into its components in $S_{\val}, S_{\val}^{\perp}$ respectively (where $\ket{\phi}_{\val}, \ket{\phi}_{\inval}$ are normalized).  $H'\ket{\phi}_{\inval}$ is contained in $S_{\val}^{\perp}$ and has inner product at least $1$ with $ \ket{\phi}_{\inval}$, so we must have $|\beta|^2 \leq .002\delta/T$.  Thus, if we define the unit vector $\ket{\phi'} := \frac{\alpha}{|\alpha|}\ket{\phi}_{\val} \in S_{\val}$, we have
\begin{equation}\label{eq:close_enuff}  \left|\left| \ket{\phi}  -    \ket{\phi'}   \right|\right| \ \leq \ O(\sqrt{\delta/T}) \ . \end{equation}
$\ket{\phi'}$ also satisfies $\bra{\phi'} H'\ket{\phi'} \leq   \frac{1}{|\alpha|^2} \cdot \bra{\phi} H'\ket{\phi} < .02 \delta / T$.  Eq.~(\ref{eq:iso}) and our analysis of the Hamiltonian $H$ from previous sections then imply that the state $\xi := R(L(\ketbra{\phi'}{\phi'}))$ satisfies $\bE[V(\xi)] \geq 1 - \delta^{\Omega(1)}$.  Now observe that, by our definition of $R'$, the state $R'(\ketbra{\phi'}{\phi'})$ is identically distributed to $\xi$ (over the randomness in the measurement of the clock register).  Thus $\bE[V(R'(\ketbra{\phi'}{\phi'}))] \geq 1 - \delta^{\Omega(1)}$.  Combining this with Eq.~(\ref{eq:close_enuff}), we conclude that $\bE[V(R'(\ketbra{\phi}{\phi}))] \geq 1 - \delta^{\Omega(1)}$.  This proves Theorem~\ref{thm:variant}.


\section{Reduction to 2-local Hamiltonians}~\label{sec:2local}  

\vspace{-1.5 em}

\subsection{Goals of the Section, and Proof of Theorem~\ref{thm:new_key}}

In this section, we complete the proof of Theorem~\ref{thm:new_key}.  The following definition will be of central importance.  Informally speaking, it gives a notion of ``witness-preserving reductions'' between two problems in $\mathsf{QMA}$, where the ``witnesses'' here are quantum states (the precise definition given here is specific to the setting of Local Hamiltonian problems).\footnote{We note that most natural $\mathsf{NP}$-hardness reductions are easily seen to have a witness-preserving property: for example, in Karp's reduction mapping a 3-SAT instance $\psi$ to a Hamiltonian Path instance $G$, the two instances are not only equivalent with respect to their underlying decision problems, but any Hamiltonian path for $G$ can also be used to efficiently obtain a satisfying assignment for $\psi$.}

\begin{definition}\label{def:agpr}  Let $k > k' > 1$ be integers.  A \emph{$(k, k')$-approximate ground-space-preserving reduction (AGPR)} is a (classical, deterministic) algorithm $A$ of the following form. $A$ takes as input a tuple $(H, W, \beta)$, where $H$ is a description of a $k$-local Hamiltonian $H = \sum_{i \in [s]} H_i$ acting on some number $n$ of qubits; $W \geq 1$ is an integer; and $\beta \in (0, 1)$ is an accuracy parameter . The $s \geq n$ terms $H_1, \ldots, H_s$ are each expected to have operator norm $||H_i||$ in the range $[W^{-1}, W]$---if not, $A$ may behave arbitrarily.  
$A$ runs in time $\poly(s, W, 1/\beta)$ and outputs a pair $(H', R)$, where: 
\begin{itemize}
\item $H'$ is a $k'$-local Hamiltonian acting on some number $n' \leq \poly(s, W, 1/\beta)$ of qubits.  Each term in the expression for $H'$ has operator norm in the range $[1/W', W']$, for some $W' \leq \poly(s, W, 1/\beta)$; 
\end{itemize}
\begin{itemize}
\item $R$ is a quantum operation involving one or more measurements, that maps a pure $n'$-qubit pure state $\ket{\psi}$ to a pure $n$-qubit state under every possible set of measurement outcomes (the resulting pure state depends on the outcomes).  $R$ is implemented by a quantum circuit of size $\poly(s, W, 1/\beta)$.   
\end{itemize}
 Letting $\lambda_1, \lambda_1' \in \bR$ denote the minimal eigenvalues of $H, H'$ respectively, the pair $(H', R)$ are required to obey the following property: there is a $\delta \leq  \beta^{\Omega(1)} \cdot \poly(W, s) $ such that, if $\ket{\psi} \in \mathcal{B}^{\ot n'}$ is any pure state such that
 \[   \bra{\psi}H'\ket{\psi} \ < \ \lambda'_1 + \beta \ ,     \]
then the state $\ket{\phi}$ outputted by $R(\ket{\psi})$ satisfies
 \[   \bra{\phi}H\ket{\phi} \ < \ \lambda_1 +  \delta      \]
with probability at least $1 -  \delta$ over the randomness in $R$.
\end{definition}

We will prove:

\begin{theorem}\label{thm:app_gs} For each of $k \in \{5, 4, 3\}$, there exists a $(k, k - 1)$-AGPR.
\end{theorem}

In fact, in the reductions we construct are able to take $\delta \leq O(\beta)$ in Definition~\ref{def:agpr}, although this is not crucial to our work.
We defer the proof of Theorem~\ref{thm:app_gs} to subsequent sections.  AGPRs also compose nicely, as we prove next:

\begin{lemma}\label{lem:compose} Let $k > k' > k'' > 1$ be integers.  Suppose there exists a $(k, k')$-AGPR, call it $A$, and a $(k', k'')$-AGPR $A'$.  Then there also exists a $(k, k'')$-AGPR.
\end{lemma}

%

%

\begin{proof}
We will compose $A$ and $A'$ with suitably chosen parameters.  At the outset we note that, by the polynomial slack factor allowed in Definition~\ref{def:agpr}, we may assume that $\beta <\frac{1}{D(s + W)^D}$ for some fixed constant $D > 1$.  We will indicate where this assumption is used.

Consider the reduction $A^*$ which takes as input: a $k$-local Hamiltonian $H^{(k)}$ (of $s$ terms, acting on $n$ qubits); a bound $W$ as in the definition; and an $\beta > 0$. $A^*$ works as follows.  First, we choose $\gamma := \beta^{c}$, with $c > 0$ a small value to be determined later.
We apply our $(k, k')$-AGPR $A$ to $(H^{(k)}, W, \gamma)$ to obtain a pair $(H^{(k')}, R)$ each acting on $n' \leq \poly(s, W, 1/\gamma)$ qubits, with $H^{(k')}$ expressed by $s' \leq \poly(s, W, 1/\gamma)$ terms.  By subdividing terms if necessary, we can assume $s' \geq n'$.  Associated with $H^{(k')}$ is a second norm-bounding value $W' \leq \poly(s, W, 1/\gamma)$ as in Definition~\ref{def:agpr}.  Let $\delta\leq  \gamma^{\Omega(1)} \cdot \poly(s, W) $ be as in the guarantee for the pair $(H^{(k)}, H^{(k')})$.

Next, we apply our $(k', k'')$-AGPR $A'$ to $( H^{(k')} , W', \beta)$.
We get a pair $(H^{(k'')}, R')$ each acting on $n'' \leq \poly(s', W', 1/\beta)$ qubits.  Let $\delta' \leq \beta^{\Omega(1)}\cdot\poly(s', W') \leq \beta^{\Omega(1)} \cdot \poly(s, W, 1/\gamma)$ be the value in the associated guarantee for the pair $(H^{(k')}, H^{(k'')})$. 

We have $\delta' \leq C (s + W)^C \beta^{1/C}/\gamma^{C}$ for some constant $C  > 1$ (independent of our choice for $\gamma$).  We choose $\gamma := \beta^{1/(3C^2)}$. It follows that $\delta' \leq C (s + W)^C \beta^{2/(3C)}$.  Now using our aforementioned slack, we require that $\beta$ is a sufficiently small inverse-polynomial in $(s + W)$ that the above also implies $\delta' \leq \gamma$.

Our reduction $A^*$ outputs $H^{(k'')}$ and the composed reduction $R^* :=R  \circ R'$, which (by the assumed properties of $R, R'$) maps pure $n''$ qubit-states to pure $n$-qubit states, and is implemented by a circuit of size $\poly(s, W, 1/\beta)$.
$H^{(k'')}$ is $k''$-local as needed, and is expressed by $s'' \leq \poly(s, W, 1/\beta)$ terms whose operator norms are each in $[1/W'', W'']$ for some $W'' \leq \poly(s, W, 1/\beta)$.

Now suppose $\ket{\psi} \in \mc{B}^{\ot n'' }$ is any state satisfying $\bra{\psi}H^{(k'')}\ket{\psi} < \lambda_1(H^{(k'')}) + \beta$.  
Let $\ket{\phi} := R'(\ket{\phi})$, where $\ket{\phi}$ is determined by the measurement outcomes in $R'$.  By the AGPR property of $R'$, with probability at least $1 - \delta'$ over $R'$ we have $\bra{\phi}H^{(k')}\ket{\phi} < \lambda_1(H^{(k')}) + \delta' \leq  \lambda_1(H^{(k')}) + \gamma$.  Condition on this event, and let $\ket{\nu} := R(\ket{\phi})$.  Then with probability at least $1 - \delta$ over $R$, we have $\bra{\nu}H^{(k)}\ket{\nu} < \lambda_1(H^{(k)}) + \gamma \leq \lambda_1(H^{(k)}) + \beta^{\Omega(1)}$.  Thus our reduction $R^*$ satisfies the desired AGPR guarantee, for the value $\delta^* := \delta + \gamma \leq \beta^{\Omega(1)}\cdot \poly(s, W)$.
\end{proof}

Theorem~\ref{thm:new_key} now follows readily from our assembled results.

\begin{proof}[Proof of Theorem~\ref{thm:new_key}]  Let $V(\xi)$ be a verifier circuit as in Theorem~\ref{thm:new_key}'s statement, and let $\eps > 0$ be given such that $\max_{\xi}\bE[V(\xi)] \geq 1 - \eps$.  We apply Theorem~\ref{thm:variant} to $V$ to obtain an $5$-local Hamiltonian $H$ on $N^* = O(T)$ qubits, with $s \leq \poly(T)$ terms of operator norm in the range $[W^{-1}, W]$ for some $W \leq \poly(T)$, and a quantum operation $R$.  


Next, it follows from the combination of Theorem~\ref{thm:app_gs} and Lemma~\ref{lem:compose} (applied twice) that there exists a $(5, 2)$-AGPR $A$.  We apply $A$ to $(H, W, \beta)$, with $\beta \geq \eps^{O(1)}/\poly(T)$ a small value to be determined later.  We obtain a 2-local Hamiltonian $H'$ and associated quantum operation $R'$ (both acting on $\mc{B}^{\ot N'}$, for some $N' \leq \poly(s, 1/\beta) \leq  \poly(T, 1/\eps)$), and a termwise operator norm bound $W' \leq \poly(T, 1/\eps)$ for $H'$.  

For the Hamiltonian $H_{V, \eps}$, we choose the Hilbert space $\mc{B}^{\ot N'}$ and let $H_{V, \eps} := H'$.  For the operation $R_{V, \eps}$, we take the composed measurement $R_{V, \eps} := R \circ R'$.
The efficient constructibility claims in Theorem~\ref{thm:new_key} are satisfied for our choice, by the efficiency properties of Theorem~\ref{thm:variant} and Definition~\ref{thm:app_gs} and the requirement $\beta \geq \eps^{O(1)}/\poly(T)$.  Similarly, the termwise operator-norm bound in Theorem~\ref{thm:new_key} is satisfied.  

Now let $\ket{\psi} \in \mc{B}^{\ot N'}$ be any ground state of $H' = H_{V, \eps}$.  Let $\ket{\phi} := R'(\ket{\psi}) \in \mc{B}^{\ot N^*}$ be the pure state determined by the measurement outcomes in $R'$ applied to $\ket{\psi}$.  By the AGPR property of $R'$, for some $\delta \leq \beta^{\Omega(1)}\cdot \poly(T)$, we have $\Pr_{R'}[\bra{\phi} H \ket{\phi} < \lambda_1(H) + \delta] \geq 1- \delta$. 
We choose $\beta \geq \eps^{O(1)}/\poly(T)$ sufficiently small so that $\delta \leq \eps$.

Consider conditioning on any outcome to $\ket{\phi}$ above such that $\bra{\phi} H \ket{\phi} < \lambda_1(H) + \delta \leq \lambda_1(H) + \eps$.
It follows from the guarantee in Theorem~\ref{thm:variant} that for $\xi := R(\ketbra{\phi}{\phi})$ the verifier satisfies $\bE[V(\xi)]  \geq  1 -  \eps^{\Omega(1)} \cdot \poly(T) $.  Thus, under no conditioning on $\ket{\phi}$ we have
\[ \bE[V(\xi)] \ \geq \ 1 -    \eps^{\Omega(1)} \cdot \poly(T)  - \delta \ \geq \ 1 -   \eps^{\Omega(1)} \cdot \poly(T)    \ .   \]
This proves Theorem~\ref{thm:new_key}.
\end{proof}

\subsection{Proof of Theorem~\ref{thm:app_gs}}


In our proof of Theorem~\ref{thm:app_gs}, we use the perturbative gadgets and analysis ideas of Oliveira and Terhal~\cite{OT}, who build upon work of Kempe, Kitaev and Regev~\cite{kkr}.  Our main effort will be to show that, for any $k \geq 4$, there exists a $(k,\lceil k/2 \rceil)$-AGPR.  This will imply Theorem~\ref{thm:app_gs} the cases $k = 5, 4$.
Then, a slightly different reduction from~\cite{OT} gives a $(3, 2)$-AGPR; this will complete the proof.

\subsection{The Locality-Halving Reduction}\label{sec:construction}

\paragraph{The initial setup:} Fix a constant $k \geq 4$.  As the input to our $(k,\lceil k/2 \rceil)$-AGPR, we are given a tuple $(H_{\targ}, W, \beta)$, where $H_{\targ}$ (which we will call the ``target Hamiltonian'') is a $k$-local Hamiltonian expressed as the sum of some number $s$ of $k$-local terms over an $n$-qubit Hilbert space $\mc{H}_{\comp} \cong \mc{B}^{\ot n}$.  All $k$-local terms of $H$ have operator norms $||H_i|| \in [W^{-1}, W]$.

By standard preprocessing steps, we can and will assume the following:
\begin{itemize}
\item $H_{\targ}$ is a sum of $s' \leq \poly(s)$ terms of form $H_{i} = H_{i, 1} H_{i, 2}\ldots H_{i, k}$, where each $H_{i, a}$ is 1-local\footnote{Here, each $H_{i, a}$ denotes an operator over all of $\mc{H}_{\comp}$, which is the tensor product $H_{i, a} = Y_{i, a} \otimes I_{\mathrm{rest}}$ of an operator $Y_{i, a}$ on the Hilbert space of a single qubit with the identity operator $I_{\mathrm{rest}}$ on the other $n - 1$ qubits.} and $||H_{i, a}|| \leq \poly(s + W)$, and $H_{i, 1}, \ldots, H_{i, k}$ act on distinct qubits (hence they commute).  In the sequel we write $s$ in place of $s'$;
\item For each $i$, we assume $\min \left( || H_{i, 1} H_{i, 2}\ldots H_{i, \lceil k/2 \rceil }||, || H_{i, \lceil k/2 \rceil + 1} \ldots H_{i, k}||\right) \in [1, K]$, for some $K \leq \poly(W)$.  (The lower bound is easily achieved by scaling $H_{\targ}$ by a $\poly(W)$ factor.)
\end{itemize}

To satisfy Definition~\ref{def:agpr}, we will create a $\lceil k/2\rceil$-local derived Hamiltonian $H' = \tilde{H}$ on the larger Hilbert space $\mc{H} = \mc{H}_{\comp} \otimes \mc{H}_{\anc}$.  We refer to $\mc{H}_{\comp}, \mc{H}_{\anc}$ as the \emph{computational} and \emph{ancilla registers}, respectively.  For our quantum operation $R$ as in Definition~\ref{def:agpr}, we will take the operation which simply measures the ancilla register in the standard basis.

%

%
%


\paragraph{Further preprocessing:}  First, we replace $H_{\targ}$ with $H^{\star}_{\targ} := H_{\targ} - M \cdot I$, for some $0 < M \leq \poly(s + W)$ chosen large enough to ensure that $\lambda_1(H^{\star}_{\targ})$ is less than $-1$.  For any $j, \ket{\psi}$ we have
\begin{equation}
\lambda_j(H^{\star}_{\targ}) \ = \ \lambda_j(H_{\targ}) - M   \quad{} \quad{} \text{and}  \quad{} \quad{}   \bra{\psi} H^{\star}_{\targ} \ket{\psi}  \ = \   \bra{\psi} H_{\targ} \ket{\psi}  - M \ .
\end{equation}
$M \cdot I$ can be implemented 1-locally, so $H^{\star}_{\targ}$ is $k$-local.

In the remainder of our work, we will use $H_{\targ}$ to denote $H^{\star}_{\targ}$, so that $\lambda_1(H_{\targ})$ is now assumed to be less than $-1$. 

\vspace{1 em}

\paragraph{The components of $H_{\targ}$:}  For each $i \in [s]$, write $H_{i} = A_{i}B_{i}$, where we have grouped the $k$ factors of $H_{i}$ into a $\lceil k/2\rceil$-local part $A_i$ and a $\lfloor k/2 \rfloor$-local part $B_{i}$.  By our assumption, $||A_{i}||, ||B_{i}|| \geq 1$.

Exploiting cancellations and the fact that $A_i, B_i$ commute, we may write
\begin{equation}
  H_{i}  \ = \ (A_i^2 + B_i^2)/2 - ( - A_i + B_i)^2/2  \ = \  - ( - A_i + B_i)^2/2  + H_{i, \els}  \ ,
\end{equation}
where $H_{i, \els}$ is a sum of $\lceil k/2\rceil$-local and $\lfloor k/2 \rfloor$-local terms.  Let
\begin{equation}
 H_{\els} \ := \   \sum_{i \in [s]} H_{i, \els} \ .   
 \end{equation}

\paragraph{The ancilla register:}  For each index $i \in [s]$ corresponding to a term in $H_{\targ}$, we introduce an ancilla qubit that we refer to as $w(i)$.  Thus $\mc{H}_{\anc}$ consists of $s$ qubits.  For a Hamiltonian $E$ acting on the space of a single qubit, we use $E_{w(i)}$ to denote the application of $E$ to $w(i)$ (tensored with the identity on the rest of $\mc{H}_{\anc}$).  Similarly, for a Hamiltonian $F$ on $s$ qubits we use $F_{\overline{w}}$ to indicate operator on  $\mc{H}_{\anc}$ which applies $F$ to the ordered qubit-set $(w(1), \ldots, w(s))$.

\paragraph{The derived Hamiltonian $\tilde{H}$:}  The construction takes a parameter $0 < \Delta \leq \poly(s, W, 1/\beta)$, to be chosen later as a sufficiently large value.  We will take
\begin{equation}
\tilde{H} \ = \ H_0  + V \ , 
\end{equation}
where
\begin{equation}
H_0 \ := \   \Delta \sum_{i \in [s]}  \ketbra{1}{1}_{w(i)}  \ ,
\end{equation}
and where 
\begin{equation}
V \ := \  H_{\els}  +    \sqrt{\Delta/2} \cdot \sum_{i \in [s]}  (-A_i + B_i)\otimes X_{w(i)}  \ .
\end{equation} 
Here, $X_{w(i)}$ is the Pauli $X$ operator applied to $w(i)$.

When we choose a large value $\Delta$, we will have $||H_0|| \gg ||V||$.  In the analytical framework of~\cite{kkr, OT}, $H_0$ is referred to as the ``unperturbed'' reference Hamiltonian; $V$ as the ``perturbation'' operator, regarded as ``small;'' and $\tilde{H}$ as the ``perturbed'' Hamiltonian, thought of as a slightly deformed version of $H_0$.

\subsection{Some Tools for the Analysis}

\paragraph{The effective Hamiltonian:}  For future use we define
\begin{equation}\label{eq:eff_def}
H_{\eff}\ = \  H_{\targ} \otimes \ketbra{0^s}{0^s}_{\overline{w}}   \ .
\end{equation}
We will show that $\tilde{H}$ ``behaves like'' $H_{\eff}$ in an appropriate sense, hence $H_{\eff}$ is referred to as the ``effective Hamiltonian'' for $\tilde{H}$.

The eigenvalues of $H_{\eff} = H_{\targ} \otimes \ket{0}\bra{0}_{\overline{w}}$ are the same as those of $H_{\targ}$, along with 0.  The introduction of this ``unwanted'' 0 eigenvalue is why we initially applied a global shift to $H_{\targ}$ to assume its eigenvalues are negative, to ensure that the ``lowest-energy part'' of $H_{\targ}$ is preserved.  In particular, we have
\begin{equation}\label{eq:preserve_eig}
\lambda_1(H_{\eff}) \ = \ \lambda_1(H_{\targ}) \ < \ -1 \ ,  \quad{}\quad{} 1 \ < \  ||H_{\eff}||  \ < \ \poly(s + W) \ .
\end{equation}

\paragraph{The eigenspaces of $H_0$, their projectors, and some notation:}  In our analysis, we will use the derived Hamiltonian $H_0$ as a ``reference'' with which we decompose our Hilbert space $\mc{H} = \mc{H}_{\comp} \otimes \mc{H}_{\anc}$.  First, it is obvious from the construction that $H_0$ has only nonnegative eigenvalues, including $0$ and $\Delta$, and with no eigenvalues in $(0, \Delta)$.  We define the subspaces
\begin{equation}
\mc{L}_- \ , \ \mc{L}_+  \ \leq \ \mc{H} \ ,
\end{equation}
where $\mc{L}_-$ is the 0 eigenspace of $H_0$, and $\mc{L}_+ := \mc{L}_-^{\perp}$.  We define $\Pi_-, \Pi_+$ as the projectors onto $\mc{L}_-$ and $\mc{L}_+$; we have the expressions
\begin{equation}
\Pi_- \ = \ \ketbra{0^s}{0^s}_{\overline{w}} \ , \quad{}\quad{}   \Pi_+ \ = \ \sum_{x \in \{0, 1\}^s \setminus 0^s}\ketbra{x}{x}_{\overline{w}} \ .
\end{equation}
Now for \emph{any} operator $A$ on $\mc{H}$, following~\cite{kkr, OT} we define
\begin{equation}
A_{++} \ := \Pi_+ A \Pi_+ \ , \quad{} \quad{} A_{--} \ := \Pi_- A \Pi_- \ , \quad{} \quad{}  A_{+-} \ := \Pi_+ A \Pi_- \ , \quad{} \quad{} A_{-+} \ := \Pi_- A \Pi_+ \ .
\end{equation}
Also define
\begin{equation}
A_{+} \ := A_{++} \ , \quad{} \quad{} A_{-} \ := A_{--} \ .
\end{equation}
The $A_+$ notation will be used when $A(\mc{L}_+) \subseteq \mc{L}_+$, and similarly for $A_-, \mc{L}_-$.

\paragraph{Some perturbation theory definitions:} We will not introduce perturbation theory, only some definitions used here.  The terms we introduce will be defined with reference to the ``unperturbed'' derived Hamiltonian $H_0$, explicitly and through the notation $A_{\pm \pm}$ introduced previously. In one definition we will also make reference to the perturbation operator $V$.

We define three functions 
\[ G, \ \tilde{G} \ , \  \Sigma_- \ , \] 
each of which takes as input a value $z \in \mathbb{C}$ and outputs an operator over $\mc{H}$; the definitions involve matrix inversion and for some values $z$ the output may be undefined.
We define $\tilde{G}$, the \emph{resolvent} of $\tilde{H}$, by
\begin{equation*}  
 \tilde{G}(z) \ := \  (zI - \tilde{H})^{-1} \ . 
\end{equation*}  
Define the \emph{self-energy} $\Sigma_-(z)$ by
\begin{equation}
\Sigma_-(z) \ := \ z I_- - \tilde{G}_{--}^{-1}(z) \ .
\end{equation}

\paragraph{The perturbation theorems:} 
Here we state a result from~\cite{OT} that expresses the sense in which $\tilde{H}$ approximates $H_{\eff}$.  First we introduce one piece of helpful notation. For an operator $A$ over Hilbert space $\mc{H}$ and a subspace $S \leq \mc{H}$, we will use
\[  || A ||_S  \ := \ \max_{\ket{v} \in S \setminus 0} \frac{||A\ket{v}||}{||\ket{v}|| }   \]
to denote the ($\ell_2$) operator norm of $A$ with inputs restricted to $S$.

\begin{theorem}[Special case of~\cite{OT}, Theorem A.1]\label{thm:ot_perturb}  Say we are given Hamiltonians $H_0, \tilde{H}, V, H_{\eff}$ and real values $\Delta > b > 0$, satisfying the following assumptions:
\begin{enumerate}
\item $\tilde{H} = H_0 + V$;
\item $|| V|| <  \Delta/2$;  
\item $H_0$ has the eigenvalues $\{0, \Delta\}$,\footnote{Here, in~\cite{OT}, Theorem A.1 we are fixing the setting $\lambda_* := \Delta/2$, as per the discussion in~\cite[p. 19-20]{OT}.} with $\mc{L}_-, \mc{L}_+$ defined as above relative to $H_0$, and with operators $A_{\pm \pm}$ defined relative to these subspaces;
\item All eigenvalues of $H_{\eff}$ are contained in $[-b, b]$;\footnote{We are setting $a := -b$ in Theorem A.1 of~\cite{OT}.}
\item $H_{\eff} = \Pi_- H_{\eff} \Pi_-$.
\end{enumerate}
Next, fix $r, \eps > 0$, and let $D_r := \{z \in \bC: |z| \leq r\}$ be the disk of radius $r$ in the complex plane, centered at the origin.  Assume that
\begin{equation} b + \eps \ < \   r \ <  \ \Delta/2 \ . \end{equation}
Now our central assumption is that for all $z \in D_r$, the resolvent $\Sigma_-(z)$ is a good approximation to $H_{\eff}$:
\begin{equation}\label{eq:central}
|| \Sigma_-(z) - H_{\eff} || \ \leq \ \eps \ .
\end{equation}
Let
\begin{equation}
\tilde{S} \ \leq \ \mc{H}
\end{equation}
denote the ``low-energy subspace'' of $\tilde{H}$, namely, the subspace generated by the eigenvectors of $\tilde{H}$ whose eigenvalues are less than $\Delta/2$.  Then $\tilde{S}$ has dimension at least 1.  Moreover, it holds that $H_{\eff}$ is well-approximated by $\tilde{H}$ on $\tilde{S}$:
\begin{equation}
|| \tilde{H} -  H_{\eff} ||_{\tilde{S} } \ \leq \     \frac{3(||H_{\eff}|| + \eps)|| \cdot ||V||}{\Delta - || H_{\eff}|| - \eps }  +  \frac{r(r + z_0)\eps}{(r - b)(r - b - \eps)}   \ .
\end{equation} 
\end{theorem}

We will also use the following theorem from~\cite{kkr} relating the spectrum of $\tilde{H}$ to that of $H_{\eff}$:

\begin{theorem}[Special case of \cite{kkr}, Thm. 3; see also~\cite{OT}, Thm. 7]\label{thm:kkr_perturb}
Under the same assumptions as in Theorem~\ref{thm:ot_perturb}, we have the following.  For every index $j$ for which $\lambda_j(\tilde{H}) < \Delta/2$ (in particular, this must include $j = 1$), we have
\begin{equation}
| \lambda_j(\tilde{H})  -  \lambda_j(H_{\eff}) | \ \leq \ \eps \ .
\end{equation}

\end{theorem}

Theorem~\ref{thm:kkr_perturb} is also used in the proof of Theorem~\ref{thm:ot_perturb}.

\subsection{Application of the Perturbation Theorems}

For the construction of $H_0, \tilde{H}, V$ described in Section~\ref{sec:construction}, it is immediate that conditions 1, 3, and 5 in Theorem~\ref{thm:ot_perturb} are satisfied.  Condition 2, asking that $||V|| < \Delta/2$, is satisfied for sufficiently large $\Delta \leq \poly(s, W, 1/\beta)$; this follows by crudely bounding the norms of all terms used to define $V$, using our initial norm-bound assumptions on $H_{\targ}$.

As noted, the eigenvalues of $H_{\eff}$ are the same as those of $H_{\targ}$, along with 0.  Thus we have $||H_{\eff}|| \leq \poly(s + W)$, independent of $\Delta$, and if we take $b := ||H_{\eff}||$, condition 4 in Theorem~\ref{thm:ot_perturb} is satisfied.

Now, to satisfy the last requirement of that Theorem, Eq.~(\ref{eq:central}), we first set $r := 2b + \eps$, with 
\[ \eps := \beta/20\ . \] 
(Recall that $\beta > 0$ is an input parameter to our desired AGPR.)
Thus $D_r$ is a disk of radius $2|| H_{\eff}|| + \eps$ in the complex plane, centered at the origin.

Our key tool is a bound shown in~\cite[p. 11, Eq. (25)]{OT}: for $|z| < \Delta$,
\begin{equation}
\Sigma_-(z) \ = \left( H_{\els} + \frac{\Delta}{2(z - \Delta)}\sum_{i \in [s]} (-A_i + B_i)^2 \right) \otimes \ketbra{0^s}{0^s}_{\overline{w}} + O \left( \frac{||V||^3}{(z - \Delta)^2}\right)  \ .
\end{equation}
Note that for $\Delta \gg z$ the left-hand term approaches $H_{\eff}$ (as defined in Eq.~(\ref{eq:eff_def})), and the right-hand error term approaches 0.  Indeed, following the discussion in~\cite[pp. 11, 20]{OT}, by taking a sufficiently large $\Delta \leq \poly(s + W)/\eps^2$ we obtain
\begin{equation}
||  \Sigma_-(z) - H_{\eff} || \ \leq \ \eps \ ,  \quad{}\quad{} \text{for all } z \in D_r \ .
\end{equation}
Thus all requirements of Theorem~\ref{thm:ot_perturb} are satisfied for our settings, and we conclude that
\begin{align}
|| \tilde{H} -  H_{\eff} ||_{\tilde{S} } \ &\leq \     \frac{3(||H_{\eff}|| + \eps)|| \cdot ||V||}{\Delta - || H_{\eff}|| - \eps }  +  \frac{r(r + z_0)\eps}{(r - b)(r - b - \eps)}  \\ 
&\leq \   \eps  +   4\eps  \ = \ 5 \eps  \ ,  \label{eq:closeness}
\end{align}
with the last inequality valid if we choose $\Delta$ large enough compared to $|| H_{\eff}||$. For future work, we also stipulate that $\Delta$ be chosen large enough to satisfy
\begin{equation}\label{eq:stipulate}
 \frac{1 }{  \Delta } \  \leq \  \frac{\eps}{ 2 ||H_{\eff}||}  \ .
 \end{equation}
All this only requires $\Delta \leq \poly(s, W, 1/\beta)$.


Under the same settings to our parameters, it is immediate that we also obtain the conclusions of Theorem~\ref{thm:kkr_perturb}.  In particular, using Eq.~(\ref{eq:preserve_eig}) we have
\begin{equation}\label{eq:lambdas}
 |\lambda_1(\tilde{H}) -  \lambda_1(H_{\targ})|  \ = \   |\lambda_1(\tilde{H}) -  \lambda_1(H_{\eff})|  \ \leq \ \eps \ .
\end{equation}
For future work, we note that $\tilde{S}$ is a \emph{proper} subspace of $\mc{H}$, since $||\mc{H}|| \geq ||H_0|| - ||V|| \geq \Delta - \Delta/2$.


\paragraph{Consequences for nearly-minimal-energy states:}    Consider any nearly-minimal-energy state $\ket{\psi} \in \mc{H}$ for the Hamiltonian $\tilde{H}$, satisfying
\begin{equation}\label{eq:braket_small}
\bra{\psi}\tilde{H}\ket{\psi} \ < \ \lambda_1(\tilde{H}) + \beta \ < \  \lambda_1(H_{\eff}) + \beta + \eps \ . 
\end{equation}
We will upper-bound $\bra{\psi}H_{\eff}\ket{\psi}$ to show that $\ket{\psi}$ is also nearly-minimal-energy for this second Hamiltonian.  

A small complication for our analysis is that $\ket{\psi}$ may not lie within $\tilde{S}$.  Decompose $\ket{\psi}$ as
\begin{equation}
\ket{\psi} \ = \  \alpha_1 \ket{\psi_{\tilde{S}}} +   \alpha_2 \ket{\psi_{\tilde{S}^{\perp}}}   \ ,
\end{equation}
according to its components in $\tilde{S}$ and its orthogonal complement $\tilde{S}^\perp$ (so, we have $|\alpha_1|^2 + |\alpha_2|^2 = 1$ and $\braket{\psi_{\tilde{S}}}{\psi_{\tilde{S}^{\perp}}  } = 0$).  Recall that both of these spaces have dimension at least 1.  We assume that $a$ is real and positive; this assumption is without loss of generality, by applying a phase factor $\overline{\alpha}_1/|\alpha_1|$ to the state if necessary, and just simplifies our expressions slightly.

By the definition of $\tilde{S}^\perp$, we see that it is spanned by eigenvectors of $\tilde{H}$ with eigenvalues $\geq \Delta/2$.  Thus,
\begin{align}
\bra{\psi}\tilde{H}\ket{\psi} \ &= \  |\alpha_1|^2 \bra{\psi_{\tilde{S}}}\tilde{H}\ket{\psi_{\tilde{S}}}   + |\alpha_2|^2 \bra{\psi_{\tilde{S}^\perp}}\tilde{H}\ket{\psi_{\tilde{S}^\perp}}       \\ 
&\geq \ \lambda_1(\tilde{H}) +  |\alpha_2|^2 \Delta / 2 \ . 
\end{align}
Combining this with Eqs.~(\ref{eq:braket_small}) and~(\ref{eq:lambdas}), we find
\begin{equation}\label{eq:derived_small}
|\alpha_2|^2 \ \leq \ \    \frac{ 2 \beta }{  \Delta }  \  \leq \ \frac{ \eps }{ ||H_{\eff}||  }  \ ,   
\end{equation} 
where the last step follows from our prior largeness requirement on $\Delta$ in Eq.~(\ref{eq:stipulate}).  It also follows that $|\alpha_1 - 1|^2 \leq |\sqrt{1 - \eps / ||H_{\eff}||} - 1|^2 \leq  |1 - \eps/||H_{\eff}|| - 1|^2 \leq  \eps^2 / ||H_{\eff}||^2$ (using here that $\alpha_1 \in \mathbb{R}^+$). For analysis purposes, define the (non-normalized) state
\begin{equation}
\ket{v} \ := \   (\alpha_1 - 1)\ket{\psi_{\tilde{S}}}  +  \alpha_2 \ket{\psi_{\tilde{S}^\perp}}  \ .
\end{equation}
We have 
\begin{equation}\label{eq:vec_small}
||\ket{v}||^2 \  = \  \braket{v}{v} \ = \   |\alpha_1 - 1|^2 +   |\alpha_2|^2 \ \leq \frac{ 2 \eps }{ ||H_{\eff}||   }   \ .
\end{equation} 

Now note that, using the definition of $\ket{v}$ and Eq.~(\ref{eq:vec_small}), we have
\begin{align}
\bra{\psi}H_{\eff}\ket{\psi} \ &=    \bra{\psi_{\tilde{S}}}H_{\eff}\ket{\psi_{\tilde{S}}}   + \bra{v}H_{\eff}\ket{v}    \\
&\leq \    \bra{\psi_{\tilde{S}}}H_{\eff}\ket{\psi_{\tilde{S}}}      +  ||H_{\eff}|| \cdot ||\ket{v}||^2 \\
&\leq \    \bra{\psi_{\tilde{S}}}H_{\eff}\ket{\psi_{\tilde{S}}}     +   4\eps \ .  \label{eq:low_yo}
\end{align}
Next, applying Eq.~(\ref{eq:closeness}) and the fact that $\ket{\psi_{\tilde{S}}} \in \tilde{S}$, we obtain
\begin{align}
\bra{\psi_{\tilde{S}}}H_{\eff}\ket{\psi_{\tilde{S}}} \ &\leq \  \left( \bra{\psi_{\tilde{S}}}\tilde{H}\ket{\psi_{\tilde{S}}}   + ||\tilde{H} - H_{\eff} ||\cdot || \ket{\psi_{\tilde{S}}}||^2\right)  + \eps \\ 
&\leq  \   \bra{\psi_{\tilde{S}}}\tilde{H}\ket{\psi_{\tilde{S}}}     +   6 \eps   \\
&\leq \ \bra{\psi}\tilde{H}\ket{\psi} + 6\eps \\
&\leq \ \lambda_1(H_{\eff}) + 6\eps + \beta \ .   \label{eq:small_yo}
\end{align}
(In the third inequality, we used the definition of $\tilde{S}$ as a low-energy subspace for $\tilde{H}$, and the fact that $\ket{\psi_{\tilde{S}}}$ is the component of $\ket{\psi}$ in $\tilde{S}$.  In the last step, we used Eq.~(\ref{eq:braket_small}).) 
Combining Eqs.~(\ref{eq:low_yo}) and~(\ref{eq:small_yo}), we conclude that
\begin{equation}\label{eq:nearly}
\bra{\psi}H_{\eff}\ket{\psi} \ \leq \    \lambda_1(H_{\eff}) + 10\eps + \beta  < \   \lambda_1(H_{\eff}) + 2\beta \ .
\end{equation}
Thus $\ket{\psi}$ is also nearly-minimal-energy for $H_{\eff} = H_{\targ} \otimes \ketbra{0^s}{0^s}_{\overline{w}}$.

\paragraph{Obtaining a nearly-minimal-energy state for $H_{\targ}$:}  Recall that $\mc{L}_-$ is the subspace of $\mc{H}$ in which the ancilla qubits are all-zero.  Any computational basis state in which the ancillas are not all-zero vanishes under the action of $H_{\eff}$.
For our state $\ket{\psi}$ as above, write 
\begin{equation}
\ket{\psi} \ = \  w \ket{\psi_-} + z \ket{\psi_+} \ ,
\end{equation}
where $\ket{\psi_{-}} \in \mc{L}_-, \ket{\psi_+} \in \mc{L}_+$ are unit vectors.  Re-expressing our inner product in this basis, we have $\bra{\psi}H_{\eff}\ket{\psi} \geq |w|^2 \cdot \lambda_1(H_{\eff}) + 0$, so by Eq.~(\ref{eq:nearly}), and using the facts that $\lambda_1(H_{\eff}) < -1$ and $10\eps + \beta < 1$, we have
\begin{equation}
|w|^2 \ \geq \ 1 - \frac{10\eps + \beta}{|\lambda_1(H_{\eff})|} \ > \ 1 - 2\beta \ .
\end{equation}
Recall that the quantum operation $R$ measures the ancilla register of $\ket{\psi}$.  By the above, with probability $ > 1 - 2\beta$ this measurement yields the all-zero outcome, and the post-measurement state is $\ket{\psi_-}$.  Identifying $\mc{L}_-$ with the Hilbert space $\mc{H}_{\comp}$, on which $H_{\targ}$ acts, we have
\begin{align}
\bra{\psi_-}H_{\targ}\ket{\psi_-} \ &= \  \frac{1}{|w|^2} \bra{\psi}H_{\eff}\ket{\psi}   \\ 
&\leq  \   \lambda_1(H_{\eff}) + 10\eps + \beta   \\
&= \ \lambda_1(H_{\targ}) + 10\eps + \beta \\
&< \ \lambda_1(H_{\targ}) + 2\beta
\end{align}
using Eq.~(\ref{eq:preserve_eig}) in the penultimate step.  Thus $(H' = \tilde{H}, R)$ have the required AGPR properties (where we may take $\delta := 2\beta$ in Definition~\ref{def:agpr}).  We have proved Theorem~\ref{thm:app_gs} for the cases $k = 5, 4$.


%
%
%
%
%

\subsection{The 3-local-to-2-local Reduction}  

Given a 3-local target Hamiltonian $H_{\targ}$, we can use a different gadget construction in~\cite[p. 11-12]{OT}.  The construction uses the same (1-local) unperturbed Hamiltonian $H_0 := \Delta\sum_{i \in [s]}\ketbra{1}{1}_{w(i)}$ and the same effective Hamiltonian $H_{\eff} := H_{\targ} \otimes \ket{0^s}\bra{0^s}_{\overline{w}}$, with a different perturbation Hamiltonian $V$ (this time 2-local), which again satisfies $||V|| < \Delta/2$ for sufficiently large $\Delta \leq \poly(s, W, 1/\beta)$.  As described in~\cite{OT}, for large enough $\Delta \leq \poly(s, W, 1/\beta)$ one can ensure $||\Sigma_-(z) - H_{\eff}|| \leq \eps$ for $\eps := \beta/20$ and for $z$ in a disk of appropriately chosen radius.  This allows us to apply Theorems~\ref{thm:ot_perturb} and~\ref{thm:kkr_perturb} in the same fashion as before.  This yields the required $(3, 2)$-AGPR, completing the proof of Theorem~\ref{thm:app_gs}.

%

\section{Further Implications for Quantum Complexity Theory\label{NERD}}

In this section, we use the $\mathsf{BQP/qpoly}=\mathsf{YQP^*/poly}$\ theorem to
harvest two more results about quantum complexity classes. \ The first is an
``exchange theorem''\ stating that $\mathsf{QCMA/qpoly}\subseteq
\mathsf{QMA/poly}$: in other words, \textit{one can always simulate quantum
advice together with a classical witness by classical advice together with a
quantum witness}. \ This is a straightforward generalization of Theorem
\ref{yqpthm}. \ The second result is a ``Quantum Karp-Lipton Theorem,''\ which
states that if $\mathsf{NP}\subset\mathsf{BQP/qpoly}$\ (that is, $\mathsf{NP}%
$-complete\ problems are efficiently solvable by quantum computers with
quantum advice), then $\mathsf{\Pi}_{\mathsf{2}}^{\mathsf{P}}\subseteq
\mathsf{QMA}^{\mathsf{P{}romiseQMA}}$, which one can think of as ``almost as
bad''\ as a collapse of the polynomial hierarchy. \ This result makes essential
use of Theorem \ref{yqpthm}, and is a good illustration of how that theorem
can be applied in quantum complexity theory.

\begin{theorem}
[Exchange Theorem]\label{qcmaqpoly}$\mathsf{QCMA/qpoly}\subseteq
\mathsf{QMA/poly}$.
\end{theorem}

\begin{proof}
The proof is almost the same as that of Theorem \ref{yqpthm}. \ Let
$L\in\mathsf{QCMA/qpoly}$. \ Then there exists a polynomial-time quantum
verifier $Q$, a family of polynomial-size advice states $\left\{  \rho
_{n}\right\}  _{n}$, and a polynomial $p$ such that for all inputs
$x\in\left\{  0,1\right\}  ^{n}$:

\begin{itemize}
\item $x\in L\implies\exists w\in\left\{  0,1\right\}  ^{p\left(  n\right)
}~\bE\left[  Q\left(  x,w,\rho_{n}\right) \right]  \geq2/3.$

\item $x\notin L\implies\forall w\in\left\{  0,1\right\}  ^{p\left(  n\right)
}~\bE\left[  Q\left(  x,w,\rho_{n}\right) \right]  \leq1/3.$
\end{itemize}

Now consider the following promise problem: given $x$ and $w$\ as input (regarded as two parts of the classical input string), as
well as a constant $c\in\left[  0,1\right]  $, decide whether $\bE\left[
Q\left(  x,w,\rho_{n}\right)  \right]  $\ is at most
$c-1/10$\ or at least $c+1/10$, promised that one of these is the case.
\ (Equivalently, \textit{estimate} the probability within an additive error
$\pm1/10$.) \ This problem is clearly in $\mathsf{P{}romiseBQP/qpoly}$, since
we can take $\rho_{n}$\ as the advice. \ So by Theorem \ref{yqpthm}, the
problem is in $\mathsf{P{}romiseYQP^*/poly}$\ as well, as witnessed by an input-oblivious advice-testing algorithm $Y((x, w), \sigma, a)$ and a classical advice string family $\{a_n\}_{n > 0}$.  (By slight abuse of index notation, the advice string $a_n$ is taken to possess the correctness guarantee in Theorem~\ref{yqpthm} for inputs $(x, w) \in \{0, 1\}^{n + p(n)}$ obeying the promise.)

Our $\mathsf{QMA/poly}%
$\ verifier takes the $\mathsf{P{}romiseYQP^*/poly}$\ advice string $a_{n}$
as its trusted classical advice, and a state of the form $\sigma\otimes\left\vert w\right\rangle
\left\langle w\right\vert $\ as its untrusted witness state. \ It acts as follows:

\begin{enumerate}
\item[(1)] Execute $Y((x, w), \sigma, a_n)$, rejecting if the advice-testing bit $b_{\adv} = 0$;

\item[(2)] If $b_{\adv} = 1$, measure the bit $b_{\out}$ from the same execution of $Y$ and output this bit. 
\end{enumerate}

The protocol is polynomial-time, since $Y$ is a polynomial-time quantum algorithm, and the completeness and soundness properties follow directly from the guarantees of Theorem~\ref{yqpthm}.
\end{proof}

Indeed, let $\mathsf{YQ}$\textsf{\textperiodcentered}$\mathsf{QCMA}$ denote
the complexity class where a $\mathsf{BQP}$ verifier receives a classical untrusted
witness that depends on the input, as well as an untrusted quantum witness that depends
only on the input size $n$. \ Then we can \textit{characterize}
$\mathsf{QCMA/qpoly}$ as equal to $\mathsf{YQ}$\textsf{\textperiodcentered
}$\mathsf{QCMA/poly}$, similarly to how we characterized $\mathsf{BQP/qpoly}$
as equal to $\mathsf{YQP/poly}$.%

\begin{figure}[ptb]%
\centering
\includegraphics[
trim=0in 0in 0in 0in,
scale=.38
]%
{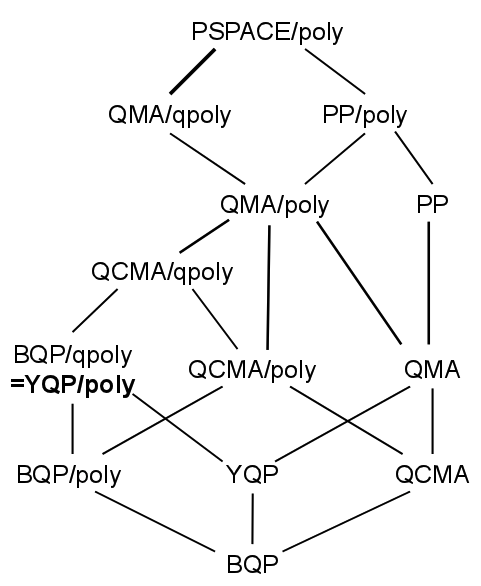}
\caption{Containments among complexity classes related to quantum proofs and
advice, in light of this paper's results. \ The containments
$\mathsf{QMA/qpoly}\subseteq\mathsf{PSPACE/poly}$\ and $\mathsf{QCMA/qpoly}%
\subseteq\mathsf{PP/poly}$\ were shown previously by Aaronson
\cite{aar:qmaqpoly}. \ This paper shows that $\mathsf{BQP/qpoly}%
\subseteq\mathsf{QMA/poly}$, and indeed $\mathsf{BQP/qpoly}=\mathsf{YQP/poly}%
$, where $\mathsf{YQP}$\ is like $\mathsf{QMA}$\ except that the untrusted quantum
witness can depend only on the input length $n$. \ It also shows that
$\mathsf{QCMA/qpoly}\subseteq\mathsf{QMA/poly}$.}%
\label{yqpnew}%
\end{figure}

We now use Theorem \ref{yqpthm}\ to prove an analogue of the Karp-Lipton
Theorem for quantum advice.  

Recall that a promise problem is a pair $\Pi = (\Pi_{yes}, \Pi_{no})$ of disjoint subsets of $\{0, 1\}^*$.
We say that a language $A$ \emph{solves} $\Pi$ if for all $x \in \Pi_{yes} \cup \Pi_{no}$, we have $x \in A \Leftrightarrow x \in \Pi_{yes}$.  We say that a language $L$ is in $\mathsf{QMA}^{\Pi}$ if there is a single $\mathsf{QMA}$ verifier $V^A$ with oracle access, that witnesses the membership $L \in \mathsf{QMA}^{A}$ for \emph{any} language $A$ solving $\Pi$.  We let $\mathsf{QMA}^{\mathsf{P{}romiseQMA}} := \bigcup_{\Pi \in \mathsf{P{}romiseQMA}} \mathsf{QMA}^{\Pi}$.  This model of oracle access to promise problems, in which the machine may query strings violating the promise $\Pi$ (and for which the oracle may give arbitrary responses), is fairly standard; see, e.g.,~\cite{BF99}.

\begin{theorem}
[Quantum Karp-Lipton Theorem]\label{qkl}If $\mathsf{NP}\subset
\mathsf{BQP/qpoly}$, then $\mathsf{\Pi}_{\mathsf{2}}^{\mathsf{P}}%
\subseteq\mathsf{QMA}^{\mathsf{P{}romiseQMA}}$.
\end{theorem}

In this result we use the model of oracle access to a promise problem which allows the algorithm to query inputs not obeying the promise; in such cases the allows the oracle to answer such queries arbitrarily.  This model is fairly standard, see e.g.~\cite{BF99}.

Previously, Aaronson \cite{aar:subtle} showed that if $\mathsf{PP}%
\subset\mathsf{BQP/qpoly}$, then the counting hierarchy $\mathsf{CH}%
$\ collapses. \ However, he had been unable to show that $\mathsf{NP}%
\subset\mathsf{BQP/qpoly}$\ would have unlikely consequences in the uniform world.

\begin{proof}[Proof of Theorem~\ref{qkl}]
By Theorem \ref{yqpthm}, the hypothesis implies $\mathsf{NP}\subset
\mathsf{YQP/poly}=\mathsf{YQP^{\mathsf{\ast}}/poly}$. \ So let $Y$ be a
$\mathsf{YQP^{\mathsf{\ast}}/poly}$ algorithm for $SAT$, which takes an
input $x \in \{0, 1\}^n$ (representing a CNF formula), a trusted classical nonuniform advice string $a \in \{0, 1\}^{\ell(n)}$ for some $\ell(n) \leq \poly(n)$, and an untrusted advice state $\rho$ on $q(n)\leq \poly(n)$ qubits.  By inspecting the proof of Theorem~\ref{new_generalized}, we see that the completeness and soundness parameters $.9, .1$ in Definition~\ref{def:yqp} can easily be strengthened to $(1 - e^{-n}, n^{-100})$; we assume that this holds for $Y$. 
Let $\{a_n\}_{n > 0}$\ be the associated family of classical advice strings of length $\ell(n)$. 

Now consider an arbitrary language $L\in\mathsf{\Pi}_{\mathsf{2}}^{\mathsf{P}%
}$.  As such, $L$ is defined by a deterministic polynomial-time predicate $R\left(  x,y,z\right)
$: 
\[ x\in L \ \Longleftrightarrow \ \forall y\exists z: \ R\left(  x,y,z\right) =1   \ , \] 
where we expect $|y| = |z| = p(n)$ for some $p(n) \leq \poly(n)$ on inputs $x \in \{0, 1\}^n$.

Using $Y$ and Cook's theorem applied to the predicate $R$, we can create a polynomial-time input-oblivious advice-testing algorithm $Y'\left(x, y, \rho, a\right)$ producing output bits $b_{\adv}, b_{\out}$ (we use the notation $Y'_{\adv}, Y'_{\out}\left(x, y, \rho, a\right)$ to denote the values of these two bits in an execution of $Y'$ on $\left(x, y, \rho, a\right)$, noting that $\bE[b_{\adv}]$ depends only on $\rho, a$), which has the following properties:

\begin{itemize}
\item[(P1)] There exists a $\rho\ $such that $\bE\left[  Y'_{\adv}\left(x, y, \rho, a_n\right)\right]  \geq 1 - 2^{-n}$ for all $x, y$.

\item[(P2)] For any $\rho$, if $\bE\left[  Y'_{\adv}\left(x, y, \rho, a_n \right)\right]  \geq  n^{-3} $, we have $\bE\left[ Y'_{\out}\left(x, y, \rho, a_n\right) |  b_{\adv} = 1  \right]  \geq 1 - 1/(n \cdot p(n))$ if there exists a $z$\ such
that $R\left(  x,y,z\right)  $ holds, and $\bE\left[  Y'_{\out}\left(x, y, \rho, a_n\right) | b_{\adv} = 1 \right]  \leq 1/(n \cdot p(n))$\ otherwise.
\end{itemize}

Using the standard search-to-decision reduction for $SAT$, we can then
strengthen property (P2) to the following, for some polynomial-time quantum algorithm
$Y''\left(x, y,   \rho, a\right)$ outputting a bit $b_{\adv}$ (denoted $Y''_{\adv}(x, y, \rho, a )$) and a string $z \in \{0, 1\}^{p(n)}$.\footnote{This reduction requires repeated use of the advice state $\rho$ to obtain the bits of a lexicographically first such $z$; these measurements may alter $\rho$.  This is not a serious obstacle, however, by the principle that a measurement whose outcome is nearly information-theoretically certain has small expected effect on the measured state.}  Here as before, the bit $b_{\adv}$ has expectation determined by $\rho, a$ alone.  The algorithm $Y''$ satisfies:

\begin{itemize}
\item[(P1')] There exists a $\rho\ $such that $\bE\left[  Y''_{\adv}\left(x, y, \rho, a_n\right)\right]  \geq 1 - 2^{-n}$ for all $x, y$.
\item[(P2')] For all $x, y$ pairs for which some $z$ satisfies $R(x, y, z) =1$, and for all states $\rho$, we have the following. If $\bE\left[ Y''_{\adv}\left(x, y, \rho, a_n \right) \right]  \geq  .01$, and if we condition on $[b_{\adv} = 1]$ in this execution, then with probability at least $.99$, $Y''\left(x, y, \rho, a_n \right)$\ outputs a $z$ such that $R(x, y, z) = 1$.
\end{itemize}

Now let $U\left(  x,y, \rho, a\right)  $\ be a quantum algorithm outputting a single bit, and expecting $y, \rho, a$ of size determined by $n = |x|$ exactly as with $Y''$. The algorithm $U$ executes $Y''\left(x, y, \rho, a \right)$ and does one
of the following, both with equal probability:

\begin{itemize}
\item Outputs $\neg b_{\adv}$;

\item Outputs 1 if and only if the string $z$ outputted by $Y''$ satisfies $R\left(  x,y,z \right) =1  $.
\end{itemize}

$U$ is polynomial-time, and we claim that

\begin{itemize}
\item[(A1)] $x\in L\Longrightarrow\exists a,\rho: \ \left[  \bE\left[
Y''_{\adv}\left( x, y,  \rho, a\right)  \right]  \geq 9/10 \right]
\wedge\left[  \forall\sigma,y:  \bE\left[  U\left(  a,\sigma,x,y\right)  \right]  \geq 1/5 \right]  $.

\item[(A2)] $x\notin L\Longrightarrow\forall a,\rho: \ \left[  \bE\left[
Y''_{\adv}\left( x, y,  \rho, a\right)  \right]  \leq 2/3\right]
\vee\left[  \exists\sigma,y~\bE\left[  U\left(  a,\sigma,x,y\right)  \right]  \leq1/6\right]  $.
\end{itemize}

With reference to the machine $U$, we define the promise problem $\Pi = (\Pi_{yes}, \Pi_{no})$ by
\begin{align*}   \Pi_{yes} \ &= \ \{( x, a ) \in \{0, 1\}^{n + \ell(n)} : \exists \rho, y \text{ such that } \bE[U(x, y, \rho, a)] \leq 1/6\} \ ,  \\
   \Pi_{no} \ &= \ \{( x, a ) \in \{0, 1\}^{n + \ell(n)} : \forall \rho, y \text{ we have }  \bE[U(x, y, \rho, a)] \geq 1/5\} \ ,  
   \end{align*}
and note that $\Pi \in \mathsf{P{}romiseQMA}$ by standard techniques.
Also, it is clear that (A1) and (A2) together imply $L\in \mathsf{QMA}^{\Pi} \subseteq \mathsf{QMA}^{\mathsf{P{}romiseQMA}}$.
(The crucial point here is that $U$\ does \textit{not} take the
existentially-quantified advice state $\rho$\ as input in our query to $\Pi$---and therefore, the
$\mathsf{QMA}$\ machine does not need to pass a quantum state to the
$\mathsf{P{}romiseQMA}$\ oracle, which would be illegal. \ This is why we
needed the\ $\mathsf{BQP/qpoly}=\mathsf{YQP^{\ast}/poly}$\ result here.  Note also that in the case where $x \in L$, our claim gives no control over the relevant acceptance probabilities of $Q_1$ and $U$ for settings to $a$ other than the ``correct'' setting; this necessitates the use of a $\mathsf{P{}romiseQMA}$\ oracle---which is allowed to behave arbitrarily on inputs not obeying the promise $\Pi$---rather than a $\mathsf{QMA}$ oracle.)

We now prove (A1) and (A2). \ First suppose $x\in L$. \ Then there exists an
advice string $a_n$\ with the following properties:

\begin{itemize}
\item[(B1)] There exists a $\rho_n$\ such that $\bE\left[  Y''_{\adv}\left(x, y,
\rho_n, a_n\right) \right]  \geq 9/10$ for all $y$. \ (By (P1').)

\item[(B2)] For all $\sigma,y$\ pairs, either $\bE\left[  Y''_{\adv}\left(x, y,
\sigma, a_n\right) \right]  \leq 1/2$, or for the string $z$ outputted by this execution of $Y''$, we have $\Pr\left[ R\left(
x,y, z \right)  \text{ holds}\right]
\geq (.5) \cdot (.99) > 2/5$. \ (By (P2') and the assumption $x\in L$.)
\end{itemize}
By (B2), we have $\forall\sigma,y~\bE\left[  U\left(  a,\sigma,x,y\right)
\right]   \geq 1/5$. \ This proves (A1).

Next suppose $x\notin L$. \ Then given an advice string $a$, suppose there
exists a pair $\rho, y$\ such that $\bE\left[  Y''_{\adv}\left(x, y, \rho, a\right)  \right]  >2/3$.  (Then this relation holds for all $y$, since $\bE[b_{\adv}]$ is a function of $\rho, a$ alone.) \ Set $\sigma:=\rho$, and choose a $y$ for which there
is \textit{no} $z$ such that $R\left(  x,y,z\right)  $\ holds. \ Then for the random string $z$ as produced by $Y''(x, y, \sigma, a)$ we have
$\Pr\left[  R\left(  x,y,z \right)= 1\right]  =0$, since $x\notin L$.

It follows from the above that $\Pr\left[  U\left(  a,\sigma,x,y\right)  \text{
accepts}\right]  < \frac{1}{2}(1/3 + 0) = 1/6$.\ \ This proves (A2), and completes the proof of the Theorem.
\end{proof}

\section{Open Problems\label{OPEN}}


One open problem is simply to find more applications of the
majority-certificates lemma, which seems likely to have uses outside of
quantum complexity theory. Can we improve the parameters of the
majority-certificates lemma (the size of the certificates or the number
$O\left(  n\right)  $\ of certificates), or alternatively, show that the
current parameters are essentially optimal? \ Also, can we prove the
real-valued majority-certificates lemma with an error tolerance $\alpha$\ that
depends only on the desired accuracy $\varepsilon$ of the final approximation,
not on $n$ or the fat-shattering dimension of $S$?

On the quantum complexity side, we mention several questions. \ First, in
Theorem \ref{localthm}, is the polynomial blowup in the number of qubits
unavoidable? \ Could one hope for a way to simulate an $n$-qubit advice state
by the ground state of $n$-qubit local Hamiltonian, or would that have
implausible complexity consequences? \ Second, can we use the ideas in this
paper to prove any upper bound on the class $\mathsf{QMA/qpoly}$ better than
the $\mathsf{PSPACE/poly}$\ upper bound shown by Aaronson \cite{aar:qmaqpoly}?
\ Third, if $\mathsf{NP}\subset\mathsf{BQP/qpoly}$, then does $\mathsf{QMA}%
^{\mathsf{P{}romiseQMA}}$\ contain not just $\mathsf{\Pi}_{\mathsf{2}%
}^{\mathsf{P}}$\ but the entire polynomial hierarchy? Finally, is
$\mathsf{BQP/qpoly} = \mathsf{BQP/poly}$?

\section{Acknowledgments}

We thank Sanjeev Arora, Kai-Min Chung, Avinatan Hassidim, Ashwin Nayak, Roberto Oliveira, Thomas Vidick, John
Watrous, and Colin Zheng for helpful comments and discussions, and the
anonymous reviewers for their comments. We are particularly grateful to a journal reviewer who pointed out the need for further analysis of the Local Hamiltonian reductions used in our work.

%

\appendix

\section{Appendix: Untrusted Oracles\label{UORACLE}}

In this appendix, we give an interesting consequence of the
majority-certificates lemma for classical complexity theory.

When we give a machine an oracle, normally we assume the oracle can be
trusted. \ But it is also natural to consider \textit{untrusted} oracles,
which are nevertheless restricted in their computational power. \ We formalize
this notion as follows:

\begin{definition}
[Untrusted Oracles]Let $\mathcal{C}$\ and $\mathcal{D}$\ be complexity
classes. \ Also, given a family $a=\left\{  a_{n}\right\}  _{n\geq1}$\ of
$p\left(  n\right)  $-bit advice strings and a machine $V$, let $V\left[
a\right]  $\ be the language decided by $V$ given $a$ as advice.\ \ Then
$\mathcal{C}^{\mathsf{Untrusted}\text{-}\mathcal{D}}$\ is the class of
languages $L$ for which there exists a $\mathcal{C}$\ machine $U$, a
$\mathcal{D}$\ machine $V$, and a polynomial $p$ such that for all $n$:

\begin{enumerate}
\item[(i)] There exist $p\left(  n\right)  $-bit\ advice strings\ $a_{1}%
,\ldots,a_{m}$\ such that $U^{V\left[  a_{1}\right]  ,\ldots,V\left[
a_{m}\right]  }$\ decides $L$.

\item[(ii)] $U^{V\left[  a_{1}\right]  ,\ldots,V\left[  a_{m}\right]  }\left(
x\right)  $\ outputs either $L\left(  x\right)  $\ or ``FAIL,'' for all inputs
$x\in\left\{  0,1\right\}  ^{n}$ and all $p\left(  n\right)  $-bit\ advice
strings $a_{1},\ldots,a_{m}$.
\end{enumerate}
\end{definition}

We can now state the consequence.

\begin{theorem}
\label{untrustedthm}Let $\mathcal{C}$ be a uniform syntactic complexity class,
such as $\mathsf{P}$, $\mathsf{NP}$, or $\mathsf{EXP}$. \ Then $\mathcal{C}%
\mathsf{/poly}\subseteq\left(  \mathsf{AC}^{0}\right)  ^{\mathsf{Untrusted}%
\text{-}\mathcal{C}}$.
\end{theorem}

\begin{proof}
Let $V$\ be a $\mathcal{C}\mathsf{/poly}$\ machine that uses a family
$a=\left\{  a_{n}\right\}  _{n\geq1}$\ of $p\left(  n\right)  $-bit advice
strings. \ Fix an input length $n$, and let $f_{w}\left(  x\right)  $\ be the
output of $V$ on input $x$\ and advice string $w\in\left\{  0,1\right\}
^{p\left(  n\right)  }$. \ Then $S=\left\{  f_{w}\right\}  _{w\in\left\{
0,1\right\}  ^{p\left(  n\right)  }}$\ is a Boolean concept class of size
$\left\vert S\right\vert \leq2^{\operatorname*{poly}\left(  n\right)  }$. \ So
by Lemma \ref{majcer}, there exist $m=O\left(  n\right)  $\ polynomial-size
certificates $C_{1},\ldots,C_{m}$, which isolate functions $f_{1},\ldots
,f_{m}\in S$\ respectively such that $\operatorname*{MAJ}\left(  f_{1}%
,\ldots,f_{m}\right)  =f_{a_{n}}$. \ Now, we can easily modify the proof of
Lemma \ref{majcer} to ensure not only that $\operatorname*{MAJ}\left(
f_{1},\ldots,f_{m}\right)  =f^{\ast}$, but also that%
\begin{align*}
f_{a_{n}}\left(  x\right)   &  =1\Longrightarrow f_{1}\left(  x\right)
+\cdots+f_{m}\left(  x\right)  \geq\frac{2m}{3},\\
f_{a_{n}}\left(  x\right)   &  =0\Longrightarrow f_{1}\left(  x\right)
+\cdots+f_{m}\left(  x\right)  \leq\frac{m}{3}%
\end{align*}
for all inputs $x$. \ To do so, we simply take $m=O\left(  n\right)
$\ sufficiently large and redo the Chernoff bound. \ Furthermore, it is known
that \textsc{Approximate Majority}---that is, \textsc{Majority} where the
fraction of $1$'s in the input is bounded away from $1/2$ by a constant---can
be computed by polynomial-size depth-$3$ circuits, so in particular, in
$\mathsf{AC}^{0}$ (see Viola \cite{viola:maj}\ for example).

By hardwiring the certificates $C_{1},\ldots,C_{m}$\ into the $\mathsf{AC}%
^{0}$\ circuit, we can produce an $\mathsf{AC}^{0}$\ circuit\ that first
checks whether $f_{i}$\ is consistent with $C_{i}$\ for all $i\in\left[
m\right]  $, outputs ``FAIL''\ if not, and otherwise outputs $U^{f_{1}%
,\ldots,f_{m}}\left(  x\right)  =f_{a_{n}}\left(  x\right)  $.
\end{proof}

If $\mathcal{C}$\ is a semantic complexity class, such as $\mathsf{BPP}$ or
$\mathsf{UP}$, the difficulty is that there might be a $\mathcal{C}%
\mathsf{/poly}$\ machine $M$\ and advice string $w$ for which the function
$f_{w}$\ is undefined (since $M$ need not decide a language for every $w$).
\ However, if we force the $\mathsf{Untrusted-}\mathcal{C}$\ oracle to
restrict itself to $w$\ for which $f_{w}$\ \textit{is} defined,\ then Theorem
\ref{untrustedthm}\ goes through for semantic classes as well. \ Using the
\textit{real}-valued majority-certificates lemma that we develop in Section
\ref{REALMAJ},\ it is possible to remove the assumption that $f_{w}$\ is
defined\ for all $w$ for semantic classes such as $\mathsf{BPP}$.

\section{Appendix: Isolatability and Learnability\label{ISOLNLEARN}}

The following definition abstracts a key notion from the majority-certificates
lemma. \ 

\begin{definition}
[Majority-Isolatability]A Boolean concept class $S$ \textit{is
majority-isolatable if for every }$f\in S$\textit{, there exist }%
$m=\operatorname*{poly}\left(  n\right)  $\textit{\ certificates }%
$C_{1},\ldots,C_{m}$\textit{, each of size }$\operatorname*{poly}\left(
n\right)  $\textit{, such that}

\begin{enumerate}
\item[(i)] $S\left[  C_{i}\right]  $ is nonempty for all $i\in\left[
m\right]  $, and

\item[(ii)] if $f_{i}\in S\left[  C_{i}\right]  $\ for all $i\in\left[
m\right]  $, then $\operatorname*{MAJ}\left(  f_{1},\ldots,f_{m}\right)  =f$,
where $\operatorname*{MAJ}$\ denotes pointwise majority.
\end{enumerate}
\end{definition}

We now show that the majority-isolatability of a Boolean concept class $S$ is
equivalent to a large number of other properties of $S$---including having
singly-exponential cardinality, having polynomial VC-dimension, being
PAC-learnable using $\operatorname*{poly}\left(  n\right)  $\ samples, and
being ``winnowable.''\ \ While we do not need this equivalence theorem elsewhere
in the paper, we feel it has independent interest. \ The equivalence
theorem we prove is easily seen to break down for concept classes with infinite input domains.

\begin{definition}
[VC-dimension]We say a Boolean concept class $S$ \textit{shatters} the set
$A\subseteq\left\{  0,1\right\}  ^{n}$\ if for all $2^{\left\vert A\right\vert
}$\ functions $g:A\rightarrow\left\{  0,1\right\}  $, there exists an $f\in
S$\ whose restriction to $A$ equals $g$. \ Then the VC-dimension of $S$, or
$\operatorname*{VCdim}\left(  S\right)  $, is the size of the largest set
shattered by $S$.
\end{definition}

Given a distribution $\mathcal{D}$ over $\left\{  0,1\right\}  ^{n}$, we say
the Boolean functions $f,g:\left\{  0,1\right\}  ^{n}\rightarrow\left\{
0,1\right\}  $\ are $\left(  \mathcal{D},\varepsilon\right)  $\textit{-close}%
\ if%
\[
\Pr_{x\sim\mathcal{D}}\left[  g\left(  x\right)  =f\left(  x\right)  \right]
\geq1-\varepsilon.
\]

\begin{definition}
[Learnability]$S$ is learnable if for all $f\in S$, distributions
$\mathcal{D}$, and $\varepsilon,\delta>0$, there exists an
$m=\operatorname*{poly}\left(  n,1/\varepsilon,\log1/\delta\right)  $\ such
that with probability at least $1-\delta$\ over sample points $x_{1}%
,\ldots,x_{m}$\ drawn independently from $\mathcal{D}$, every $g\in
S$\ satisfying $g\left(  x_{1}\right)  =f\left(  x_{1}\right)  ,\ldots
,g\left(  x_{m}\right)  =f\left(  x_{m}\right)  $\ is $\left(  \mathcal{D}%
,\varepsilon\right)  $-close to $f$.
\end{definition}

We can also define ``approximability,''\ which is like learnability except that
the choice of training examples can be nondeterministic:

\begin{definition}
[Approximability]$S$ is approximable if for all $f\in S$ and distributions
$\mathcal{D}$, there exists a certificate $C$\ of size $\operatorname*{poly}%
\left(  n,1/\varepsilon\right)  $\ such that every $g\in S\left[  C\right]  $
is $\left(  \mathcal{D},\varepsilon\right)  $-close to $f$.
\end{definition}

Finally, let us call attention to a notion that implicitly played a major role
in the proof of Lemma \ref{majcer}.

\begin{definition}
[Winnowability]$S$ is winnowable if for all nonempty subsets $S^{\prime
}\subseteq S$, there exists a certificate $C$\ of size $\operatorname*{poly}%
\left(  n\right)  $\ such that $\left\vert S^{\prime}\left[  C\right]
\right\vert =1$.
\end{definition}

We can now prove the equivalence theorem.

\begin{theorem}
\label{equivthm}Let $S$ be a Boolean concept class. \ Then $\left\vert
S\right\vert \leq2^{\operatorname*{poly}\left(  n\right)  }$ iff
$\operatorname*{VCdim}\left(  S\right)  \leq\operatorname*{poly}\left(
n\right)  $ iff $S$ is learnable iff $S$ is approximable iff $S$ is
majority-isolatable iff $S$ is winnowable.
\end{theorem}

\begin{proof}
$\left\vert S\right\vert \leq2^{\operatorname*{poly}\left(  n\right)
}\Longrightarrow\operatorname*{VCdim}\left(  S\right)  \leq
\operatorname*{poly}\left(  n\right)  $ follows from the trivial upper bound
$\operatorname*{VCdim}\left(  S\right)  \leq\log_{2}\left\vert S\right\vert $.

$\operatorname*{VCdim}\left(  S\right)  \leq\operatorname*{poly}\left(
n\right)  \Longrightarrow\left\vert S\right\vert \leq2^{\operatorname*{poly}%
\left(  n\right)  }$ is Sauer's Lemma \cite{sauer}, which implies the relation
$\left\vert S\right\vert \leq2^{n\operatorname*{VCdim}\left(  S\right)  }$.

$\left\vert S\right\vert \leq2^{\operatorname*{poly}\left(  n\right)
}\Longrightarrow$\textbf{Learnable} was proved by Valiant \cite{valiant:pac}.

\textbf{Learnable}$\Longrightarrow$\textbf{Approximable} is immediate, and
\textbf{Approximable}$\Longrightarrow\operatorname*{VCdim}\left(  S\right)
\leq\operatorname*{poly}\left(  n\right)  $ follows from a counting argument
(see Blumer et al.\ \cite{behw}\ for details).

$\left\vert S\right\vert \leq2^{\operatorname*{poly}\left(  n\right)
}\Longrightarrow$\textbf{Majority-Isolatable} was the content of Lemma
\ref{majcer}.

\textbf{Majority-Isolatable}$\Longrightarrow\left\vert S\right\vert
\leq2^{\operatorname*{poly}\left(  n\right)  }$\ follows from another counting
argument: if $S$ is majority-isolatable, then every $f\in S$\ is uniquely
determined by $\operatorname*{poly}\left(  n\right)  $ certificates
$C_{1},\ldots,C_{m}$, each of which can be specified using
$\operatorname*{poly}\left(  n\right)  $ bits.

For $\left\vert S\right\vert \leq2^{\operatorname*{poly}\left(  n\right)
}\Longrightarrow$\textbf{Winnowable}, let $S^{\prime}\subseteq S$. \ Then as
in the proof of Lemma \ref{majcer}, we can use binary search to winnow
$S^{\prime}$\ down to a single function $f\in S^{\prime}$, which yields a
certificate of size at most $\log_{2}\left\vert S^{\prime}\right\vert \leq
\log_{2}\left\vert S\right\vert $.

For \textbf{Winnowable}$\Longrightarrow\left\vert S\right\vert \leq
2^{\operatorname*{poly}\left(  n\right)  }$, we prove the contrapositive.
\ Suppose $\left\vert S\right\vert \geq2^{t\left(  n\right)  }$ for some
superpolynomial function $t\left(  n\right)  $ (at least, for infinitely many
$n$). \ Then define a subset $S^{\prime}\subseteq S$ by the following
iterative procedure. \ Initially $S^{\prime}=S$. \ Then so long as there
exists a certificate $C$ of size at most $t\left(  n\right)  /\left(
2n+2\right)  $\ such that $\left\vert S^{\prime}\left[  C\right]  \right\vert
=1$, remove the function $f\in S^{\prime}\left[  C\right]  $ from $S^{\prime}%
$, halting only when no more such ``isolating certificates''\ can be found.

The number of certificates of size $k$ is at most $2^{\left(  n+1\right)  k}$,
and a given certificate $C$\ can only be chosen once, since thereafter
$S^{\prime}\left[  C\right]  $\ is empty. \ So when the above procedure halts,
we are left with a set $S^{\prime}$ such that $\left\vert S^{\prime
}\right\vert \geq2^{t\left(  n\right)  }-2^{\left(  n+1\right)  t\left(
n\right)  /\left(  2n+2\right)  }>0$. \ Furthermore, for every function
$f$\ remaining in $S^{\prime}$, there can be no polynomial-size certificate
$C$ such that $S^{\prime}\left[  C\right]  =\left\{  f\right\}  $---for if
there were, then we would already have eliminated $f$ in the process of
forming $S^{\prime}$. \ Hence $S$\ is not winnowable.
\end{proof}

\section{Appendix: Winnowing of p-Concept Classes\label{WINNOW}}

In this appendix, we look more closely at the problem solved by Lemma
\ref{linfwinnow} (the ``Safe Winnowing Lemma''),\ and ask in what senses it is
possible to winnow a p-concept class down to ``essentially''\ just one function.
\ The answer turns out to be interesting, even though we do not need it for
our quantum complexity applications.

We first give a definition that abstracts part of what Lemma \ref{linfwinnow}
was trying to accomplish.

\begin{definition}
[Winnowability]A p-concept class $S$ is $L_{1}$-winnowable if the following
holds. \ For all nonempty subsets $S^{\prime}\subseteq S$ and $\varepsilon>0$,
there exists a function $f\in S^{\prime}$, a set $X\subseteq\left\{
0,1\right\}  ^{n}$\ of size $\operatorname*{poly}\left(  n,1/\varepsilon
\right)  $, and a $\delta=\operatorname*{poly}\left(  \varepsilon\right)
$\ such that every $g\in S^{\prime}$ that satisfies $\Delta_{1}\left(
f,g\right)  \left[  X\right]  \leq\delta$\ also satisfies $\Delta_{\infty
}\left(  f,g\right)  \leq\varepsilon$. \ Likewise, $S$ is $L_{2}%
$-winnowable\ if $\Delta_{2}\left(  f,g\right)  \left[  X\right]  \leq\delta
$\ implies $\Delta_{\infty}\left(  f,g\right)  \leq\varepsilon$, and
$L_{\infty}$-winnowable\ if $\Delta_{\infty}\left(  f,g\right)  \left[
X\right]  \leq\delta$\ implies $\Delta_{\infty}\left(  f,g\right)
\leq\varepsilon$.
\end{definition}

Clearly $L_{\infty}$-winnowability\ implies $L_{2}$-winnowability implies
$L_{1}$-winnowability. \ The following lemma will imply that every set of
functions with a small cover is $L_{1}$-winnowable.

\begin{lemma}
[$L_{1}$-Winnowing Lemma]\label{l1lem}Let $S$ be a set of functions
$f:\left\{  0,1\right\}  ^{n}\rightarrow\left[  0,1\right]  $. \ For some
parameter $\varepsilon>0$, let $C$ be a finite $\varepsilon$-cover for $S$.
\ Then there exists an $f\in S$, as well as a subset $X\subseteq\left\{
0,1\right\}  ^{n}$\ of size $O\left(  \frac{1}{\varepsilon}\log\left\vert
C\right\vert \right)  $,\ such that every $g\in S$\ that satisfies $\Delta
_{1}\left(  f,g\right)  \left[  X\right]  \leq0.4\varepsilon$\ also satisfies
$\Delta_{\infty}\left(  f,g\right)  \leq2\varepsilon$.
\end{lemma}

\begin{proof}
We will consider functions\ $P:S\rightarrow\left[  0,1\right]  $, which we
think of as assigning a probability weight $P\left(  g\right)  $\ to each
function $g\in S$. \ In particular, given an $f\in S$\ and a subset of inputs
$X\subseteq\left\{  0,1\right\}  ^{n}$, define%
\[
P_{f,X}\left(  g\right)  :=\exp\left(  -\Delta_{1}\left(  f,g\right)  \left[
X\right]  \right)  .
\]
Clearly $P_{f,X}\left(  f\right)  =1$. \ Our goal will be to find $f\in
S$\ and $X\subseteq\left\{  0,1\right\}  ^{n}$, with $\left\vert X\right\vert
=O\left(  \frac{1}{\varepsilon}\log\left\vert C\right\vert \right)  $, such
that every $g\in S$ that satisfies $P_{f,X}\left(  g\right)  \geq
e^{-0.4\varepsilon}$ also satisfies $\Delta_{\infty}\left(  f,g\right)
\leq2\varepsilon$. \ Supposing we have found such an $\left(  f,X\right)
$\ pair, the lemma is proved.

Consider the progress measure%
\[
M_{f,X}:=\sum_{h\in C}P_{f,X}\left(  h\right)  .
\]
Clearly $M_{f,X}\leq\left\vert C\right\vert $ for all $\left(  f,X\right)  $.
\ We claim, furthermore, that $M_{f,X}\geq\exp\left(  -\varepsilon\left\vert
X\right\vert \right)  $\ for all $\left(  f,X\right)  $. \ For since $C$\ is
an $\varepsilon$-cover for $S$, there always exists an $h\in C$\ such that
$\Delta_{1}\left(  f,h\right)  \left[  X\right]  \leq\varepsilon\left\vert
X\right\vert $, and that $h$ alone contributes at least $\exp\left(
-\varepsilon\left\vert X\right\vert \right)  $\ to $M_{f,X}$.

We will construct $\left(  f,X\right)  $\ by an iterative process. \ Initially
$f$ is arbitrary and $X$ is the empty set, so $P_{f,X}\left(  g\right)
=1$\ for all $g$, and $M_{f,X}=\left\vert C\right\vert $. \ Now, suppose there
exists a $g\in S$\ such that $P_{f,X}\left(  g\right)  \geq e^{-0.4\varepsilon
}$, as well as an input $y$ such that $\left\vert f\left(  y\right)  -g\left(
y\right)  \right\vert >2\varepsilon$. \ As a first step, let $Y:=X\cup\left\{
y\right\}  $\ (that is, add $y$ into our set of inputs). \ Then the crucial
claim is that either $M_{f,Y}$\ or $M_{g,Y}$\ is a $1-\Omega\left(
\varepsilon\right)  $ factor smaller than $M_{f,X}$. \ This means in
particular that, by replacing $X$ with $Y$ (increasing $\left\vert
X\right\vert $\ by $1$), and possibly also replacing $f$ with $g$, we can
decrease $M_{f,X}$\ by a $1-\Omega\left(  \varepsilon\right)  $ factor
compared to its previous value. \ Since $\exp\left(  -\varepsilon\left\vert
X\right\vert \right)  \leq M_{f,X}\leq\left\vert C\right\vert $, it is clear
that $M_{f,X}$ can decrease in this way at most%
\[
O\left(  \log_{1+\varepsilon}\frac{\left\vert C\right\vert }{\exp\left(
-\varepsilon\left\vert X\right\vert \right)  }\right)
\]
times. \ Setting the above expression equal to $\left\vert X\right\vert $\ and
solving, we find that the process must terminate when $\left\vert X\right\vert
=O\left(  \frac{1}{\varepsilon}\log\left\vert C\right\vert \right)  $,
returning an $\left(  f,X\right)  $\ pair with the properties we want.

We now prove the crucial claim. \ The first step is to show that either%
\[
M_{f,Y}=\sum_{h\in C}P_{f,X}\left(  h\right)  e^{-\left\vert f\left(
y\right)  -h\left(  y\right)  \right\vert }%
\]
or else%
\[
M^{\prime}:=\sum_{h\in C}P_{f,X}\left(  h\right)  e^{-\left\vert g\left(
y\right)  -h\left(  y\right)  \right\vert }%
\]
is at most%
\[
\frac{1+e^{-\varepsilon}}{2}M_{f,X}.
\]
For since $\left\vert f\left(  y\right)  -g\left(  y\right)  \right\vert
>2\varepsilon$, either $\left\vert f\left(  y\right)  -h\left(  y\right)
\right\vert >\varepsilon$\ or $\left\vert g\left(  y\right)  -h\left(
y\right)  \right\vert >\varepsilon$ by the triangle inequality. \ So for every
$y$, either $e^{-\left\vert f\left(  y\right)  -h\left(  y\right)  \right\vert
}<e^{-\varepsilon}$\ or $e^{-\left\vert g\left(  y\right)  -h\left(  y\right)
\right\vert }<e^{-\varepsilon}$. \ This in turn means that either $M_{f,Y}$ or
$M^{\prime}$\ must have at least half its terms (as weighted by the
$P_{f,X}\left(  h\right)  $'s) shrunk by an $e^{-\varepsilon}$\ factor.

If $M_{f,Y}<\frac{1+e^{-\varepsilon}}{2}M_{f,X}$\ then we are done. \ So
suppose instead that $M^{\prime}<\frac{1+e^{-\varepsilon}}{2}M_{f,X}$. \ Then%
\begin{align*}
M_{g,Y}  &  =\sum_{h\in C}P_{g,X}\left(  h\right)  e^{-\left\vert g\left(
y\right)  -h\left(  y\right)  \right\vert }\\
&  \leq M^{\prime}\max_{h\in C}\frac{P_{g,X}\left(  h\right)  }{P_{f,X}\left(
h\right)  }\\
&  =M^{\prime}\max_{h\in C}\frac{\exp\left(  -\Delta_{1}\left(  g,h\right)
\left[  X\right]  \right)  }{\exp\left(  -\Delta_{1}\left(  f,h\right)
\left[  X\right]  \right)  }\\
&  \leq M^{\prime}\exp\left(  \Delta_{1}\left(  f,g\right)  \left[  X\right]
\right) \\
&  =\frac{M^{\prime}}{P_{f,X}\left(  g\right)  }\\
&  <\frac{\frac{1+e^{-\varepsilon}}{2}M_{f,X}}{e^{-0.4\varepsilon}}\\
&  <\left(  1-\frac{\varepsilon}{20}\right)  M_{f,X}%
\end{align*}
and we are done.
\end{proof}

Recall that $S$ is \textit{coverable} if for all $\varepsilon>0$, there exists
an $\varepsilon$-cover\ for $S$\ of size $2^{\operatorname*{poly}\left(
n,1/\varepsilon\right)  }$. \ We can now prove the following equivalence theorem.

\begin{theorem}
\label{l1winnowthm}A p-concept class $S$ is coverable if and only if it is
$L_{1}$-winnowable.
\end{theorem}

\begin{proof}
For \textbf{Coverable}$\Longrightarrow$\textbf{L}$_{1}$\textbf{-Winnowable}:
fix a subset $S^{\prime}\subseteq S$\ and an $\varepsilon>0$. \ Let $C$\ be an
$\varepsilon/2$-cover\ for $S^{\prime}$ of size $2^{\operatorname*{poly}%
\left(  n,1/\varepsilon\right)  }$. \ Then by Lemma \ref{l1lem}, there exists
an $f\in S^{\prime}$, as well as a subset $X\subseteq\left\{  0,1\right\}
^{n}$\ of size $O\left(  \frac{1}{\varepsilon}\log\left\vert C\right\vert
\right)  =\operatorname*{poly}\left(  n,1/\varepsilon\right)  $,\ such that
every $g\in S^{\prime}$\ that satisfies $\Delta_{1}\left(  f,g\right)  \left[
X\right]  \leq\varepsilon/5$\ also satisfies $\Delta_{\infty}\left(
f,g\right)  \leq\varepsilon$.

For \textbf{L}$_{1}$-\textbf{Winnowable}$\Longrightarrow$\textbf{Coverable},
we prove the contrapositive. \ Suppose there exists a function $t\left(
n,1/\varepsilon\right)  $, superpolynomial in either $n$\ or $1/\varepsilon$,
such that $S$ has no $\varepsilon$-cover of size $2^{t\left(  n,1/\varepsilon
\right)  }$ (at least, for infinitely many $n$ or $1/\varepsilon$). \ Let
$p=\operatorname*{poly}\left(  n,1/\varepsilon\right)  $ and $\delta
=\operatorname*{poly}\left(  \varepsilon\right)  $. \ Given a function $f$ and
subset $X\subseteq\left\{  0,1\right\}  ^{n}$, let $L\left[  f,X\right]  $\ be
the set of all functions $g$\ such that $\Delta_{1}\left(  f,g\right)  \left[
X\right]  \leq\delta$. \ Then our goal is to construct a subset $S^{\prime
}\subseteq S$\ for which there is no pair $\left(  f,X\right)  $ such that

\begin{itemize}
\item $f\in S^{\prime}$,

\item $X\subseteq\left\{  0,1\right\}  ^{n}$\ is a set of inputs with
$\left\vert X\right\vert =p$, and

\item $g\in S^{\prime}\cap L\left[  f,X\right]  $ implies $\Delta_{\infty
}\left(  f,g\right)  \leq\varepsilon$.
\end{itemize}

Let $W:=\left\lceil 2p/\delta\right\rceil $. \ Also, call a set $B$ of
functions $f:\left\{  0,1\right\}  ^{n}\rightarrow\left[  0,1\right]  $\ a
\textit{sliver} if there exists a set $X\subseteq\left\{  0,1\right\}  ^{n}%
$\ with $\left\vert X\right\vert =p$, as well a function $a:X\rightarrow
\left[  W\right]  $, such that%
\[
f\in B\Longleftrightarrow f\left(  x\right)  \in\left[  \frac{a\left(
x\right)  -1}{W},\frac{a\left(  x\right)  }{W}\right]  ~\forall x\in X.
\]
Then define a subset $S^{\prime}\subseteq S$ by the following iterative
procedure. \ Initially $S^{\prime}=S$. \ Then so long as there exists a sliver
$B$ such that $S^{\prime}\cap B$ is nonempty, together with a function
$f_{B}\in S$ such that%
\[
g\in S^{\prime}\cap B\Longrightarrow\Delta_{\infty}\left(  f_{B},g\right)
\leq\varepsilon,
\]
remove $B$ from $S^{\prime}$\ (that is, set $S^{\prime}:=S^{\prime}\setminus
B$). \ Halt only when no more such slivers $B$ can be found.

As a first observation, the total number of slivers is at most $\left(
2^{n}W\right)  ^{p}=2^{\operatorname*{poly}\left(  n,1/\varepsilon\right)  }$.
\ Thus, the above procedure must halt after at most $2^{\operatorname*{poly}%
\left(  n,1/\varepsilon\right)  }$\ iterations.

As a consequence, we claim that $S^{\prime}$\ must be nonempty after the
procedure has halted. \ For suppose not. \ Then the sequence of functions
$f_{B}$\ chosen by the procedure would form an $\varepsilon$-cover\ for
$S$\ of size $2^{\operatorname*{poly}\left(  n,1/\varepsilon\right)  }%
$---since for all $g\in S$, we would simply need to find a sliver $B$
containing $g$ that was removed by the procedure; then $f_{B}$\ would satisfy
$\Delta_{\infty}\left(  f_{B},g\right)  \leq\varepsilon$. \ But this
contradicts the assumption that no such $\varepsilon$-cover exists.

Finally, we claim that once the procedure halts, there can be no $f\in
S^{\prime}$\ and set $X$\ of $p$ inputs\ such that $\Delta_{\infty}\left(
f,g\right)  \leq\varepsilon$\ for all $g\in S^{\prime}\cap L\left[
f,X\right]  $. \ For suppose to the contrary that such an $\left(  f,X\right)
$\ pair existed. \ It is not hard to see that for every $\left(  f,X\right)
$, there exists a sliver $B$ that contains $f$\ and is contained in $L\left[
f,X\right]  $. \ But then $S^{\prime}\cap B$\ would be nonempty, and $\left(
B,f\right)  $\ would satisfy the condition $g\in S^{\prime}\cap
B\Longrightarrow\Delta_{\infty}\left(  f,g\right)  \leq\varepsilon$. \ So $B$
(or some other sliver containing $f$) would already have been eliminated in
the process of forming $S^{\prime}$.
\end{proof}

A natural question is whether Lemma \ref{l1lem}\ and Theorem \ref{l1winnowthm}%
\ would also hold with $L_{2}$-winnowability or $L_{\infty}$-winnowability in
place of $L_{1}$-winnowability. \ The next theorem\ shows, somewhat
surprisingly, that the use of the $L_{1}$\ norm is essential.

\begin{theorem}
\label{l2no}There exists a p-concept class $S$ that is coverable, but not
$L_{2}$-winnowable or $L_{\infty}$-winnowable.
\end{theorem}

\begin{proof}
We prove a stronger statement: there exists a \textit{finite} p-concept class
$S$, of size $\left\vert S\right\vert \leq2^{\operatorname*{poly}\left(
n\right)  }$, that is not $L_{2}$-winnowable (and as a direct consequence, not
$L_{\infty}$-winnowable either). \ To prove this, it suffices to find a set
$S$ with $\left\vert S\right\vert \leq2^{\operatorname*{poly}\left(  n\right)
}$, as well as a constant $\varepsilon>0$, for which the following holds.
\ For all $f\in S$,\ subsets $X\subseteq\left\{  0,1\right\}  ^{n}$ of size
less than $2^{n}-n^{2}$, and constants $\delta$\ depending on $\varepsilon$,
there exists a $g\in S$\ such that $\Delta_{2}\left(  f,g\right)  \left[
X\right]  \leq\delta$\ but $\Delta_{\infty}\left(  f,g\right)  >\varepsilon
$\ (at least, for all sufficiently large $n$).

Let $\varepsilon$ be any constant in $\left(  0,1\right)  $, and let $S$ be
the class of all functions $f:\left\{  0,1\right\}  ^{n}\rightarrow\left[
0,1\right]  $\ of the form%
\[
f\left(  x\right)  =\frac{a_{x}}{n},
\]
where the $a_{x}$'s are nonnegative integers satisfying%
\[
\sum_{x\in\left\{  0,1\right\}  ^{n}}a_{x}=n^{2}.
\]
Then clearly $\left\vert S\right\vert \leq\left(  2^{n}\right)  ^{n^{2}}$,
since we can form any $f\in S$ by starting from the identically-$0$ function,
then choosing $n^{2}$\ inputs $x$ (with repetition) on which to increment $f$
by $1/n$.

Now let $f\in S$, and let $X\subseteq\left\{  0,1\right\}  ^{n}$ have size
$\left\vert X\right\vert <2^{n}-n^{2}$. \ Then we can ``corrupt''\ $f$ to create
a new function $g\in S$ as follows. \ Let $Z$\ be a set of $n$ inputs
$x\in\left\{  0,1\right\}  ^{n}$\ on which $f\left(  x\right)  >0$\ (note that
such a $Z$ must exist, since $\sum_{x}f\left(  x\right)  =n$\ but $f\left(
x\right)  \leq1$\ for all $x$). \ By the pigeonhole principle, there exists a
$y\in\left\{  0,1\right\}  ^{n}\setminus X$\ such that $f\left(  y\right)
=0$. \ Fix that $y$, and define%
\[
g\left(  x\right)  :=\left\{
\begin{array}
[c]{cc}%
1 & \text{if }x=y\\
f\left(  x\right)  -1/n & \text{if }x\in Z\\
f\left(  x\right)  & \text{otherwise.}%
\end{array}
\right.
\]
Clearly $g\in S$ and%
\[
\Delta_{2}\left(  f,g\right)  \left[  X\right]  =\sqrt{\sum_{x\in Z\cap
X}\frac{1}{n^{2}}}\leq\frac{1}{\sqrt{n}}.
\]
On the other hand, we have $f\left(  y\right)  =0$\ and $g\left(  y\right)
=1$, so\ $\Delta_{\infty}\left(  f,g\right)  =1$. \ Therefore $S$\ is not
$L_{2}$-winnowable.
\end{proof}

\end{document}